%% file: paper-tr.tex
\documentclass[nocopyrightspace,preprint]{sigplanconf}
\usepackage{natbib}
\bibpunct();A{},
\let\cite=\citep
\usepackage{latexsym}
\usepackage{amssymb}
\usepackage{amsmath}
\usepackage{amsthm}
\usepackage{proof}
\usepackage{stmaryrd}
\usepackage{tikz}
\usepackage{paralist}
\input{macros}
\usepackage{version}
\includeversion{tr}
\excludeversion{sub}

\begin{document}

%\conferenceinfo{WXYZ '05}{date, City.} 
%\copyrightyear{2005} 
%\copyrightdata{supplied by printer} 

\title{Provenance Traces}
\subtitle{Extended Report}
\input{body.tex}

\end{document}

%% file: macros.tex
\newtheorem{theorem}{Theorem}
\newtheorem{proposition}{Proposition}
\newtheorem{lemma}{Lemma}
\newtheorem{corollary}{Corollary}

\newtheoremstyle{example}{\topsep}{\topsep}%
     {}%         Body font
     {}%         Indent amount (empty = no indent, \parindent = para indent)
     {\bfseries}% Thm head font
     {}%        Punctuation after thm head
     { }%     Space after thm head (\newline = linebreak)
     {\thmname{#1}\thmnumber{ #2}\thmnote{ (#3)}}%         Thm head spec

\theoremstyle{example}       
\newtheorem{example}{Example}
\newtheorem{remark}{Remark}
\newtheorem{definition}{Definition}

\newcommand{\flatten}{\textstyle \bigcup}
\newcommand{\mflatten}{\textstyle \bigsqcup}
\newcommand{\summ}{\textstyle \sum}

\newcommand{\SB}[1]{{\left\llbracket #1\right\rrbracket}}

\newcommand{\KSB}[2]{K\SB{#1}#2}
\newcommand{\WSB}[2]{W\SB{#1}#2}
\newcommand{\DSB}[2]{D\SB{#1}#2}
\newcommand{\getv}[1]{\lfloor #1 \rfloor}
\newcommand{\geta}[1]{\lceil #1 \rceil}
\newcommand{\kw}[1]{\mathsf{#1}}
\newcommand{\kcond}{\kw{cond}}
\newcommand{\kcomp}{\kw{comp}}
\newcommand{\kproj}{\kw{proj}}

\newcommand{\ksum}{\kw{sum}}

\newcommand{\ktrue}{\kw{t}}
\newcommand{\kfalse}{\kw{f}}
\newcommand{\kif}{\kw{if}}
\newcommand{\kthen}{\kw{then}}
\newcommand{\kelse}{\kw{else}}
\newcommand{\klet}{\kw{let}}
\newcommand{\kin}{\kw{in}}
\newcommand{\kempty}{\kw{empty}}
\newcommand{\kmap}{\kw{map}}

\newcommand{\fresh}{~\mathrm{fresh}}
\newcommand{\inputs}{\mathrm{in}^\star}
\newcommand{\outputs}{\mathrm{out}^\star}
\newcommand{\result}{\mathrm{out}}
\newcommand{\supp}{\mathrm{supp}}
\newcommand{\op}{\mathrm{op}}
\newcommand{\where}{\mathrm{where}}
\newcommand{\dep}{\mathrm{dep}}
\newcommand{\semiring}{\mathrm{semiring}}
\newcommand{\Loc}{\mathrm{Loc}}
\newcommand{\id}{\mathrm{id}}
\newcommand{\Wr}{\mathrm{Wr}}

\newcommand{\Int}{\mathbb{Z}}
\newcommand{\Bool}{\mathbb{B}}
\newcommand{\Nat}{\mathbb{N}}

\newcommand{\dom}{\mathrm{dom}}
\newcommand{\Pow}[1]{\mathcal{P}(#1)}
\newcommand{\Ts}{\Theta}
\newcommand{\Tr}{\mathrm{Tr}}
\newcommand{\Fid}{\mathrm{Fid}}

\newcommand{\Type}{\mathit{Type}}
\newcommand{\Exp}{\mathit{Exp}}
\newcommand{\Lab}{\mathit{Lab}}
\newcommand{\Var}{\mathit{Var}}
\newcommand{\Con}{\mathit{Con}}

\newcommand{\Term}{\mathit{Term}}
\newcommand{\Val}{\mathit{Val}}
\newcommand{\KVal}[1]{#1\text{-}\mathit{Val}}
\newcommand{\KK}{\mathcal{K}}

\newcommand{\boolTy}{\mathsf{bool}}
\newcommand{\intTy}{\mathsf{int}}
\newcommand{\setTy}[1]{\{#1\}}

\newcommand{\letin}[2]{\klet~#1~\kin~#2}
\newcommand{\setof}[1]{\{#1\}}
\newcommand{\ifthenelse}[3]{\kif~#1~\kthen~#2~\kelse~#3}
\newcommand{\andd}{\wedge}
\newcommand{\orr}{\vee}
\newcommand{\nott}{\neg}
\newcommand{\eq}{\approx}

\newcommand{\smerge}{\uplus}
\newcommand{\extend}{\propto}

\newcommand{\trlet}[2]{#1;#2}
\newcommand{\trcond}[4][]{\kcond_{#1}(#2,#3,#4)}

\newcommand{\trcomp}[2]{\kcomp(#1,#2)}
\newcommand{\trsum}[2]{\ksum(#1,#2)}
\newcommand{\trproj}[3]{\kproj_{#1}(#2,#3)}

\newcommand{\trconde}[6][]{\kcond_{#1}(#2,#3,#4)_{#5}^{#6}}
\newcommand{\trmape}[3]{\kmap(#1,#2)_{#3}}
\newcommand{\trcompe}[3]{\kcomp(#1,#2)_{#3}}
\newcommand{\trsume}[3]{\ksum(#1,#2)_{#3}}

\newcommand{\wf}[3]{#1 \vdash #2 : #3}

\newcommand{\wfstore}[2]{#2 : #1}
\newcommand{\wfstorectx}[4]{#1 \vdash #2,#3 : #4}
\newcommand{\wfcon}[3]{#1 \vdash_{\kw{con}} #2 : #3}
\newcommand{\wfctx}[3]{#1 \vdash #2 : #3}
\newcommand{\wftrace}[3]{#1 \vdash #2 : #3 }
\newcommand{\wftraces}[5]{#1 \vdash [#2]~#3~[#4] : #5 }
\newcommand{\wfterm}[3]{#1 \vdash_{\kw{term}} #2 : #3 }
\newcommand{\wftr}[4]{#1 \vdash #2 \triangleright #3 : #4}
\newcommand{\wftrs}[4]{#1 \vdash #2 \triangleright #3  \triangleright #4}
\newcommand{\matchavoids}[3]{#1 \matches #2 ~\#~ #3}

\newcommand{\matches}{\mathrel{{<}{:}}}
\newcommand{\annot}[2]{#1^{(#2)}}

\newcommand{\heval}[7][]{\annot{#2}{#3},#4 \gets #5\Downarrow_{#1} \annot{#6}{#7}}
\newcommand{\hevals}[9][]{\annot{#2}{#3},#4 \in #5,#6 \Downarrow^\star_{#1} \annot{#7}{#8},#9}
\newcommand{\eval}[4]{#1,#2 \Leftarrow #3\Downarrow #4}
\newcommand{\evals}[6]{#1,#2{\in} #3, #4\Downarrow^\star #5,#6}
\newcommand{\treval}[5]{#1,#2 \Leftarrow #3 \Downarrow #4,#5}
\newcommand{\trevals}[7]{#1,#2{\in} #3,#4 \Downarrow^\star #5,#6,#7}

\newcommand{\cp}[4]{#1,#2 \curvearrowright #3,#4}
\newcommand{\cps}[8]{#1,#2 {\in} #3,#4,#5 \curvearrowright^\star #6,#7,#8}

\newcommand{\ext}[4][w]{#2,#3 \rightsquigarrow_{#1} #4}
\newcommand{\exts}[4][w]{#2,#3 \rightsquigarrow^\star_{#1} #4}
\newcommand{\wext}[3]{\ext[W]{#1}{#2}{#3}}
\newcommand{\wexts}[3]{\exts[W]{#1}{#2}{#3}}
\newcommand{\dext}[3]{\ext[D]{#1}{#2}{#3}}
\newcommand{\dexts}[3]{\exts[D]{#1}{#2}{#3}}
\newcommand{\kext}[3]{\ext[K]{#1}{#2}{#3}}
\newcommand{\kexts}[3]{\exts[K]{#1}{#2}{#3}}

\newcommand{\bslices}[4]{#1,#2 \swarrow^\star #3,#4}
\newcommand{\bslice}[4]{#1,#2 \swarrow #3,#4}
\newcommand{\bslicee}[2]{#1,#2 \swarrow \epsilon}

\newcommand{\fslices}[4]{#1,#2 \searrow^\star #3,#4}
\newcommand{\fslice}[4]{#1,#2 \searrow #3,#4}
\newcommand{\fslicee}[2]{#1,#2 \searrow \epsilon}

\newcommand{\trsat}[2]{#1\models #2}
\newcommand{\trsats}[2]{#1 \models^\star #2}

\newcommand{\labelSec}[1]{\label{sec:#1}}
\newcommand{\refSec}[1]{Section~\ref{sec:#1}}
\newcommand{\refSecs}[2]{Sections~\ref{sec:#1}--\ref{sec:#2}}
\newcommand{\labelFig}[1]{\label{fig:#1}}
\newcommand{\refFig}[1]{Figure~\ref{fig:#1}}

\newcommand{\labelThm}[1]{\label{thm:#1}}
\newcommand{\refThm}[1]{Theorem~\ref{thm:#1}}
\newcommand{\labelLem}[1]{\label{lem:#1}}
\newcommand{\refLem}[1]{Lemma~\ref{lem:#1}}

\newcommand{\labelCor}[1]{\label{cor:#1}}

\newcommand{\labelEx}[1]{\label{ex:#1}}
\newcommand{\refEx}[1]{Example~\ref{ex:#1}}
\newcommand{\refExs}[2]{Examples~\ref{ex:#1}--\ref{ex:#2}}

\newcommand{\todo}[1]{\textbf{[TODO: #1]}}

%% file: body.tex
\authorinfo{James Cheney}
           {University of Edinburgh}
           {jcheney@inf.ed.ac.uk}
\authorinfo{Umut A. Acar \and Amal Ahmed}
           {Toyota Technological Institute, Chicago}
           {[umut$|$amal]@tti-c.org}

\maketitle

\input{abstract}

\section{Introduction}\labelSec{intro}

% \begin{itemize}
% \item Motivation for provenance in DB/scientific data settings
% \item Central question: What is provenance?
% \item Our approach: Traces as semantic foundation
% \item Forward vs. backwards ``trace slicing'' 
% \item Other forms of provenance can be ``extracted''
% \end{itemize}

Sophisticated computer systems and programming techniques,
particularly database management systems and distributed computation,
are now being used for large-scale scientific endeavors in many fields
including biology, physics and astronomy.  Moreover, they are used
directly by scientists who --- often justifiably --- view the behavior
of such systems is opaque and unreliable. Simply presenting the result
of a computation is not considered sufficient to establish its
repeatability or scientific value in (for example) a journal article.
Instead, it is considered essential to provide high-level explanations
of how a part of the result of a database query or distributed
computation was derived from its inputs, or how a database came to be
the way it is.  Such information about the source, context,
derivation, or history of a (data) object is often called
\emph{provenance}.

Currently, many systems either require their users to deal with
provenance manually or provide one of a variety of ad hoc, custom
solutions.  Manual recordkeeping is tedious and error-prone, while
both manual and custom solutions are expensive and provide few formal
correctness guarantees.  This state of affairs strongly motivates
research into automatic and standardized techniques for recording,
managing, and exploiting provenance in databases and other systems.

A number of approaches to automatic provenance tracking have been
studied, each aiming to capture some intuitive aspect of provenance
such as ``Where did a result come from in the input?''
\cite{buneman01icdt}, ``What inputs influenced a
result?''\cite{DBLP:journals/tods/CuiWW00,buneman01icdt}, ``How was a
result produced from the input?''
\cite{DBLP:conf/pods/2007/GreenKT07}, or ``What inputs do results
depend on?''  \cite{DBLP:conf/dbpl/CheneyAA07}.
\begin{comment}
  \begin{compactitem}
  \item ``Where did a result come from in the input?''
    \cite{buneman01icdt},
  \item ``What inputs influenced a
    result?''\cite{DBLP:journals/tods/CuiWW00,buneman01icdt},
  \item ``How was a result produced from the input?''
    \cite{DBLP:conf/pods/2007/GreenKT07}, or
  \item ``What inputs do results depend on?''
    \cite{DBLP:conf/dbpl/CheneyAA07}
  \end{compactitem}
\end{comment}
However, there is not yet much understanding of the advantages,
disadvantages and formal guarantees offered by each, or of the
relationships among them.  Many of these techniques have
been presented as ad hoc definitions without clear formal
specifications of the problem the definitions are meant to solve.  In
some cases, loose specifications have been developed, but they appear
difficult to extend beyond simple settings such as monotone relational
queries.  

Therefore, we believe that semantic foundations for provenance need to
be developed in order to understand and relate existing techniques, as
well as to motivate and validate new techniques.  We focus on
provenance in database management systems, because of its practical
importance and because several interesting provenance techniques have
already been developed in this setting.  We investigate a semantic
foundation for provenance in databases based on \emph{traces}. We
begin with an operational semantics based on stores in which each part
of each value has a label.  We instrument the semantics so that as an
expression evaluates, we record certain properties of the operational
derivation in a {\em provenance trace}.  Provenance traces record the
relationships between the labels in the store, ultimately linking the
result of a computation to the input.  Traces can be viewed as a
concrete representation of the operational semantics derivation
showing how each part of the output was computed from the input and
intermediate values.

We employ the \emph{nested
  relational calculus} (NRC), a core database query language closely
related to monadic comprehensions as used in Haskell and other
functional programming languages~\cite{wadler92mscs}.  The nested
relational model also forms the basis for distributed programming
systems such as MapReduce~\cite{mapreduce} and
PigLatin~\cite{piglatin} and is closely related to XML.  Thus, our
results should generalize to these other settings.

This paper makes the following contributions:
\begin{compactitem}
\item We define traces, traced evaluation for NRC queries, and a trace
  adaptation semantics.  
\item We show that we can extract several other forms of provenance
  that have been developed for the NRC from traces, including
  \emph{where-provenance}~\cite{buneman01icdt,buneman07icdt},
  \emph{dependency provenance}~\cite{DBLP:conf/dbpl/CheneyAA07}, and
  \emph{semiring-provenance}~\cite{DBLP:conf/pods/2007/GreenKT07,DBLP:conf/pods/FosterGT08}.
  The semiring-provenance model already generalizes several other
  forms of provenance such as \emph{why-provenance}~\cite{buneman01icdt} and
  \emph{lineage}~\cite{DBLP:journals/tods/CuiWW00}, but where-provenance and
  dependency-provenance are not instances of the semiring model.
  Provenance traces thus unify three previously unrelated provenance
  models.
\item We state and prove properties which establish traces as a solid
  semantic foundation for provenance.  Specifically, we show that the
  trace generated by evaluating an expression is consistent with the
  resulting store, and that such traces are ``explanations'' that help
  us understand how the expression would behave if the input store is
  changed.  This is the main contribution of the paper, and in
  particular the explanation property is a key ``correctness''
  property for provenance that has been absent from previous work on
  this topic.
\end{compactitem}

We want to emphasize that \emph{provenance traces are not a proposal
  for a concrete, practical form of provenance}.  Traces are a
candidate answer to the question ``what is the most detailed form of
provenance we could imagine recording?''  We expect that it is
unlikely that provenance traces would be implementable within a
large-scale database system.  Other practical provenance techniques
will necessarily sacrifice or approximate some of the detail of
provenance traces in return for efficiency.  Thus, the role of
provenance traces is to provide a way to explain precisely what is
lost in the process.

The traces used in this paper are also related to traces studied in
other settings, particularly in AFL, an adaptive functional language
introduced by \citet{acar06toplas}.  However, there are important
differences.  First, while AFL leaves it up to the programmer to
identify \emph{modifiable} inputs and \emph{changeable} outputs,
provenance traces implicitly treat every part of the input as
modifiable and every part of the output as changeable.  This may make
provenance traces too inefficient for practical use, but our main goal
here is to identify a rich, principled form of provenance and
efficiency is a secondary concern.  Second, AFL traces are based
directly on source language expressions, and were not designed with
human-readability or provenance extraction in mind.  In contrast,
provenance traces can be viewed as directed acyclic graphs (with some
extra structure and annotations) that can easily be traversed to
extract other forms of provenance.  Finally, AFL includes
user-defined, recursive functions, whereas the NRC does not include
function definitions but does provide collection types and
comprehension operations.  These differences are minor; it appears
straightforward to add the missing features to the respective
languages.

\paragraph{An example}

As a simple example, consider an expression
$\ifthenelse{x=5}{y+42}{x}$.  If we run this on an input store $x =
5^{l_x}, y = 42^{l_y}$ then the result is $47^{l'}$, and the trace is 
\begin{verbatim}
  l_1' <- l_x = 5;
  cond(l_1', t, l' <- l_y+42)
\end{verbatim}
This trace records that we first test whether $l_x = 5$, then do a
conditional branch.  The $\kcond$ trace records the tested label
$l_1'$, its value, and a subtrace showing how we computed the final
result $l'$ by copying from $l_y$.

As a more complicated example illustrating traces for relational
operations, consider a SQL-style query that selects only the
$B$-values of records in table $R$:
\begin{verbatim}
SELECT B FROM R
\end{verbatim}
which corresponds to the NRC expression $\{\pi_B(x) \mid x \in R\}$.
When run on $R = \{(A:1,B:2),(A:2,B:3)\}$ the result is $\{2,3\}$.  If
we regard the input as labeled as follows:
$\{(A:1^{l_{11}},B:2^{l_{12}})^{l_1},(A:2^{l_{21}},B:3^{l_{22}})^{l_2}\}^l$
then the resulting trace is
\begin{verbatim}
  l' <- comp(l,{[l_1] l_1' <- proj_B (l_1,l_12),
                [l_2] l_2' <- proj_B (l_2,l_22)})
\end{verbatim}
producing labeled output $\{2^{l_1'},3^{l_2'}\}^{l'}$.  This trace
shows that the result is obtained by comprehension over $l$.  There
are two elements, $l_1$ and $l_2$, yielding results $l_1' = l_{12}$
and $l_2' = l_{22}$, which were obtained by projecting the $B$ field
from $l_1$ and $l_2$ respectively.

It should be clear that traces can in general be large and difficult
to interpret because they are very low-level. As mentioned above, we
can \emph{slice} traces by discarding irrelevant information to obtain
smaller traces that are more useful as explanations of how a specific
part of the output was produced or how a part of the input was used.
As a simple example, if we are only interested in $l_1'$ in the output
of the second example, we can slice the trace ``backwards'' from
$l_1'$ to obtain
\begin{verbatim}
  l' <- comp(l,{[l_1] l_1' <- proj_B (l_1,l_12)},
            x. \pi_B(x))
\end{verbatim}
Dually, if we wish to see how some part of the input influences parts
of the output, we can slice ``forwards''.  For example, the forward
slice from $l_{21}$ is empty, meaning that it did not play any role in
the execution, whereas a forward slice from $l_{22}$ is
\begin{verbatim}
  l' <- comp(l,{[l_2] l_2' <- proj_B (l_2,l_22)},
            x. \pi_B(x))
\end{verbatim}

We can also extract other forms of provenance directly from traces.
For example, in the second query above, we can see that $l_2'$ in the
output ``comes from'' $l_{12}$ in the input since it is copied by the
projection operation $l_1' \gets \trproj{B}{l_1}{l_{12}}$.  Similarly,
if we inspect the forward trace slice from $l_{22}$, we can see that
the labels $l_2'$ and $l'$ in the output mat ``depend on'' $l_{22}$,
and that the edge $(l',l_2')$ is ``produced'' by the comprehension
from the edge $(l,l_2)$.

\paragraph{Synopsis}
The structure of the rest of this paper is as follows.  \refSec{nrc}
reviews the nested relational calculus, and introduces an operational,
destination-passing, store-based semantics for NRC.  \refSec{traces}
defines provenance traces and introduces a traced operational
semantics for NRC queries and a trace adaptation semantics for
adjusting traces to changes to the input.  \refSec{metatheory}
establishes the key metatheoretic and semantic properties of traces.
\refSec{extraction} discusses extracting other forms of provenance
from traces, and \refSec{slicing} briefly discusses trace slicing and
simplification techniques.  We discuss related and future work and
conclude in \refSecs{related}{concl}.  
\begin{sub}
  A companion technical report~\cite{tr} provides additional
  discussion and proofs.
\end{sub}
\section{Nested relational calculus}\labelSec{nrc}

The nested relational calculus~\cite{buneman95tcs}, or NRC, is a
simply-typed core language, closely related to monadic
comprehensions~\cite{wadler92mscs}.  The NRC that is as expressive as
standard database query languages such as SQL but has simpler syntax
and cleaner semantics.  (We do not address certain dark corners of SQL
such as NULL values.)  The syntax of NRC types $\tau \in \Type$ is as
follows:
\begin{eqnarray*}
  \tau &::=& \intTy \mid \boolTy \mid \tau_1 \times \tau_2 \mid \{\tau\}
\end{eqnarray*}
Types include base types such as $\intTy$ and $\boolTy$, pairing types
$\tau_1\times\tau_2$, and collection types $\{\tau\}$.  Collection
types $\{\tau\}$ are often taken to be sets, bags (multisets), or
lists; in this paper, we consider multiset collections only.  We omit
first-class function types and $\lambda$-terms because most database
systems do not support them.

We assume countably infinite, disjoint sets $\Var$ of \emph{variables}
and \emph{labels} $\Lab$.  The syntax of NRC expressions $e \in \Exp$
is as follows:
\begin{eqnarray*}
  e &::=& l \mid x \mid \letin{x=e_1}{e_2} \mid (e_1,e_2) \mid \pi_i(e) \\
  &\mid& b \mid \nott e \mid e_1 \andd e_2 
  \mid \ifthenelse{e_0}{e_1}{e_2}\\
  &\mid& \emptyset \mid \setof{e} \mid e_1 \cup e_2   
  \mid  \flatten \{e_2 \mid x \in e_1\} \mid \kempty(e) \\
  &\mid& i \mid e_1 + e_2 \mid  e_1 \eq e_2\mid \summ \{e_2 \mid x \in e_1\} 
\end{eqnarray*}
Variables and $\klet$-expressions, pairing, boolean, and integer
operations are standard.  Labels are used in the operational semantics
(\refSec{opsem}).  The expression $\emptyset$ denotes the empty
collection; $\{e\}$ constructs a singleton collection, $e_1 \cup e_2$
takes the (multiset) union of two collections, and $\flatten \{e \mid
x \in e_0\}$ iterates over a collection obtained by evaluating $e$,
applying $e(x)$ to each element of the collection, and unioning the
results.  Note that we can define $\{e \mid x \in e_0\}$ as
$\flatten\{\{e\} \mid x \in e_0\}$.  We include integer constants,
addition ($e_1+e_2$), and equality ($e_1 \eq e_2$). Finally, the
$\kempty(e)$ predicate tests whether the collection denoted by $e$ is
empty, and the $\summ\{e \mid x \in e_0\}$ operation takes the sum of
a collection of integers.

Expressions are identified modulo alpha-equivalence, regarding $x$
bound in $e(x)$ in the expressions $\flatten\{e(x) \mid x \in e_0\}$,
$\summ\{e(x) \mid x \in e_0\}$ and $\letin{x=e_0}{e(x)}$.  We write
$e[l/x]$ for the result of substituting a label $l$ for a variable $x$
in $e$; labels cannot be bound so substitution is naturally
capture-avoiding.

\subsection{Examples}

As with many core languages, it is inconvenient to program directly in
NRC.  Instead, it is often more convenient to use idiomatic
``comprehension syntax'' similar to Haskell's list
comprehensions~\cite{wadler92mscs,DBLP:journals/sigmod/BunemanLSTW94}.
These can be viewed as syntactic sugar for primitive NRC expressions,
just as in Haskell list comprehensions can be translated to the
primitive monadic operations on lists.  Although we use
unlabeled pairs, the NRC can also be extended easily with convenient
named-record syntax.  These techniques are standard so here we only
illustrate them via examples which will be used later in the paper.

\begin{example}\labelEx{eg1}
  Suppose we have relations
  $R:\setTy{(A{:}\intTy,B{:}\intTy,C{:}\intTy)}$,
  $S:\setTy{(C{:}\intTy,D{:}\intTy)}$.  Consider the SQL
  ``join'' query
\begin{verbatim}
  SELECT R.A,R.B,S.D FROM R,S WHERE R.C = S.C
\end{verbatim}
  This is equivalent to the core NRC expression
\[\small
\begin{array}{rcl}
Q_1 &=&  \flatten\{\flatten\{\ifthenelse{r.C=s.C\\
&&\qquad\qquad}{\{(A{:}r.A,B{:}r.B,D{:}s.D)\}}{\emptyset} \\
&&\quad \mid s \in S\} \mid r \in R\}
\end{array}
\]
\end{example}

\begin{example}\labelEx{eg2}
  Given $R,S$ as above, the SQL ``aggregation'' query
\begin{verbatim}
  SELECT 42 AS C, SUM(D) FROM S WHERE C = 2
  UNION
  SELECT B AS C, A AS D FROM R WHERE C = 4
\end{verbatim}
can be expressed as
  \[\small\begin{array}{rcl}
    Q_2 &=&  
\{(C:42,D:\summ\{\ifthenelse{s.C = 2}{s.D}{0} \mid s \in S\})\}\\
    &\cup & \flatten\{\ifthenelse{r.C = 4}{\{(C{:}r.B,D{:}r.A)\}}{\emptyset} \mid r \in R\}
\end{array}
  \]
\end{example}

Some sample input tables and the results of running $Q_1$ and $Q_2$ on
them are shown in \refFig{examples}.  The labels $r,r_1,\ldots$ in are
used in the operational semantics, as discussed in \refSec{opsem}.

  \begin{figure}

\begin{tikzpicture}
  \node at (0,0) {A};
  \node at (1,0) {B};
  \node at (2,0) {C};
  
  \node (atom11) at (0,-0.5) [inner sep = 0pt] {$1$};
  \node (atom12) at (1,-0.5) [inner sep = 0pt] {$2$};
  \node (atom13) at (2,-0.5) [inner sep = 0pt] {$3$};
  \node (atom21) at (0,-1) [inner sep = 0pt] {$1$};
  \node (atom22) at (1,-1) [inner sep = 0pt] {$3$};
  \node (atom22) at (2,-1) [inner sep = 0pt] {$3$};
  \node (atom31) at (0,-1.5) [inner sep = 0pt] {$7$};
  \node (atom32) at (1,-1.5) [inner sep = 0pt] {$42$};
  \node (atom33) at (2,-1.5) [inner sep = 0pt] {$4$};

  \node (table) at (-0.3,0.3) [anchor=north west,rectangle,draw,minimum width = 3cm, minimum height
  = 2.1cm] {};
  
  \node (caption) at ([yshift=-0.4cm]table.south) {Input table $R(A,B,C)$};

  \node (row1) at (-0.2,-0.3) [help lines,anchor=north west,rectangle,draw,minimum width = 2.8cm, minimum height
  = 0.4cm] {};
  \node (row2) at (-0.2,-0.8) [help lines,anchor=north west,rectangle,draw,minimum width = 2.8cm, minimum height
  = 0.4cm] {};
  \node (row3) at (-0.2,-1.3) [help lines,anchor=north west,rectangle,draw,minimum width = 2.8cm, minimum height
  = 0.4cm] {};

  \node (ctab) at ([yshift=+5mm]table.north) {\small $r$};
  \node (crow1) at ([xshift=-6mm]row1.west) {\small $r_1$};
  \node (crow2) at ([xshift=-6mm]row2.west) {\small $r_2$};
  \node (crow3) at ([xshift=-6mm]row3.west) {\small $r_3$};
  \node (catom11) at (0.4,-0.45) {\small $r_{11}$};
  \node (catom12) at (1.4,-0.45) {\small $r_{12}$};
  \node (catom13) at (2.4,-0.45) {\small $r_{13}$};
  \node (catom21) at (0.4,-0.95) {\small $r_{21}$};
  \node (catom22) at (1.4,-0.95) {\small $r_{22}$};
  \node (catom22) at (2.4,-0.95) {\small $r_{23}$};
  \node (catom31) at (0.4,-1.45) {\small $r_{31}$};
  \node (catom32) at (1.4,-1.45) {\small $r_{32}$};
  \node (catom33) at (2.4,-1.45) {\small $r_{33}$};

  \draw (table.north) -- (ctab.south);
  \draw (row1.west) -- (crow1);
  \draw (row2.west) -- (crow2);
  \draw (row3.west) -- (crow3);

  \node at (4,0) {C};
  \node at (5,0) {D};
  
  \node (atom1) at (4,-0.5) [inner sep = 0pt] {$2$};
  \node (atom2) at (5,-0.5) [inner sep = 0pt] {$3$};
  \node (atom3) at (4,-1) [inner sep = 0pt] {$2$};
  \node (atom4) at (5,-1) [inner sep = 0pt] {$4$};
  \node (atom5) at (4,-1.5) [inner sep = 0pt] {$3$};
  \node (atom6) at (5,-1.5) [inner sep = 0pt] {$7$};

  \node (table) at (3.7,0.3) [anchor=north west,rectangle,draw,minimum width = 2cm, minimum height
  = 2.1cm] {};
  
  \node (caption) at ([yshift=-0.4cm]table.south) {Input table $S(C,D)$};

  \node (row1) at (3.8,-0.3) [help lines,anchor=north west,rectangle,draw,minimum width = 1.8cm, minimum height
  = 0.4cm] {};
  \node (row2) at (3.8,-0.8) [help lines,anchor=north west,rectangle,draw,minimum width = 1.8cm, minimum height
  = 0.4cm] {};
  \node (row3) at (3.8,-1.3) [help lines,anchor=north west,rectangle,draw,minimum width = 1.8cm, minimum height
  = 0.4cm] {};

  \node (ctab) at ([yshift=+5mm]table.north) {\small $s$};
  \node (crow1) at ([xshift=-6mm]row1.west) {\small $s_1$};
  \node (crow2) at ([xshift=-6mm]row2.west) {\small $s_2$};
  \node (crow3) at ([xshift=-6mm]row3.west) {\small $s_3$};
  \node (catom1) at (4.4,-0.45) {\small $s_{11}$};
  \node (catom2) at (5.4,-0.45) {\small $s_{12}$};
  \node (catom3) at (4.4,-0.95) {\small $s_{21}$};
  \node (catom4) at (5.4,-0.95) {\small $s_{22}$};
  \node (catom5) at (4.4,-1.45) {\small $s_{31}$};
  \node (catom6) at (5.4,-1.45) {\small $s_{23}$};

  \draw (table.north) -- (ctab.south);
  \draw (row1.west) -- (crow1);
  \draw (row2.west) -- (crow2);
  \draw (row3.west) -- (crow3);
\end{tikzpicture}%

\begin{tikzpicture}
  \node at (0,0) {A};
  \node at (1,0) {B};
  \node at (2,0) {D};
  
  \node (atom11) at (0,-0.5) [inner sep = 0pt] {$1$};
  \node (atom12) at (1,-0.5) [inner sep = 0pt] {$2$};
  \node (atom13) at (2,-0.5) [inner sep = 0pt] {$7$};
  \node (atom21) at (0,-1) [inner sep = 0pt] {$1$};
  \node (atom22) at (1,-1) [inner sep = 0pt] {$3$};
  \node (atom23) at (2,-1) [inner sep = 0pt] {$7$};

  \node (table) at (-0.3,0.3) [anchor=north west,rectangle,draw,minimum width = 3cm, minimum height
  = 1.6cm] {};
  
  \node (caption) at ([yshift=-0.4cm]table.south) {Output table $Q_1(A,B,D)$};

  \node (row1) at (-0.2,-0.3) [help lines,anchor=north west,rectangle,draw,minimum width = 2.8cm, minimum height
  = 0.4cm] {};
  \node (row2) at (-0.2,-0.8) [help lines,anchor=north west,rectangle,draw,minimum width = 2.8cm, minimum height
  = 0.4cm] {};

  \node (ctab) at ([yshift=+5mm]table.north) {\small $l$};
  \node (crow1) at ([xshift=-6mm]row1.west) {\small $l_1$};
  \node (crow2) at ([xshift=-6mm]row2.west) {\small $l_2$};
  \node (catom11) at (0.4,-0.45) {\small $l_{11}$};
  \node (catom12) at (1.4,-0.45) {\small $l_{12}$};
  \node (catom13) at (2.4,-0.45) {\small $1_{13}$};
  \node (catom21) at (0.4,-0.95) {\small $l_{21}$};
  \node (catom22) at (1.4,-0.95) {\small $l_{22}$};
  \node (catom23) at (2.4,-0.95) {\small $l_{23}$};

  \draw (table.north) -- (ctab.south);
  \draw (row1.west) -- (crow1);
  \draw (row2.west) -- (crow2);

  \node at (4,0) {C};
  \node at (5,0) {D};
  
  \node (atom1) at (4,-0.5) [inner sep = 0pt] {$42$};
  \node (atom2) at (5,-0.5) [inner sep = 0pt] {$7$};
  \node (atom3) at (4,-1) [inner sep = 0pt] {$42$};
  \node (atom4) at (5,-1) [inner sep = 0pt] {$7$};

  \node (table) at (3.7,0.3) [anchor=north west,rectangle,draw,minimum width = 2cm, minimum height
  = 1.6cm] {};
  
  \node (caption) at ([yshift=-0.4cm]table.south) {Output table $Q_2(C,D)$};

  \node (row1) at (3.8,-0.3) [help lines,anchor=north west,rectangle,draw,minimum width = 1.8cm, minimum height
  = 0.4cm] {};
  \node (row2) at (3.8,-0.8) [help lines,anchor=north west,rectangle,draw,minimum width = 1.8cm, minimum height
  = 0.4cm] {};

  \node (ctab) at ([yshift=+5mm]table.north) {\small $l'$};
  \node (crow1) at ([xshift=-6mm]row1.west) {\small $l'_1$};
  \node (crow2) at ([xshift=-6mm]row2.west) {\small $l'_2$};
  \node (catom1) at (4.4,-0.45) {\small $l'_{11}$};
  \node (catom2) at (5.4,-0.45) {\small $l'_{12}$};
  \node (catom3) at (4.4,-0.95) {\small $l'_{21}$};
  \node (catom4) at (5.4,-0.95) {\small $l'_{22}$};

  \draw (table.north) -- (ctab.south);
  \draw (row1.west) -- (crow1);
  \draw (row2.west) -- (crow2);
\end{tikzpicture}%
    \caption{Examples}
    \label{fig:examples}
\end{figure}
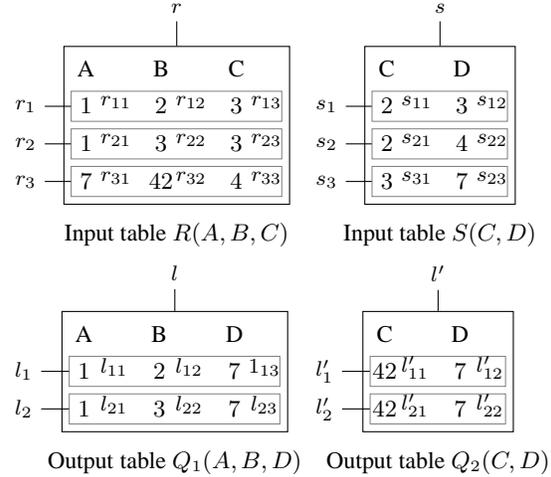

\begin{tr}
  \subsection{Type system}

\begin{figure}[h]
  
  \[
  \begin{array}{c}
    \infer{\wf{\Gamma}{x}{\tau}}{x:\tau \in \Gamma}
    \quad
    \infer{\wf{\Gamma}{\letin{x=e_1}{e_2}}{\tau_2}}{\wf{\Gamma}{e_1}{\tau_1} & \wf{\Gamma,x{:}\tau}{e_2}{\tau_2}}
    \smallskip\\
    \infer{\wf{\Gamma}{i}{\intTy}}{i \in \Int}
    \quad
    \infer{\wf{\Gamma}{e_1 + e_2}{\intTy}}{\wf{\Gamma}{e_1}{\intTy} & \wf{\Gamma}{e_2}{\intTy}}
    \smallskip\\
    \infer{\wf{\Gamma}{b}{\boolTy}}{b \in \Bool}
    \quad
    \infer{\wf{\Gamma}{\nott e}{\boolTy}}{\wf{\Gamma}{e}{\boolTy}}
    \quad
    \infer{\wf{\Gamma}{e_1 \andd e_2}{\boolTy}}{\wf{\Gamma}{e_1}{\boolTy} & \wf{\Gamma}{e_2}{\boolTy}}
    \smallskip\\
    \infer{\wf{\Gamma}{e_1 \eq e_2}{\boolTy}}{\wf{\Gamma}{e_1}{\intTy} & \wf{\Gamma}{e_2}{\intTy}}
    \quad
    \infer{\wf{\Gamma}{\ifthenelse{e}{e_1}{e_2}}{\tau}}
    {\wf{\Gamma}{e}{\boolTy}
      &
      \wf{\Gamma}{e_1}{\tau}
      &
      \wf{\Gamma}{e_2}{\tau}
    }      
    \smallskip\\
    \infer{\wf{\Gamma}{\kempty(e)}{\boolTy}}{\wf{\Gamma}{e}{\{\tau\}}}
    \quad
    \infer{\wf{\Gamma}{(e_1,e_2)}{\tau_1 \times \tau_2}}{\wf{\Gamma}{e_1}{\tau_1} & \wf{\Gamma}{e_2}{\tau_2}}
    \quad
    \infer{\wf{\Gamma}{\pi_i (e)}{\tau_i}}{\wf{\Gamma}{e}{\tau_1 \times \tau_2}}
    \smallskip\\
    \infer{\wf{\Gamma}{\emptyset}{\setTy{\tau}}}{}
    \quad
    \infer{\wf{\Gamma}{\{e\}}{\setTy{\tau}}}{\wf{\Gamma}{e}{\tau}}
    \quad
    \infer{\wf{\Gamma}{e_1 \cup e_2}{\setTy{\tau}}}{\wf{\Gamma}{e_1}{\setTy{\tau}} & \wf{\Gamma}{e_2}{\setTy{\tau}}}
    \smallskip\\
    \infer{\wf{\Gamma}{\flatten\{e\mid x \in e_0\}}{\setTy{\tau}}}{\wf{\Gamma}{e_0}{\setTy{\tau_0}} & \wf{\Gamma,x{:}\tau_0}{e}{\{\tau\}}}
\quad
    \infer{\wf{\Gamma}{\summ\{e\mid x \in e_0\}}{\intTy}}{\wf{\Gamma}{e_0}{\setTy{\tau_0}} & \wf{\Gamma,x{:}\tau_0}{e}{\intTy}}
  \end{array}
  \]
  \caption{Expression well-formedness}
  \label{fig:wf-exp}
\end{figure}

NRC expressions can be typechecked using standard techniques.  The
typechecking rules are shown in \refFig{wf-exp}.  We employ contexts
$\Gamma$ of the form $ \Gamma ::= \cdot \mid \Gamma,x{:}\tau$.

\subsection{Denotational semantics}
The semantics of NRC expressions is usually defined denotationally.
We consider values $v \in \Val$ of the form:
\[v ::= i \mid b \mid (v_1,v_2) \mid \{v_1,\ldots,v_n\}\]
where $i \in \Int$ and $b \in \Bool$, and interpret types as sets of
values, as follows:
\begin{eqnarray*}
  \SB{\intTy} &=& \Int = \{\ldots,-1,0,1,2,\ldots\}\\
  \SB{\boolTy} &=& \Bool = \{\ktrue,\kfalse\} \\
  \SB{\tau_1 \times \tau_2} &=& \SB{\tau_1} \times \SB{\tau_2}\\
  \SB{\setTy{\tau}} &=& \mathcal{M}_{\mathsf{fin}}(\SB{\tau})
\end{eqnarray*}
We write $\mathcal{M}_{\mathsf{fin}}(X)$ for the set of \emph{finite
  multisets} of values.  \refFig{nrc-denotational} shows the
(standard) equations defining the denotational semantics of NRC
expressions.  NRC does not include arbitrary recursive definitions, so
we do not need to deal with nontermination.

We write $\gamma: \Var \to \Val$ for a finite function (or
environment) mapping variables $x$ to values $v$.  We write
$\SB{\Gamma}$ for the set of all environments $\gamma$ such that
$\gamma(x) \in \SB{\Gamma(x)}$ for all $x \in \dom(\gamma)$.

The type system given above is sound in the following sense:
\begin{proposition}
  If $\wf{\Gamma}{e}{\tau}$ then $\SB{e} : \SB{\Gamma} \to \SB{\tau}$.
\end{proposition}
\begin{figure}
  
  \begin{eqnarray*}
    \SB{x}\gamma &=& \gamma(x)\\
    \SB{\letin{x=e_1}{e_2}}\gamma&=& \SB{e_2}\gamma[x \mapsto \SB{e_1}\gamma]\\
    \SB{i}\gamma &=& i\\
    \SB{e_1 + e_2}\gamma &=& \SB{e_1}\gamma + \SB{e_2}\gamma\\
    \SB{\summ \{e \mid x \in e_0\}}\gamma &=& \summ \{\SB{e}\gamma[x\mapsto v] \mid v \in \SB{e_0}\gamma\}\\
    \SB{b}\gamma &=& b\\
    \SB{\nott e}\gamma &=& \nott\SB{e}\gamma\\
    \SB{e_1 \andd e_2}\gamma &=& \SB{e_1}\gamma \andd \SB{e_2}\gamma\\
    \SB{(e_1,e_2)}\gamma &=& (\SB{e_1}\gamma ,\SB{e_2}\gamma )\\
    \SB{\pi_i(e)}\gamma &=& \pi_i(\SB{e}\gamma)\\
    \SB{\emptyset}\gamma &=& \emptyset\\
    \SB{\setof{e}}\gamma &=& \{\SB{e}\gamma\}\\
    \SB{e_1 \cup e_2}\gamma &=& \SB{e_1}\gamma \sqcup \SB{e_2}\gamma\\
    \SB{\flatten\{e \mid x \in e_0\}}\gamma &=& \mflatten\{\SB{e}\gamma[x\mapsto v] \mid v \in \SB{e_0}\gamma\}\\
    \SB{\ifthenelse{e_0}{e_1}{e_2}}\gamma &=& \left\{ 
      \begin{array}{ll}
        \SB{e_1}\gamma & \text{if $\SB{e_0}{\gamma} = \ktrue$}\\
        \SB{e_2}\gamma & \text{if $\SB{e_0}{\gamma} = \kfalse$}
      \end{array}\right.\\
    \SB{e_1 \eq e_2}\gamma &=& \left\{ 
      \begin{array}{ll}
        \ktrue & \text{if $\SB{e_1}{\gamma} = \SB{e_2}{\gamma}$}\\
        \kfalse& \text{if $\SB{e_1}{\gamma} \neq \SB{e_2}{\gamma}$}
      \end{array}\right.
    \\
    \SB{\kempty(e)}\gamma &=& \left\{ 
      \begin{array}{ll}
        \ktrue & \text{if $\SB{e}{\gamma} = \emptyset$}\\
        \kfalse& \text{if $\SB{e}{\gamma} \neq \emptyset$}
      \end{array}\right.
  \end{eqnarray*}
  
  \caption{Denotational semantics of NRC}
  \label{fig:nrc-denotational}
\end{figure}
\end{tr}

\subsection{Operational semantics}\labelSec{opsem}

The semantics of NRC is usually presented denotationally.  For the
purposes of this paper, we will introduce an operational semantics
based on \emph{stores} in which every part of every value has a label.
This semantics will serve as the basis for our trace semantics, since
labels can easily be used to address parts of the input, output, and
intermediate values of a query.  Thus, labels play a dual role as
\emph{addresses} of values in the store and as ``locations'' mentioned
in traces.  Note that NRC is a purely functional language and so
labels are written at most once.

In order to ensure that each part of each value has a label, we employ
a store mapping labels to \emph{value constructors}, which can be
thought of as individual heap cells each describing one part of a
value.  We define value constructors $k \in \Con$ as follows:
\[k ::= i \mid b \mid (l_1,l_2) \mid \{l_1:m_1,\ldots,l_n:m_n\}\]
Here, $\{l_1:m_1,\ldots,l_n:m_n\}$ denotes a multiset of labels (often
denoted $L,L'$), where $m_i$ is the multiplicity of $l_i$.
Multiplicities are assumed nonzero and omitted when equal to 1.
Multisets are equivalent up to reordering and we assume the elements
$l_i$ are distinct.  We write $M \sqcup N$ for multiset union and $M
\oplus N$ for domain-disjoint multiset union, defined only when
$\dom(M) \cap \dom(N) = \emptyset$.

We write $\Lab(k)$ for the set of labels mentioned in $k$.  Stores are
finite maps $\sigma:\Lab \to \Con$ from labels to constructors.  We
also consider label environments to be finite maps from variables to
labels $\gamma : \Var \to \Lab$.

  \begin{figure}
    
    \[\small\begin{array}{rcl}\small
      \op(l,\sigma) &=& \sigma(l)\\
      \op(i,\sigma) &=& i\\
      \op(l_1+l_2,\sigma) &=& \sigma(l_1) +_\Int \sigma(l_2)\\
      \op(l_1 \eq l_2,\sigma) &=& \left\{
        \begin{array}{ll}
          \ktrue & (\sigma(l_1) = \sigma(l_2))\\
          \kfalse & (\sigma(l_1) \neq \sigma(l_2))
        \end{array}\right.\\
      \op(b,\sigma) &=& b\\
      \op(l_1\andd l_2,\sigma) &=& \sigma(l_1) \andd_\Bool \sigma(l_2)\\
      \op(\nott l,\sigma) &=& \nott_\Bool \sigma(l)\\
      \op((l_1,l_2),\sigma) &=& (l_1,l_2)\\
      \op(\emptyset,\sigma) &=& \emptyset\\
      \op(\{l\},\sigma) &=& \{l:1\}\\
      \op(l_1 \cup l_2,\sigma) &=& \sigma(l_1) \sqcup \sigma(l_2)\\
      \op(\kempty(l),\sigma) &=& \left\{
        \begin{array}{ll}
          \ktrue & (\sigma(l) = \emptyset)\\
          \kfalse & (\sigma(l) \neq \emptyset)
        \end{array}\right.
    \end{array}\]
    \caption{Definition of $\op$}
\label{fig:op-def}
  \[\small
  \begin{array}{c}
    \infer{\eval{\sigma}{l}{t}{\sigma[l:=\op(t,\sigma)]}}{}
    \smallskip\\
    \infer{\eval{\sigma}{l}{\letin{x=e_1}{e_2}}{\sigma''}}{
      \eval{\sigma}{l'}{e_1}{\sigma'} &
      \eval{\sigma'}{l}{e_2[l'/x]}{\sigma''} & 
      l' \fresh
    }
    \smallskip\\
    \infer{\eval{\sigma}{l}{\ifthenelse{l'}{e_{\ktrue}}{e_{\kfalse}}}{\sigma'}}{
      \sigma(l') = b & \eval{\sigma}{l}{e_b}{\sigma'}}
    \quad%\smallskip\\
    \infer{\eval{\sigma}{l}{\pi_i(l')}{\sigma[l:=\sigma(l_i)]}}{\sigma(l') = (l_1,l_2)}
    \smallskip\\%\quad
%    \infer{\eval{\sigma}{l}{\flatten l'}{\sigma[l:=\flatten \sigma[L]]}}{\sigma(l') = L}
%    \smallskip\\%\quad
%    \infer{\eval{\sigma}{l}{\summ l'}{\sigma[l:=\summ \sigma[L]]}}{\sigma(l') = L}
%    \smallskip\\
    \infer{\eval{\sigma}{l}{\flatten\{e\mid x \in l_0\}}{\sigma'[l:=\mflatten \sigma'[L']]}}
    {\evals{\sigma}{x}{\sigma(l_0)}{e}{\sigma'}{L'}}
    \smallskip\\%\quad
    \infer{\eval{\sigma}{l}{\summ\{e\mid x \in l_0\}}{\sigma'[l:=\summ \sigma'[L']]}}
    {\evals{\sigma}{x}{\sigma(l_0)}{e}{\sigma'}{L'}}  
\smallskip\\
%    \infer{\eval{\sigma}{l}{\{e\mid x \in l'\}}{\sigma'[l:=L']}}
%    {\evals{\sigma}{x}{ \sigma(l')}{e}{\sigma'}{L'}}
%    \smallskip\\%\quad
    \infer{\evals{\sigma}{x}{\emptyset}{e}{\sigma}{\emptyset}}{}
    \quad
    \infer{\evals{\sigma}{x}{\{l:m\}}{e}{\sigma'}{\{l':m\}}}
    {\eval{\sigma}{l'}{e[l/x]}{\sigma'} & l' \fresh}
    \smallskip\\%\quad
    \infer{\evals{\sigma}{x}{L_1\oplus L_2}{e}{\sigma_1 \smerge_{\sigma} \sigma_2}{L_1' \oplus L_2'}}
    {\evals{\sigma}{x}{ L_1}{e}{\sigma_1}{L_1'}
      &
      \evals{\sigma}{x}{ L_2}{e}{\sigma_2}{L_2'}}
  \end{array}\]
  \caption{Operational semantics}
  \labelFig{opsem}
\end{figure}

  We will restrict attention to NRC expressions in ``A-normal form'', defined as follows:
  \begin{eqnarray*}
    w &::=& x \mid l\\
    e &::=& w \mid \letin{x=e_1}{e_2} \mid (w_1,w_2) \mid \pi_i(w) \\
    &\mid& b \mid \nott w \mid w_1 \andd w_2
    \mid \ifthenelse{w_0}{e_1}{e_2}  \\
    &\mid& i \mid w_1 + w_2 \mid \summ \{e_2 \mid x \in w_1\}  \mid  w_1 \eq w_2 \\
    &\mid& \emptyset \mid \setof{w} \mid w_1 \cup w_2 
    %\mid \{e_2 \mid x \in w_1\} 
\mid \flatten \{e_2 \mid x \in w_1\} \mid \kempty(w)
  \end{eqnarray*}
  The A-normalization translation is standard and straightforward, so
  omitted.  The operational semantics rules are shown in
  \refFig{opsem}.  The rules are in destination-passing style.  We use
  two judgments: $\eval{\sigma}{l}{e}{\sigma'}$, meaning ``in store
  $\sigma$, evaluating $e$ at location $l$ yields store $\sigma'$'';
  and $\evals{\sigma}{x}{L}{e}{\sigma'}{L'}$, meaning ``in store
  $\sigma$, iterating $e$ with $x$ bound to each element of $L$ yields
  store $\sigma'$ and result labels $L'$.''  The second judgment deals
  with iteration over multisets involved in comprehensions; this
  exemplifies a common pattern used throughout the paper.

  Many of the rules are similar; for brevity, we use a single rule for
  \emph{terms} $t$ of the following forms:
  \begin{eqnarray*}
    t & ::= & i \mid l_1+l_2 \mid  l_1 \eq l_2 \mid b \mid \nott l \mid l_1\andd l_2\\
    &\mid& (l_1,l_2) \mid l \mid \emptyset \mid \{l\} \mid l_1 \cup l_2  \mid \kempty(l)
  \end{eqnarray*}
  Each term is either a constant, a label, or a constructor or
  primitive function applied to some labels.  The meaning of each of
  these operations is defined via the $\op$ function, as shown in
  \refFig{op-def}, which maps a term $t \in \Term$ and a store $\sigma
  : \Lab \to \Con$ to a constructor.

  When $L$ is a set of labels, we write $\sigma[L]$ for the multiset
  of constructors $\{\sigma(l):m \mid l:m \in L\}$.  This notation is
  used in the rules for $\flatten$ and $\summ$.  In this notation, the
  standard definition of summation of multisets of integers is
  $\summ\{i_1:m_1,\ldots,i_n:m_n\} = \summ_{j=1}^n i_j \cdot m_j$.
  Similarly, $\mflatten \{L_1:m_1,\ldots,L_n:m_n\} = m_1 \cdot L_1 \sqcup
  \cdots \sqcup m_n \cdot L_n\}$, where $m \cdot \{ l_1: k_1, \ldots,
  l_n:k_n\} = \{l_1 : m \cdot k_1, \ldots, l_n:m \cdot k_n\}$.

  The iteration rules $\evals{\sigma}{x}{L}{e}{\sigma'}{L'}$, evaluate
  $e$ with $x$ bound to each $l \in L$ independently, preserving the
  multiplicity of labels.  They split $L$ using $\oplus$ and combine
  the result stores using the orthogonal store merging operation
  $\smerge_\sigma$ defined as follows:
\begin{definition}[Orthogonal extensions and merging]
  We say $\sigma_1$ and $\sigma_2$ are \emph{orthogonal extensions}
  of $\sigma$ if $\sigma_1 = \sigma\uplus \sigma_1'$ and $\sigma_2=
  \sigma \uplus \sigma_2'$ and $\dom(\sigma_1') \cap \dom(\sigma_2') =
  \emptyset$, and we write $\sigma_1 \smerge_{\sigma} \sigma_2$ for
  $\sigma \uplus \sigma_1' \uplus \sigma_2'$.
\end{definition}

The operational semantics is illustrated on the \refExs{eg1}{eg2} in
\refFig{examples}; here, the labels $r, r_1,\ldots,s,\ldots$ uniquely
identify each part of the input tables $R,S$ and the labels on the
results reflect one possible labeling that is consistent with examples
given later.

%\begin{comment}

  \subsection{Type system for A-normalized expressions}

  We define typing rules for (normalized) NRC expressions as shown in
  \refFig{wf-nrc}.  We use standard \emph{contexts} $\Gamma ::= \cdot
  \mid \Gamma,x{:}\tau$ mapping variables to types and \emph{store
    types} $\Psi$ of the form $\Psi ::= \cdot \mid \Psi,l{:}\tau$.
  For brevity, we write $\Omega$ for a pair $\Psi,\Gamma$ and
  $\Omega(w)$ for $\Psi(l)$ if $l = w$ or $\Gamma(x)$ if $w = x$
  respectively.  The judgment $\wf{\Psi,\Gamma}{e}{\tau}$ means that
  given store type $\Psi$ and context $\Gamma$, expression $e$ has
  type $\tau$.

  \begin{figure}
    
    \[\small
    \begin{array}{c}
    \infer{\wfterm{\Omega}{i}{\intTy}}{}
    \quad%\smallskip\\
    \infer{\wfterm{\Omega}{w_1+w_2}{\intTy}}{\Omega(w_1) = \Omega(w_2) = \intTy}
    \quad%\smallskip\\
    \infer{\wfterm{\Omega}{w_1\eq w_2}{\boolTy}}{\Omega(w_1) = \Omega(w_2) = \intTy}
    \smallskip\\
    \infer{\wfterm{\Omega}{(w_1,w_2)}{\Omega(w_1) \times \Omega(w_2)}}{}
    \smallskip\\
    \infer{\wfterm{\Omega}{b}{\boolTy}}{}
    \quad%\smallskip\\
    \infer{\wfterm{\Omega}{w_1\andd w_2}{\boolTy}}{\Omega(w_1) = \Omega(w_2) = \boolTy}
    \quad%\smallskip\\
    \infer{\wfterm{\Omega}{\nott w}{\boolTy}}{\Omega(w) = \boolTy}
    \smallskip\\
    \infer{\wfterm{\Omega}{\emptyset}{\{\tau\}}}{}
    \quad%\smallskip\\
    \infer{\wfterm{\Omega}{\{w\}}{\{\tau\}}}{\Omega(w) = \tau}
    \quad%\smallskip\\
    \infer{\wfterm{\Omega}{w_1\cup w_2}{\{\tau\}}}{\Omega(w_1) = \{\tau\} = \Omega(w_2)}
    \smallskip\\
\infer{\wfterm{\Omega}{\kempty(w)}{\boolTy}}{\Omega(w) = \{\tau\}}
    \quad
    \infer{\wfterm{\Omega}{w}{\Omega(w)}}{}
\smallskip\\
\infer{\wf{\Omega}{t}{\tau}}{\wfterm{\Omega}{t}{\tau}}
      \quad     
\infer{\wf{\Omega}{\letin{x=e_1}{e_2}}{\tau}}{\wf{\Omega}{e_1}{\tau'} & \wf{\Omega,x{:}\tau'}{e_2}{\tau}}
     \smallskip\\
 \infer{\wf{\Omega}{\pi_i(w)}{\tau_i}}{\Omega(w)=\tau_1\times\tau_2}
     \quad
\infer{\wf{\Omega}{\ifthenelse{w}{e_\ktrue}{e_\kfalse}}{\tau}}{
    \Omega(w) = \boolTy & 
  \wf{\Omega}{e_\ktrue}{\tau} &
  \wf{\Omega}{e_\kfalse}{\tau}}
\smallskip\\
\infer{\wf{\Omega}{\flatten \{e \mid x\in w\}}{\{\tau'\}}}
{\Omega(w) = \{\tau\} & 
\wf{\Omega,x{:}\tau }{e}{\{\tau'\}}
}
\quad%\smallskip\\
\infer{\wf{\Omega}{\summ \{e \mid x\in w\}}{\intTy}}
{\Omega(w) = \{\tau\} & 
\wf{\Omega,x{:}\tau }{e}{\intTy}
}
    \end{array}
    \]
    \caption{Well-formed A-normalized NRC expressions}
    \label{fig:wf-nrc}
  \end{figure}
\begin{figure}
  
\[
\begin{array}{c}
\infer{\wfcon{\Psi}{i}{\intTy}}{}
\quad
\infer{\wfcon{\Psi}{b}{\boolTy}}{}
\quad%\smallskip\\
\infer{\wfcon{\Psi}{(l_1,l_2)}{\Psi(l_1) \times \Psi(l_2)}}{}
\smallskip\\
\infer{\wfcon{\Psi}{\{l_1:m_1,\ldots,l_n:m_n\}}{\{\tau\}}}{\tau = \Psi(l_1) =  \cdots = \Psi(l_n)}
\quad
\infer{\wfstore{\cdot}{\cdot}}{}
\quad
\infer{\wfstore{\Psi,l:\tau}{\sigma,l\mapsto k}}{\wfstore{\Psi}{\sigma} & \wfcon{\Psi}{k}{\tau}}
\end{array}
\]
\caption{Store and constructor well-formedness}
  \label{fig:wf-store}
\end{figure}

The well-formedness judgment for stores is $\wfstore{\Psi}{\sigma}$,
or ``$\sigma$ has store type $\Psi$''.  This judgment is defined in
\refFig{wf-store}, using an auxiliary judgment
$\wfcon{\Psi}{k}{\tau}$, meaning ``in stores of type $\Psi$,
constructor $k$ has type $\tau$''.  Note that well-formed stores must
be acyclic according to this judgment since the last rule permits each
label to be traversed at most once.
\begin{tr}
  The well-formedness judgment for environments $\gamma : \Var \to
  \Lab$ is $\wfctx{\Psi}{\gamma}{\Gamma}$, or ``in a store with type
  $\Psi$, environment $\gamma$ matches context $\Gamma$''.  The rules are as follows:
  \[
  \infer{\wfctx{\Psi}{\cdot}{\cdot}}{} 
  \quad
  \infer{\wfctx{\Psi}{\gamma,x \mapsto l}{\Gamma,x \mapsto \tau}}{\wfctx{\Psi}{\gamma}{\Gamma} & \Psi(l) = \tau}
  \]  
  We
  sometimes combine the judgments and write
  $\wfstorectx{\Psi}{\sigma}{\gamma}{\Gamma}$ to indicate
  $\wfstore{\Psi}{\sigma}$ and $\wfctx{\Psi}{\gamma}{\Gamma}$.
\end{tr}
The operational semantics is sound with respect to the store typing
rules:
\begin{theorem}
  Suppose $\wf{\Psi}{e}{\tau}$ and
  $\wfstore{\Psi}{\sigma}$.  Then if
  $\eval{\sigma}{l}{e}{\sigma'}$ then there exists $\Psi'$ such
  that $\Psi'(l) = \tau$ and $\wfstore{\Psi'}{\sigma'}$.
\end{theorem}
%\end{comment}

\begin{tr}
\subsection{Correctness of operational semantics}

To show the correctness of the operational semantics relative to the
denotational semantics, we need to translate from stores and labels to
values.  We define the functions $\sigma \uparrow_\tau l$ by induction
on types as follows:
\begin{eqnarray*}
  \sigma \uparrow_{\intTy} l &=& \sigma(l)\\
  \sigma \uparrow_{\boolTy} l &=& \sigma(l)\\
  \sigma \uparrow_{\tau_1 \times \tau_2} l &=& (\sigma \uparrow_{\tau_1} \pi_1 (\sigma(l)), \sigma \uparrow_{\tau_2} \pi_2(\sigma(l)))\\
\sigma \uparrow_{\setTy{\tau}} l &=& \{\sigma \uparrow_\tau l' \mid l' \in \sigma(l)\}
\end{eqnarray*}
We also define $\sigma \uparrow_\Gamma \gamma$ pointwise, so that
$(\sigma \uparrow_\Gamma \gamma)(x) = \sigma \uparrow_{\Gamma(x)}
\gamma(x)$.  We can easily show that:
\begin{proposition}
  If $\wfstore{\Psi}{\sigma}$ and $l:\tau \in \Psi$ then $\sigma
  \uparrow_\tau l \in \SB{\tau}$.  Moreover, if
  $\wfctx{\Psi}{\gamma}{\Gamma}$ then $\sigma \uparrow_\Gamma \gamma
  \in \SB{\Gamma}$.
\end{proposition}
The correctness of the operational semantics can then be established by induction on the structure of derivations:
\begin{proposition}
  Suppose that $\wf{\Gamma}{e}{\tau}$ and
  $\wfstorectx{\Psi}{\sigma}{\gamma}{\Gamma}$.  Then there exists
  $\sigma'$ such that $\eval{\sigma}{l}{\gamma(e)}{\sigma'}$.
  Moreover, for any such $\sigma'$, $\SB{e}(\sigma \uparrow_\Gamma
  \gamma) = \sigma' \uparrow_\tau l$.
\end{proposition}

\end{tr}

\section{Traced evaluation}\labelSec{traces}

We now consider \emph{traces} which are intended to capture the
``execution history'' of a query in a form that is itself suitable for
querying.  We define traces $T$ using the terms
introduced earlier as follows:
\begin{eqnarray*}
  T &::=&  l \gets t \mid l \gets \trproj{i}{l'}{l''} \mid \trconde[l]{l'}{b}{T}{e_1}{e_2} \mid \trlet{T_1}{T_2}\\
  &\mid& l \gets \trsume{l'}{\Ts}{x.e} \mid l \gets \trcompe{l'}{\Ts}{x.e} 
%\mid l \gets \trmape{l'}{\Ts}{x.e}
  \\
  \Ts &::=& \{[l_1]T_1:m_1,\ldots,[l_n]T_n:m_n\}
\end{eqnarray*}
Terms, introduced above, describe single computation steps.  Labeled
trace collections $\Ts$ are multisets of labeled traces $[l]T$.
\emph{Assignment traces} $l \gets t$ record that a new label $l$ was
created and assigned the value described by trace term $t$.
\emph{Projection traces} $l \gets \trproj{i}{l'}{l''}$ record that $l$
was created and assigned the value at $l''$, by projecting the $i$-th
component of pair $l'$.  \emph{Sequential composition traces}
$T_1;T_2$ indicate that $T_1$ was performed first followed by $T_2$.
\emph{Conditional traces} $\trconde[l]{l'}{b}{T}{e_1}{e_2}$ record
that a conditional expression tested $l'$, found it equal to boolean
$b$, and then performed trace $T$ that writes to $l$.  In addition,
conditional traces record the alternative expressions $e_1$ and $e_2$
corresponding to the true and false branches.  \emph{Comprehension
  traces} $l \gets \trcompe{l'}{\Ts}{x.e}$ record that $l$ was created
by performing a comprehension over the set at $l'$, with subtraces
$\Ts$ describing the iterations; the expression $x.e$ records the body
of the comprehension with its bound variable $x$. Sum traces $l \gets
\trsume{l}{\Ts}{x.e}$ are similar.

When the expressions $e_1,e_2,x.e$ in conditional or comprehension
traces are irrelevant to the discussion we often omit them for
brevity, e.g. writing $\trcond[l]{l'}{b}{T}$ or $\trcomp{l}{\Ts}$.

We define the result label of a trace as follows:
\begin{eqnarray*}
  \result(l \gets t) &=& l\\
  \result(\trlet{T_1}{T_2}) &=& \result(T_2)\\
  \result(\trconde[l]{l'}{b}{T}{e_1}{e_2}) &=& l\\
  \result(l\gets \trproj{i}{l'}{l''}) &=& l\\
%  \result(l\gets \trmape{l'}{\Ts}{x.e}) &=& l\\
  \result(l\gets \trcompe{l'}{\Ts}{x.e}) &=& l\\
  \result(l\gets \trsume{l'}{\Ts}{x.e}) &=& l
\end{eqnarray*}
We define the input labels of a labeled trace set $\Ts$ as
$\inputs(\Ts) = \{l:m \mid [l]T:m \in \Ts\}$.  Similarly, the result
labels of $\Ts$ are defined as $\outputs(\Ts) = \{\result(T):m \mid
[l]T:m \in \Ts\}$.  Note that we treat both as multisets.

\subsection{Traced operational semantics}

\begin{figure}[tb]
  
\[\small
\begin{array}{c}
\infer{\treval{\sigma}{l}{t}{\sigma[l:=\op(t,\sigma)]}{l\gets t}}{}
\smallskip\\
\infer[l' \fresh]{\treval{\sigma}{l}{\letin{x=e_1}{e_2}}{\sigma_2}{\trlet{T_1}{T_2}}}
{\treval{\sigma}{l'}{e_1}{\sigma_1}{T_1} & 
\treval{\sigma}{l}{e_2[l'/x]}{\sigma_2}{T_2} }
\smallskip\\
\infer{\treval{\sigma}{l}{\ifthenelse{l'}{e_\ktrue}{e_\kfalse}}{\sigma'}{\trconde[l]{l'}{b}{T}{e_\ktrue}{e_\kfalse}}}{\sigma(l')
  = b & \treval{\sigma}{l}{e_b}{\sigma'}{T} }
\smallskip\\
\infer{\treval{\sigma}{l}{\pi_i l'}{\sigma[l:=\sigma(l_i)]}{l\gets \kproj_i(l',l_i)}}{\sigma(l') = (l_1,l_2)}
\smallskip\\
%\infer{\treval{\sigma}{l}{\{e\mid x \in l'\}}{\sigma'[l:=L']}{l\gets\trmape{l'}{\Ts}{x.e}}}
%{\trevals{\sigma}{x}{ \sigma(l')}{e}{\sigma'}{L'}{\Ts}}
%\smallskip\\
%    \infer{\treval{\sigma}{l}{\flatten l'}{\sigma[l:=\flatten \sigma[\sigma(l')]]}{l \gets \trflat{l'}{\sigma(l')}}}{}
%      \smallskip\\
%    \infer{\treval{\sigma}{l}{\summ l'}{\sigma[l:=\summ \sigma[\sigma(l')]]}{l \gets \trsum{l'}{\sigma(l')}}}{}
%\smallskip\\  
\infer{\treval{\sigma}{l}{\flatten\{e\mid x \in l'\}}{\sigma'[l:=\mflatten \sigma'[L']]}{l\gets \trcompe{l'}{\Ts}{x.e}}}
{\trevals{\sigma}{x}{\sigma(l')}{e}{\sigma'}{L'}{\Ts}}
\smallskip\\
\infer{\treval{\sigma}{l}{\summ\{e\mid x \in l'\}}{\sigma'[l:=\summ \sigma'[L']]}{l\gets \trsume{l'}{\Ts}{x.e}}}
{\trevals{\sigma}{x}{\sigma(l')}{e}{\sigma'}{L'}{\Ts}}
\smallskip\\
\infer{\trevals{\sigma}{x}{\emptyset}{e}{\sigma}{\emptyset}{\emptyset}}{}
\qquad
\infer{\trevals{\sigma}{x}{\{l:m\}}{e}{\sigma'}{\{l':m\}}{\{[l]T:m\}}}
{\treval{\sigma}{l'}{e[l/x]}{\sigma'}{T} & l' \fresh}
\smallskip\\
\infer{\trevals{\sigma}{x}{L_1\oplus L_2}{e}{\sigma_1 \smerge_{\sigma} \sigma_2}{L_1' \oplus L_2'}{\Ts_1 \oplus \Ts_2}}
{\trevals{\sigma}{x}{ L_1}{e}{\sigma_1}{L_1'}{\Ts_1}
&
\trevals{\sigma}{x}{ L_2}{e}{\sigma_2}{L_2'}{\Ts_2}}
\end{array}\]
  \caption{Traced evaluation}
\labelFig{trace-semantics}
\end{figure}

We now define \emph{traced evaluation}, a refinement of the
operational semantics in \refSec{opsem}.  The rules for traced
evaluation are shown in \refFig{trace-semantics}.  There are two
judgments: $\treval{\sigma}{l}{e}{\sigma'}{T}$, meaning ``Starting in
store $\sigma$, evaluating $e$ and storing the result at $l$ yields
store $\sigma'$ and trace $T$'', and
$\trevals{\sigma}{x}{L}{e}{\sigma'}{L'}{\Ts}$, meaning ``Starting in
store $\sigma$, evaluating $e$ with $x$ bound to each label in $L$ in
turn yields store $\sigma'$, result labels $L'$ and labeled traces
$\Ts$''.

Each operational semantics rule relates a different expression form to
its trace form.  Thus, traces can be viewed as explaining the dynamic
execution history of the expression.  (We will make this precise in
\refSec{metatheory}). In particular, terms $t$ are translated to
assignment traces.  Let-expressions are translated to sequential
compositions of traces.  For these expressions, it would be
superfluous to record additional information such as the values of the
inputs and outputs, since this can be recovered from the input store
and the trace (as we shall see below).  However, more detailed trace
information is needed for some expressions, such as projections,
conditionals, comprehensions, and sums.  Their traces record some
expression annotations and some information about the structure of the
input store.  Conditionals record the boolean value of the conditional
test as well as both branches of the conditional; comprehensions and
sums record the labels and subtraces of the elements of the input set
as well as the body of the comprehension.  This information is
necessary to obtain the fidelity property (\refSec{metatheory}) and to
ensure that we can extract other forms of provenance from traces
(\refSec{extraction}).

  \begin{figure}
\small
\begin{verbatim}
l <- comp(r,{
  [r1] x11 <- proj_C(r1,r13); x1 <- comp(s,{
    [s1] x111 <- proj_C(s1,s11); x112 <- x11 = x111; 
         cond(x112,f,x113 <- {}),
    [s2] x121 <- proj_C(s2,s21); x122 <- x11 = x121; 
         cond(x122,f,x123 <- {}),
    [s3] x131 <- proj_C(s3,s31); x132 <- x11 = x131; 
         cond(x132,t,l11 <- proj_A(r1,r11);
                     l12 <- proj_B(r1,r12);
                     l13 <- proj_D(s3,s32);
                     l1 <- (A:l11,B:l12,D:l13);
                     x136 <- {l1})}),
  [r2] x21 <- proj_C(r2,r23); x2 <- comp(s,{
    [s1] x211 <- proj_C(s1,s11); x212 <- x21 = x211; 
         cond(x212,f,x213 <- {}),
    [s2] x221 <- proj_C(s2,s21); x222 <- x21 = x221; 
         cond(x222,f,x223 <- {}),
    [s3] x231 <- proj_C(s3,s31); x232 <- x21 = x231; 
         cond(x232,t,l21 <- proj_A(r2,r21);
                     l22 <- proj_B(r2,r22);
                     l23 <- proj_D(s3,s32);
                     l2 <- (A:l21,B:l22,D:l23);
                     x126 <- {l2})}),
  [r3] x31 <- proj_C(r3,r33); x3 <- comp(s,{
    [s1] x311 <- proj_C(s1,s11); x312 <- x31 = x311;  
         cond(x312,f,x313 <- {}),
    [s2] x321 <- proj_C(s2,s21); x322 <- x31 = x321; 
         cond(x322,f,x323 <- {}),
    [s3] x331 <- proj_C(s3,s31); x332 <- x31 = x331;  
         cond(x332,f,x333 <- {})})})
\end{verbatim}
  \caption{Example trace for query $Q_1$}\labelFig{trace-example1}
\end{figure}
\begin{figure}
    \small
\begin{verbatim}
l11' <- 42;  x1 <- 2;
l12' <- sum(s,{
  [s1] x11 <- proj_C(s1,s11); x12 <- x11 = x1; 
       cond(x12,t, x13 <- proj_D(s1,s12)),
  [s2] x21 <- proj_C(s2,s21); x22 <- x21 = x1; 
       cond(x22,t, x23 <- proj_D(s2,s22)),
  [s3] x31 <- proj_C(s3,s31); x32 <- x31 = x1;
       cond(x32,f, x33 <- 0)});
l1' <- (C:l11',D:l12'); x <- {l1'}; y12 <- 4;
y <- comp(r,{
  [r1] y11 <- proj_C(r1,r13); y12 <- y11 = y1;
       cond(y12,f, y13 <- {}),
  [r2] y21 <- proj_C(r2,r21); y22 <- y21 = y1;
       cond(y22,f,y23 <- {}),
  [r3] y31 <- proj_C(r3,r31); y32 <- y31 = y1;
       cond(y32,t,l21' <- proj_B(r3,r32);
                  l22' <- proj_A(r3,r31);
                  l2' <- (C:l21',D:l22')
                  y33 <- {l2'})});
l' <- x U y
\end{verbatim}
    
    \caption{Example trace for query $Q_2$}
\label{fig:trace-example2}
\end{figure}

\begin{example}
  \refFig{trace-example1} shows one possible trace resulting from
  normalizing and running query $Q_1$ from \refEx{eg1} on the data in
  \refFig{examples}.  Similarly, \refFig{trace-example2} shows a
  possible trace of the grouping-aggregation query $Q_2$ from
  \refEx{eg2}.  Since the example queries use record syntax, we use
  terms such as $(\vec{A}:\vec{l})$ and traces $l \gets
  \trproj{A}{l'}{l''}$ for record construction and field projection.
  These operations are natural generalizations of pair
  terms and projection traces.  For brevity, the examples omit
  expression annotations.
\end{example}

\begin{tr}
  We will need the following property:
  \begin{lemma}\labelLem{result-tech}
    If $\treval{\sigma}{l}{e}{\sigma'}{T}$ then $\result(T) = l$.
  \end{lemma}
  \begin{proof}
    Easy induction on derivations.
  \end{proof}
  \end{tr}

\section{Provenance extraction}\labelSec{extraction}

% \begin{itemize}
% \item Other proposed forms of provenance can be defined in terms of
%   traces, and extracted from backwards trace slices
% \item Where-provenance: copying subtrace starting at origin and ending at result
% \item Dependency-provenance: subtrace that shows dependences.  Need to be careful: do we need ``parents'' (set, record nodes) of $l$ in its relevant subtrace? 

% \end{itemize}

As we discussed in \refSec{intro}, a number of forms of provenance
have been defined already in the literature.  Although most of this
work has focused on flat relational queries, several techniques have
recently been extended to the NRC.  Thus, a natural question is: are
traces related to these other forms of provenance?

In this section we describe algorithms for extracting
where-provenance~\cite{buneman07icdt}, dependency
provenance~\cite{DBLP:conf/dbpl/CheneyAA07}, and semiring
provenance~\cite{DBLP:conf/pods/FosterGT08} from traces.  We will
develop extraction algorithms and prove them correct relative to the
existing definitions.  However, our operational formulation of traces
is rather different from existing denotational presentations of
provenance semantics, so we need to set up appropriate correspondences
between store-based and value-based representations.
\begin{sub} 
  Precisely formulating these equivalences requires introducing
  several auxiliary definitions and properties, which are largely
  irrelevant to the rest of this paper.  Therefore, the full
  translations and details are in the companion technical
  report~\cite{tr}.
\end{sub}
\begin{tr}
  Precisely formulating these equivalences requires introducing
  several auxiliary definitions and properties.
\end{tr}

We also discuss how provenance extraction yields insight into the
meaning of other forms of provenance.  We can view the extraction
algorithms as dynamic analyses of the provenance trace.  For example,
where-provenance can be viewed an analysis that identifies ``chains of
copies'' form the input to the output.  Conversely, we can view
high-level properties of traces as clear specifications that can be
used to justify new provenance-tracking techniques.

The fact that several distinct forms of provenance can all be
extracted from traces is a clear qualitative indication that traces
are very general.  This generality is not surprising in light of the
fidelity property, which essentially requires that the traces
accurately represent the query in all inputs.  In fact, the provenance
extraction rules do not inspect the expression annotations
$x.e,e_1,e_2$ in comprehension and conditional traces; thus, they all
work correctly even without these annotations.  Also, the extraction
rules do not have access to the underlying store $\sigma$; nor do they
need to reconstruct the intermediate store.  The trace itself records
enough information about the store labels actually accessed.

We first fix some terminology used in the rest of the section.  We
consider an \emph{annotated store} $\annot{\sigma}{h}$ to consist of a
store $\sigma$ and a function $h : \dom(\sigma) \to A$ assigning each
label in $\sigma$ to an annotation in $A$.
\begin{tr}
  We also consider several kinds of \emph{annotated values}.  In
  general, a value $v \in \annot{\Val}{A}$ with annotations $a$ from some set
  $A$ is an expression of the form
  \begin{eqnarray*}
    v &::=& w^x\\
    w &::=& i \mid b \mid (v_1,v_2) \mid \{v_1,\ldots,v_n\}
  \end{eqnarray*}
  This syntax strictly generalizes that of ordinary values since
  ordinary values can be viewed as values annotated by elements of
  some unit set $\{\star\}$, up to an obvious isomorphism.  Also, we
  write $|v|$ for the ordinary value obtained by erasing the
  annotations from $v$.  This is defined as:
  \begin{eqnarray*}
    |i^x| &=& i\quad  |b^x| = b \quad
    |(v_1,v_2)^x| = (|v_1|,|v_2|)\\
    |\{v_1,\ldots,v_n\}| &=& \{|v_1|,\ldots,|v_n|\}
  \end{eqnarray*}
  Moreover, we define $\lfloor w^x \rfloor = w$ and $\lceil w^x \rceil
  = x$.

  Given an $A$-annotated store $\annot{\sigma}{h}$, we can extract
  annotated values using the same technique as extracting ordinary
  values from an ordinary store:
  \begin{eqnarray*}
    \annot{\sigma}{h} \uparrow^A_{\intTy} l &=& \sigma(l)^{h(l)}\\
    \annot{\sigma}{h} \uparrow^A_{\boolTy} l &=&  \sigma(l)^{h(l)}\\
    \annot{\sigma}{h}\uparrow^A_{\tau_1 \times \tau_2}  l &=& ( \annot{\sigma}{h} \uparrow_{\tau_1}^A l_1, \annot{\sigma}{h} \uparrow_{\tau_2}^A l_2)^{h(l)} \quad (\sigma(l) = (l_1,l_2))\\
   \annot{\sigma}{h}\uparrow^A_{\{\tau\}}  l &=& \{\annot{\sigma}{h} :m\uparrow_{\tau}^A l' \mid l':m \in \sigma(l)\}^{h(l)}
  \end{eqnarray*}
  Moreover, for $\gamma : \Var \to \Lab$ we again write
  $\annot{\sigma}{h} \uparrow_\Gamma^A \gamma : \Var \to
  \annot{\Val}{A}$ for the extension of the annotated value extraction
  function from labels to environments.  Similarly, for $L $ a
  collection of labels we write $\annot{\sigma}{h}
  \uparrow_{\{\tau\}}^A L$ for $\{\sigma \uparrow_\tau^A l : m \mid l:m \in
  L\}$.
\end{tr}

\subsection{Where-provenance}

As discussed by~\cite{buneman01icdt,buneman07icdt}, where-provenance
is information about ``where an output value came from in the input''.
\citet{buneman07icdt} defined where-provenance semantics for NRC
queries via values annotated with optional annotations $A_\bot = A
\uplus \{\bot\}$.  Here, $\bot$ stands for the absence of
where-provenance, and $A$ is a set of tokens chosen to uniquely
address each part of the input.

The idea of where-provenance is that values ``copied'' via variable or
projection expressions retain their annotations, while other
operations produce results annotated with $\bot$. We use an auxiliary
function
\begin{eqnarray*}
  \where(l,h) &=& h(l)\\
  \where(t,h) &= &\bot \quad (t \neq l)
\end{eqnarray*}
that defines the annotation of the result of a term $t$ with respect to $h
: \Lab \to A_\bot$ to be preserved if $t = l$ and otherwise $\bot$.
\citet{buneman07icdt} did not consider integer operations or sums; we
support them by annotating the results with $\bot$.

  \begin{tr}
    We first review the denotational presentation of where-provenance
    from ~\cite{buneman07icdt}.  \refFig{where-denot} shows the
    semantics of expressions $e$ as a function $\WSB{e}$ mapping
    contexts $\gamma : \Var \to \annot{\Val}{A_\bot}$ to $A_\bot$-annotated
    values.

    \begin{figure*}[p]
      
      \begin{eqnarray*}
        \WSB{x}\gamma &=& \gamma(x)\\
        \WSB{\letin{x = e_1}{e_2}} &=& \WSB{e_2}\gamma[x := \WSB{e_1}\gamma]\\
        \WSB{i} \gamma &=& i^\bot\\
        \WSB{e_1+e_2} \gamma &=& (\getv{\WSB{e_1}\gamma} + \getv{\WSB{e_2}\gamma})^\bot\\
        \WSB{\summ \{e \mid x \in e_0\}}\gamma &=& 
        (\summ \{\getv{\WSB{e}\gamma[x\mapsto v]} \mid v \in \getv{\WSB{e_0}\gamma}\})^\bot \\
        \WSB{b}\gamma &=& b^\bot\\
        \WSB{\nott e}\gamma &=& (\nott\geta{\WSB{e}\gamma})^\bot\\
        \WSB{e_1 \andd e_2}\gamma &=& (\geta{\WSB{e_1}\gamma} \andd \geta{\WSB{e_2}\gamma})^\bot\\
        \WSB{(e_1,e_2)}\gamma &=& (\WSB{e_1}\gamma ,\WSB{e_2}\gamma )^\bot\\
        \WSB{\pi_i(e)}\gamma &=& \pi_i(\getv{\WSB{e}\gamma})\\
        \WSB{\emptyset}\gamma &=& \emptyset^\bot\\
        \WSB{\setof{e}}\gamma &=& \{\WSB{e}\gamma\}^\bot\\
        \WSB{e_1 \cup e_2}\gamma &=& (\getv{\WSB{e_1}\gamma} \cup \getv{\WSB{e_2}\gamma})^\bot\\
        \WSB{\flatten\{e \mid x \in e_0\}}\gamma &=& (\mflatten\{\getv{\WSB{e}\gamma[x\mapsto v]}\mid v \in \getv{\WSB{e_0}\gamma}\})^\bot\\
        \WSB{\ifthenelse{e_0}{e_1}{e_2}}\gamma &=& \left\{ 
          \begin{array}{ll}
            \WSB{e_1}\gamma & \text{if $\getv{\WSB{e_0}{\gamma}} = \ktrue$}\\
            \WSB{e_2}\gamma & \text{if $\getv{\WSB{e_0}{\gamma}} = \kfalse$}
          \end{array}\right.\\
        \WSB{e_1 \eq e_2}\gamma &=& \left\{ 
          \begin{array}{ll}
            \ktrue^\bot & \text{if $\getv{\WSB{e_1}{\gamma}} = \getv{\WSB{e_2}{\gamma}}$}\\
            \kfalse^\bot& \text{if $\getv{\WSB{e_1}{\gamma}} \neq \getv{\WSB{e_2}{\gamma}}$}
          \end{array}\right.
        \\
        \WSB{\kempty(e)}\gamma &=& \left\{ 
          \begin{array}{ll}
            \ktrue^\bot & \text{if $\getv{\WSB{e}{\gamma}} = \emptyset$}\\
            \kfalse^\bot& \text{if $\getv{\WSB{e}{\gamma}} \neq \emptyset$}
          \end{array}\right.
      \end{eqnarray*}
      
      \caption{Where-provenance, denotationally}
      \label{fig:where-denot}
    \end{figure*}

    In \refFig{where-op}, we introduce an equivalent operational
    formulation.  We define judgments
    $\heval[W]{\sigma}{h}{l}{e}{\sigma'}{h'}$ for expression
    evaluation and $\hevals[W]{\sigma}{h}{x}{L}{e}{\sigma'}{h'}{L'}$
    for iteration, both with where-provenance propagation.

    It is straightforward to prove by induction that:
    \begin{theorem}
      ~
      \begin{enumerate}
      \item Suppose $\wf{\Gamma}{e}{\tau}$ and
        $\wfstorectx{\Psi}{\sigma}{\gamma}{\Gamma}$.  Then
        $\heval[W]{\sigma}{h}{l}{\gamma(e)}{\sigma'}{h'}$ if and only
        if $\WSB{e} (\annot{\sigma}{h} \uparrow^{A_\bot}_\Gamma
        \gamma) = \annot{\sigma'}{h'} \uparrow^{A_\bot}_\tau l$.
        \item Suppose $\wf{\Gamma,x:\tau}{e}{\{\tau'\}}$ and $\wfstorectx{\Psi}{\sigma}{\gamma}{\Gamma}$.  Then 
$\hevals[W]{\sigma}{h}{x}{L}{\gamma(e)}{\sigma'}{h'}{L'}$ if and only if
         $\{\WSB{e}\gamma[x := v] \mid v \in
        \annot{\sigma}{h}\uparrow^{A_\bot}_{\{\tau\}} L\} =
        \annot{\sigma'}{h'}\uparrow^{A_\bot}_{\{\tau'\}} L'$.
        \end{enumerate}
    \end{theorem}
\end{tr}

\begin{sub}
  In \refFig{where-op}, we define an operational formulation of
  where-provenance that is equivalent to the original denotational
  presentation (see ~\cite{tr}).  We define judgments
  $\heval[W]{\sigma}{h}{l}{e}{\sigma'}{h'}$ for expression evaluation
  and $\hevals[W]{\sigma}{h}{x}{L}{e}{\sigma'}{h'}{L'}$ for iteration,
  both with where-provenance propagation.
\end{sub}
\begin{figure}[tb]
  \small
  \[
  \begin{array}{c}
    \infer{\heval[W]{\sigma}{h}{l}{t}{\sigma[l:=t]}{h[l:=\where(t,h)]}}{}
    \smallskip\\
    \infer{\heval[W]{\sigma}{h}{l}{\letin{x=e_1}{e_2}}{\sigma''}{h''}}{
      \heval[W]{\sigma}{h}{l'}{e_1}{\sigma'}{h'} &
      \heval[W]{\sigma'}{h'}{l}{e_2[l'/x]}{\sigma''}{h''} & 
      l' \fresh
    }
    \smallskip\\
    \infer{\heval[W]{\sigma}{h}{l}{\pi_i(l')}{\sigma[l:=\sigma(l_i)]}{h[l:=h(l_i)]}}{\sigma(l') = (l_1,l_2)}
    \smallskip\\
    \infer{\heval[W]{\sigma}{h}{l}{\ifthenelse{l'}{e_{\ktrue}}{e_{\kfalse}}}{\sigma'}{h'}}{
      \sigma(l') = b & \heval[W]{\sigma}{h}{l}{e_b}{\sigma'}{h'}}
    \smallskip\\
%    \infer{\heval[W]{\sigma}{h}{l}{\{e \mid x \in l'\}}{\sigma'[l:=L']}{h[l:=\bot]}}{\hevals[W]{\sigma}{h}{x}{\sigma(l)}{e}{\sigma'}{h'}{L'}}
%    \smallskip\\
    \infer{\heval[W]{\sigma}{h}{l}{\flatten \{e \mid x \in l'\}}{\sigma'[l:=\mflatten \sigma'[L']]}{h[l:=\bot]}}{\hevals[W]{\sigma}{h}{x}{\sigma(l)}{e}{\sigma'}{h'}{L'}}
    \smallskip\\
    \infer{\heval[W]{\sigma}{h}{l}{\sum \{e \mid x \in l'\}}{\sigma'[l:=\summ \sigma'[L']]}{h[l:=\bot]}}{\hevals[W]{\sigma}{h}{x}{\sigma(l)}{e}{\sigma'}{h'}{L'}}
\smallskip\\
\infer{\hevals[W]{\sigma}{h}{x}{\emptyset}{e}{\sigma}{h}{\emptyset}}{}
\smallskip\\
\infer{\hevals[W]{\sigma}{h}{x}{L_1 \oplus L_2}{e}{\sigma_1 \uplus_\sigma \sigma_2}{h_1 \uplus _h h_2}{L_1' \oplus L_2'}}{\hevals[W]{\sigma}{h}{x}{L_1}{e}{\sigma_1}{h_1}{L_1'} & \hevals[W]{\sigma}{h}{x}{L_1}{e}{\sigma_2}{h_2}{L_2'}}
\smallskip\\
\infer{\hevals[W]{\sigma}{h}{x}{\{l:m\}}{e}{\sigma'}{h'}{\{l':m\}}}{\heval[W]{\sigma}{h}{l'}{e[l/x]}{\sigma'}{h'} & l' \fresh}
  \end{array}
  \]
  \caption{Where-provenance, operationally}
  \label{fig:where-op}
\end{figure}

\begin{figure}
  
\small
  \[
  \begin{array}{c}
    \infer{\wext{h}{l\gets t}{h[l:=\where(t,h)]}}{}
\quad
\infer{\wext{h}{\trlet{T_1}{T_2}}{h''}}{\wext{h}{T_1}{h'} & \wext{h'}{T_2}{h''}}
\smallskip\\
\infer{\wext{h}{l \gets \trproj{i}{l'}{l''}}{h[l:=h(l'')]}}{}
\quad
\infer{\wext{h}{\trcond[l]{l'}{b}{T}}{h'}}{\wext{h}{T}{h'}}
\smallskip\\
%\infer{\wext{h}{l \gets \trflat{l'}{L}}{h[l:=\bot]}}{}
%\quad
%\infer{\wext{h}{l \gets \trsum{l'}{L}}{h[l:=\bot]}}{}
%\smallskip\\
%\infer{\wext{h}{l\gets \trmap{l'}{\Ts}}{h'[l:=\bot]}}{\wexts{h}{\Ts}{h'}}
%\smallskip\\
\infer{\wext{h}{l \gets \trcomp{l'}{\Ts}}{h'[l:=\bot]}}{\wexts{h}{\Ts}{h'}}
\quad
\infer{\wext{h}{l \gets \trsum{l'}{\Ts}}{h'[l:=\bot]}}{\wexts{h}{\Ts}{h'}}
\smallskip\\
\infer{\wexts{h}{\emptyset}{h}}{}
\quad
\infer{\wexts{h}{\Ts_1 \oplus \Ts_2}{h_1 \uplus_h h_2}}{ \wexts{h}{\Ts_1}{h_1} & \wexts{h}{\Ts_2}{h_2}}
\quad
\infer{\wexts{h}{\{[l]T:m\}}{h'}}{\wext{h}{T}{h'}}
  \end{array}
  \]
  \caption{Extracting where-provenance}
\label{fig:where-from-trace}
\end{figure}

The where-provenance extraction relation is shown in
\refFig{where-from-trace}; we define judgment $\wext{h}{T}{h'}$, which
takes input annotations $h$ and propagates them through $T$ to yield
output annotations $h'$, and judgment $\wexts{h}{\Ts}{h'}$ which
propagates annotations through a set of traces.  Where-provenance
extraction can be shown correct relative to the operational
where-provenance semantics, as follows:
\begin{sub}
  \begin{theorem}
    Suppose $\treval{\sigma}{l}{e}{\sigma'}{T}$ and $h :
    \dom(\sigma)\to A_\bot$ is given.  Then
    $\heval[W]{\sigma}{h}{l}{e}{\sigma'}{h'}$ holds if and only if
    $\wext{h}{T}{h'}$ holds.  
  \end{theorem}
\end{sub}
\begin{tr}
  \begin{theorem}
    ~
    \begin{enumerate}
    \item Suppose $\treval{\sigma}{l}{e}{\sigma'}{T}$ and $h :
      \dom(\sigma)\to A_\bot$ is given.  Then
      $\heval[W]{\sigma}{h}{l}{e}{\sigma'}{h'}$ holds if and only if
      $\wext{h}{T}{h'}$ holds.  
    \item If $\trevals{\sigma}{x}{L}{e}{\sigma'}{L'}$ then
      $\hevals[W]{\sigma}{h}{x}{L}{e}{\sigma'}{h'}{L'}$ if and only if
      $\wexts{h}{\Ts}{h'}$.
    \end{enumerate}
  \end{theorem}
\end{tr}
\begin{figure}

\begin{tikzpicture}
  \node at (0,0) {A};
  \node at (1,0) {B};
  \node at (2,0) {D};
  
  \node (atom11) at (0,-0.5) [inner sep = 0pt] {$1$};
  \node (atom12) at (1,-0.5) [inner sep = 0pt] {$2$};
  \node (atom13) at (2,-0.5) [inner sep = 0pt] {$7$};
  \node (atom21) at (0,-1) [inner sep = 0pt] {$1$};
  \node (atom22) at (1,-1) [inner sep = 0pt] {$3$};
  \node (atom23) at (2,-1) [inner sep = 0pt] {$7$};

  \node (table) at (-0.3,0.3) [anchor=north west,rectangle,draw,minimum width = 3cm, minimum height
  = 1.6cm] {};
  
  \node (caption) at ([yshift=-0.4cm]table.south) {Output table $Q_1(A,B,D)$};

  \node (row1) at (-0.2,-0.3) [help lines,anchor=north west,rectangle,draw,minimum width = 2.8cm, minimum height
  = 0.4cm] {};
  \node (row2) at (-0.2,-0.8) [help lines,anchor=north west,rectangle,draw,minimum width = 2.8cm, minimum height
  = 0.4cm] {};

  \node (ctab) at ([yshift=+5mm]table.north) {\small $\bot$};
  \node (crow1) at ([xshift=-6mm]row1.west) {\small $\bot$};
  \node (crow2) at ([xshift=-6mm]row2.west) {\small $\bot$};
  \node (catom11) at (0.4,-0.45) {\small $r_{11}$};
  \node (catom12) at (1.4,-0.45) {\small $r_{12}$};
  \node (catom13) at (2.4,-0.45) {\small $s_{23}$};
  \node (catom21) at (0.4,-0.95) {\small $r_{21}$};
  \node (catom22) at (1.4,-0.95) {\small $r_{22}$};
  \node (catom23) at (2.4,-0.95) {\small $s_{23}$};

  \draw (table.north) -- (ctab.south);
  \draw (row1.west) -- (crow1);
  \draw (row2.west) -- (crow2);

  \node at (4,0) {C};
  \node at (5,0) {D};
  
  \node (atom1) at (4,-0.5) [inner sep = 0pt] {$42$};
  \node (atom2) at (5,-0.5) [inner sep = 0pt] {$7$};
  \node (atom3) at (4,-1) [inner sep = 0pt] {$42$};
  \node (atom4) at (5,-1) [inner sep = 0pt] {$7$};

  \node (table) at (3.7,0.3) [anchor=north west,rectangle,draw,minimum width = 2cm, minimum height
  = 1.6cm] {};
  
  \node (caption) at ([yshift=-0.4cm]table.south) {Output table $Q_2(C,D)$};

  \node (row1) at (3.8,-0.3) [help lines,anchor=north west,rectangle,draw,minimum width = 1.8cm, minimum height
  = 0.4cm] {};
  \node (row2) at (3.8,-0.8) [help lines,anchor=north west,rectangle,draw,minimum width = 1.8cm, minimum height
  = 0.4cm] {};

  \node (ctab) at ([yshift=+5mm]table.north) {\small $\bot$};
  \node (crow1) at ([xshift=-6mm]row1.west) {\small $\bot$};
  \node (crow2) at ([xshift=-6mm]row2.west) {\small $\bot$};
  \node (catom1) at (4.4,-0.45) {\small $\bot$};
  \node (catom2) at (5.4,-0.45) {\small $\bot$};
  \node (catom3) at (4.4,-0.95) {\small $r_{32}$};
  \node (catom4) at (5.4,-0.95) {\small $r_{31}$};

  \draw (table.north) -- (ctab.south);
  \draw (row1.west) -- (crow1);
  \draw (row2.west) -- (crow2);
\end{tikzpicture}%
\caption{Where-provenance extraction examples}
\labelFig{where-ext-examples}
\end{figure}
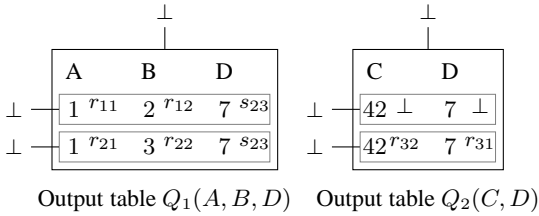

\begin{example}
  \refFig{where-ext-examples} shows the results of where-provenance
  extraction for \refExs{eg1}{eg2}.  For the inputs and results in
  \refFig{examples}, the field values copied from the input have
  provenance links to their sources, whereas values computed from
  several values have no where-provenance ($\bot$).
\end{example}

\begin{definition}
  A \emph{copy} with source $l'$ and target $l$ is a trace of either
  the form $l \gets l'$ or $l \gets \trproj{i}{l''}{l'}$.  A
  \emph{chain of copies} from $l_0$ to $l_n$ is a sequence of trace
  steps $T_1;\ldots;T_n$ where each step $T_i$ is a copy from
  $l_{i-1}$ to $l_i$.  We say that a trace $T$ \emph{contains a chain
    of copies} from $l'$ to $l$ if there is a chain of copies from
  $l'$ to $l$ all of whose operations are present in $T$.
\end{definition}
Let $\id_\sigma : \dom(\sigma) \to
\dom(\sigma)_\bot$ be the (lifted) identity function on $\sigma$.
\begin{proposition}
  Suppose $\treval{\sigma}{l}{e}{\sigma'}{T}$ and
  $\wext{\id_\sigma}{T}{h}$.  Then for each $l' \in \dom(\sigma')$,
  $h(l') \neq \bot$ if and only if there is a chain of copies from
  $h(l')$ to $l'$ in $T$.
\end{proposition}

Moreover, where-provenance can easily be extracted from a trace for a
single input or output label rather than for all of the labels
simultaneously, simply by traversing the trace.  Though this takes
time $O(|T|)$ in the worst case, we could do much better if the traces
are represented as graphs rather than as syntax trees.  

\subsection{Dependency provenance}

We next consider extracting the \emph{dependency provenance}
introduced in our previous work~\cite{DBLP:conf/dbpl/CheneyAA07}.
Dependency provenance is motivated by the concepts of dependency that
underlie program slicing~\cite{venkatesh91pldi} and noninterference
in information flow security, as formalized, for instance, in the
Dependency Core Calculus~\cite{abadi99popl}.  We consider NRC values
annotated with sets of tokens and define an annotation-propagating
semantics.

Dependency provenance annotations are viewed as correct when they link
each part of the input to all parts of the output that \emph{may}
change if the input part is changed.  This is similar to
non-interference.  The resulting links can be used to ``slice'' the
input with respect to the output and vice versa.
\citet{DBLP:conf/dbpl/CheneyAA07} established that, as with minimal
program slices, minimal dependency provenance is not computable, but
gave dynamic and static approximations.  Here, we will show how to
extract the dynamic approximation from traces.

Dependency provenance can be modeled using values $v \in \annot{\Val}{\Pow{A}}$
annotated with sets of tokens from $A$.  We introduce an auxiliary
function $dep(t,h)$ for calculating the dependences of basic terms $t$
relative to annotation functions $h : \Lab \to \Pow{A}$.
\begin{eqnarray*}
  \dep(i,h) = \dep(b,h) =\dep(\emptyset,h)&=& \emptyset\\
  \dep(\{l\},h) = \dep(\nott l,h)=     \dep(l,h) &=& h(l)\\
   \dep(\kempty(l),h) &=& h(l)\\
  \dep(l_1+l_2,h) =\dep(l_1\eq l_2,h) &=& h(l_1) \cup h(l_2)\\
  \dep(l_1\andd l_2,h) =\dep((l_1,l_2),h) &=& h(l_1) \cup h(l_2)\\
  \dep(l_1\cup l_2,h) &=& h(l_1) \cup h(l_2)
\end{eqnarray*}
Essentially, $\dep$ simply takes the union of the annotations of all
labels mentioned in a term.  

  \begin{tr}
    \citet{DBLP:conf/dbpl/CheneyAA07} defined dynamic
    provenance-tracking denotationally as a function $\DSB{e}$ mapping
    contexts $\gamma : \Var \to \annot{\Val}{\Pow{A}}$ to $\Pow{A}$-annotated
    values.  We present this definition in \refFig{dep-denot}.  Note
    that we use an auxiliary notation $v^{+a}$ to indicate adding an
    annotation to the toplevel of a $\Pow{A}$-annotated value.  That
    is, $(w^b)^{+a} = w^{b \cup a}$.

    \begin{figure*}[p]
      
      \begin{eqnarray*}
        \mflatten^D(\{w_1^{a_1}:m_1\ldots,w_n^{a_n}:m_n\})^a &=& (\mflatten(\{w_1:m_1,\ldots,w_n:m_n\}))^{a \cup a_1 \cup \cdots \cup a_n}\\
        \summ^D(\{w_1^{a_1}:m_1\ldots,w_n^{a_n}:m_n\})^a &=& (\summ(\{w_1:m_1,\ldots,w_n:m_n\}))^{a \cup a_1 \cup \cdots \cup a_n}
      \end{eqnarray*}
      \[\begin{array}{rclcrcl}
        \DSB{x}\gamma &=& \gamma(x)\\
        \DSB{\letin{x = e_1}{e_2}} &=& \DSB{e_2}\gamma[x := \DSB{e_1}\gamma]\\
        \DSB{i} \gamma &=& i^\emptyset\\
        \DSB{e_1+e_2} \gamma &=& \DSB{e_1}\gamma +^D \DSB{e_2}\gamma && w_1^{a_1} +^D w_2^{a_2} &=& (w_1 + w_2)^{a_1\cup a_2}\\
        \DSB{\summ \{e \mid x \in e_0\}}\gamma &=& \summ^D\{\DSB{e}\gamma[x\mapsto v] \mid v \in \DSB{e_0}\gamma\}
\\
        \DSB{b}\gamma &=& b^\emptyset\\
        \DSB{\nott e}\gamma &=& \nott^D\DSB{e}\gamma
&&
\nott^D(w^a) &=& (\nott w)^a\\
        \DSB{e_1 \andd e_2}\gamma &=& \DSB{e_1}\gamma \andd^D \DSB{e_2}\gamma
&&
w_1^{a_1} \andd^D w_2^{a_2} &=& (w_1 \andd w_2)^{a_1\cup a_2}\\
        \DSB{(e_1,e_2)}\gamma &=& (\DSB{e_1}\gamma ,\DSB{e_2}\gamma )^\emptyset\\
        \DSB{\pi_i(e)}\gamma &=& \pi_i(\getv{\DSB{e}\gamma})^{+\geta{\DSB{e}\gamma}}\\
        \DSB{\emptyset}\gamma &=& \emptyset^\emptyset\\
        \DSB{\setof{e}}\gamma &=& \{\DSB{e}\gamma\}^{\emptyset}\\
        \DSB{e_1 \cup e_2}\gamma &=& \DSB{e_1}\gamma\cup^D \DSB{e_2}\gamma
&&
 w_1^{a_1} \cup^D w_2^{a_2} &=& (w_1 \cup w_2)^{a_1\cup a_2}\\
        \DSB{\flatten\{e \mid x \in e_0\}}\gamma &=& \mflatten^D\{\DSB{e}\gamma[x\mapsto v] \mid v \in \DSB{e_0}\gamma\}
\\
        \DSB{\ifthenelse{e_0}{e_1}{e_2}}\gamma &=& \left\{ 
          \begin{array}{ll}
            \DSB{e_1}\gamma^{+\geta{\DSB{e_0}\gamma}} & \text{if $\SB{e_0}{\gamma} = \ktrue$}\\
            \DSB{e_2}\gamma^{+\geta{\DSB{e_0}\gamma}} & \text{if $\SB{e_0}{\gamma} = \kfalse$}
          \end{array}\right.\\
        \DSB{e_1 \eq e_2}\gamma &=& \DSB{e_1}\gamma\eq^D \DSB{e_2}\gamma
        &&
        w_1^{a_1} \eq^D w_2^{a_2} &=& (w_1 \eq w_2)^{a_1\cup a_2}\\
        \DSB{\kempty(e)}\gamma &=&\kempty^D(\DSB{e}\gamma)
        &&
        \kempty^D(w^a) &=& (\kempty(w))^a
      \end{array}\]
      
      \caption{Dependency-provenance, denotationally}
      \label{fig:dep-denot}
    \end{figure*}

    Next we introduce an operational version.  We define judgments
    $\heval[D]{\sigma}{h}{l}{e}{\sigma'}{h'}$ for expression
    evaluation and
    $\hevals[D]{\sigma}{h}{x}{L}{e}{\sigma'}{h'}{\annot{L'}{a}}$ for
    comprehension evaluation, both with dependency-provenance propagation.
    Note that the iteration rules maintain an annotation set $a$
    collecting the top-level annotations of the elements of $L'$.

    It is straightforward to prove by induction that:
    \begin{theorem}
      ~
      \begin{enumerate}
      \item Suppose $\wf{\Gamma}{e}{\tau}$ and
        $\wfstorectx{\Psi}{\sigma}{\gamma}{\Gamma}$.  Then
        $\heval[D]{\sigma}{h}{l}{e}{\sigma'}{h'}$ if and only if
        $\DSB{e} (\annot{\sigma}{h} \uparrow^{\Pow{A}}_\Gamma \gamma)
        = \annot{\sigma'}{h'} \uparrow^{\Pow{A}}_\tau l$.
      \item Suppose $\wf{\Gamma,x:\tau}{e}{\{\tau'\}}$ and
        $\wfstorectx{\Psi}{\sigma}{\gamma}{\Gamma}$.  Then
        $\hevals[D]{\sigma}{h}{x}{L}{e}{\sigma'}{h'}{\annot{L'}{a}}$
        if and only if $ \{\DSB{e}\gamma[x := v] \mid v \in
        \annot{\sigma}{h}\uparrow^{\Pow{A}}_{\{\tau\}} L\} =
        \annot{\sigma'}{h'}\uparrow^{\Pow{A}}_{\{\tau'\}} L' $ and $a
        = \cup \{\sigma(l') \mid l' \in L'\}$.
        \end{enumerate}
    \end{theorem}
\end{tr}
\begin{sub}
  In \refFig{dep-op} we define an equivalent operational formulation
  of dependency provenance, which we relate to the denotational
  presentation in \cite{tr}.  We define judgments
  $\heval[D]{\sigma}{h}{l}{e}{\sigma'}{h'}$ for expression evaluation
  and $\hevals[D]{\sigma}{h}{x}{L}{e}{\sigma'}{h'}{L'}$ for
  comprehension evaluation, both with where-provenance propagation.

\end{sub}
\begin{figure}[tb]
  \small
  \[
  \begin{array}{c}
    \infer{\heval[D]{\sigma}{h}{l}{t}{\sigma[l:=t]}{h[l:=\dep(t,h)]}}{}
    \smallskip\\
    \infer{\heval[D]{\sigma}{h}{l}{\letin{x=e_1}{e_2}}{\sigma''}{h''}}{
      \heval[D]{\sigma}{h}{l'}{e_1}{\sigma'}{h'} &
      \heval[D]{\sigma'}{h'}{l}{e_2[l'/x]}{\sigma''}{h''} & 
      l' \fresh
    }
    \smallskip\\
    \infer{\heval[D]{\sigma}{h}{l}{\pi_i(l')}{\sigma[l:=\sigma(l_i)]}{h[l:=h(l_i) \cup h(l')]}}{\sigma(l') = (l_1,l_2)}
    \smallskip\\
    \infer{\heval[D]{\sigma}{h}{l}{\ifthenelse{l'}{e_{\ktrue}}{e_{\kfalse}}}{\sigma'}{h'[l:= h'(l) \cup h'(l')]}}{
      \sigma(l') = b & \heval[D]{\sigma}{h}{l}{e_b}{\sigma'}{h'}}
    \smallskip\\
    \infer{\heval[D]{\sigma}{h}{l}{\flatten \{e \mid x \in l'\}}{\sigma'[l:=\mflatten \sigma'[L']]}{h'[l:=h'(l') \cup a]}}{\hevals[D]{\sigma}{h}{x}{\sigma(l)}{e}{\sigma'}{h'}{\annot{L'}{a}}}
    \smallskip\\
    \infer{\heval[D]{\sigma}{h}{l}{\sum \{e \mid x \in l'\}}{\sigma'[l:=\sum \sigma'[L']]}{h'[l:=h'(l') \cup a]}}{\hevals[D]{\sigma}{h}{x}{\sigma(l)}{e}{\sigma'}{h'}{\annot{L'}{a}}}
    \smallskip\\
    \infer{\hevals[D]{\sigma}{h}{x}{\emptyset}{e}{\sigma}{h}{\annot{\emptyset}{\emptyset}}}{}
    \smallskip\\
    \infer{\hevals[D]{\sigma}{h}{x}{L_1 \oplus L_2}{e}{\sigma_1 \uplus_\sigma \sigma_2}{h_1 \uplus _h h_2}{\annot{(L_1' \oplus L_2')}{a_1 \cup a_2}}}{\hevals[D]{\sigma}{h}{x}{L_1}{e}{\sigma_1}{h_1}{\annot{L_1'}{a_1}} & \hevals[D]{\sigma}{h}{x}{L_1}{e}{\sigma_2}{h_2}{\annot{L_2'}{a_2}}}
    \smallskip\\
    \infer{\hevals[D]{\sigma}{h}{x}{\{l:m\}}{e}{\sigma'}{h'}{\annot{\{l':m\}}{h'(l')}}}{\heval[D]{\sigma}{h}{l'}{e[l/x]}{\sigma'}{h'} & l' \fresh}
  \end{array}
  \]
  \caption{Dependency-provenance, operationally}
  \label{fig:dep-op}
\end{figure}

We define the dependency-provenance extraction judgments
$\dext{h}{T}{h'}$ and $\dexts{h}{\Ts}{h'}$ in
\refFig{dep-from-traces}.  As usual, we have two judgments, one for
traversing traces and another for traversing trace sets.
\begin{sub}
  \begin{theorem}
    Suppose $\treval{\sigma}{l}{e}{\sigma'}{T}$ and $h : \dom(\sigma)
    \to \Pow{A}$.  Then $\heval[D]{\sigma}{h}{l}{e}{\sigma'}{h'}$
    holds if and only if $\dext{h}{T}{h'}$ holds.
  \end{theorem}
\end{sub}
\begin{tr}
  \begin{theorem}
    \begin{enumerate}
    \item Suppose $\treval{\sigma}{l}{e}{\sigma'}{T}$ and $h :
      \dom(\sigma) \to \Pow{A}$.  Then
      $\heval[D]{\sigma}{h}{l}{e}{\sigma'}{h'}$ holds if and only if
      $\dext{h}{T}{h'}$ holds.
    \item If $\trevals{\sigma}{x}{L}{e}{\sigma'}{L'}{\Ts}$ and $h :
      \dom(\sigma) \to \Pow{A}$ then
      $\hevals[D]{\sigma}{h}{x}{L}{e}{\sigma'}{h'}{\annot{L'}{a}}$
      holds if and only if $\dexts{h}{\Ts}{\annot{h'}{a}}$ holds.

    \end{enumerate}

  \end{theorem}

\end{tr}
\begin{figure}
  \begin{eqnarray*}
    A_1 &=& \{r,s,r_1,r_2,r_3,s_1,s_2,s_3,r_{13},s_{11},r_{23},s_{21},r_{33},s_{31}\}
\\
A_2 &=& \{r,s,r_1,r_2,r_3,s_1,s_2,s_3,r_{12},r_{22},r_{32}\}
\\
A_3 &=& \{s_{11},s_{12},s_{21},s_{22},s_{31}\}
\end{eqnarray*}
\begin{tikzpicture}
  \node at (0,0) {A};
  \node at (1,0) {B};
  \node at (2,0) {D};
  
  \node (atom11) at (0,-0.5) [inner sep = 0pt] {$1$};
  \node (atom12) at (1,-0.5) [inner sep = 0pt] {$2$};
  \node (atom13) at (2,-0.5) [inner sep = 0pt] {$7$};
  \node (atom21) at (0,-1) [inner sep = 0pt] {$1$};
  \node (atom22) at (1,-1) [inner sep = 0pt] {$3$};
  \node (atom23) at (2,-1) [inner sep = 0pt] {$7$};

  \node (table) at (-0.3,0.3) [anchor=north west,rectangle,draw,minimum width = 3cm, minimum height
  = 1.6cm] {};
  
  \node (caption) at ([yshift=-0.4cm]table.south) {Output table $Q_1(A,B,D)$};

  \node (row1) at (-0.2,-0.3) [help lines,anchor=north west,rectangle,draw,minimum width = 2.8cm, minimum height
  = 0.4cm] {};
  \node (row2) at (-0.2,-0.8) [help lines,anchor=north west,rectangle,draw,minimum width = 2.8cm, minimum height
  = 0.4cm] {};

  \node (ctab) at ([yshift=+5mm]table.north) {\small $A_1$};
  \node (crow1) at ([xshift=-6mm]row1.west) {\small $\emptyset$};
  \node (crow2) at ([xshift=-6mm]row2.west) {\small $\emptyset$};
  \node (catom11) at (0.4,-0.45) {\small $r_{11}$};
  \node (catom12) at (1.4,-0.45) {\small $r_{12}$};
  \node (catom13) at (2.4,-0.45) {\small  $s_{22}$};
  \node (catom21) at (0.4,-0.95) {\small  $r_{21}$};
  \node (catom22) at (1.4,-0.95) {\small  $r_{22}$};
  \node (catom23) at (2.4,-0.95) {\small  $s_{22}$};

  \draw (table.north) -- (ctab.south);
  \draw (row1.west) -- (crow1);
  \draw (row2.west) -- (crow2);

  \node at (4,0) {C};
  \node at (5,0) {D};
  
  \node (atom1) at (4,-0.5) [inner sep = 0pt] {$42$};
  \node (atom2) at (5,-0.5) [inner sep = 0pt] {$7$};
  \node (atom3) at (4,-1) [inner sep = 0pt] {$42$};
  \node (atom4) at (5,-1) [inner sep = 0pt] {$7$};

  \node (table) at (3.7,0.3) [anchor=north west,rectangle,draw,minimum width = 2cm, minimum height
  = 1.6cm] {};
  
  \node (caption) at ([yshift=-0.4cm]table.south) {Output table $Q_2(C,D)$};

  \node (row1) at (3.8,-0.3) [help lines,anchor=north west,rectangle,draw,minimum width = 1.8cm, minimum height
  = 0.4cm] {};
  \node (row2) at (3.8,-0.8) [help lines,anchor=north west,rectangle,draw,minimum width = 1.8cm, minimum height
  = 0.4cm] {};

  \node (ctab) at ([yshift=+5mm]table.north) {\small $A_2$};
  \node (crow1) at ([xshift=-6mm]row1.west) {\small $\emptyset$};
  \node (crow2) at ([xshift=-6mm]row2.west) {\small $\emptyset$};
  \node (catom1) at (4.4,-0.45) {\small $\emptyset$};
  \node (catom2) at (5.4,-0.45) {\small $A_3$};
  \node (catom3) at (4.4,-0.95) {\small $r_{32}$};
  \node (catom4) at (5.4,-0.95) {\small $r_{31}$};

  \draw (table.north) -- (ctab.south);
  \draw (row1.west) -- (crow1);
  \draw (row2.west) -- (crow2);
\end{tikzpicture}%
\caption{Dependency provenance extraction examples}
\labelFig{dep-ext-examples}
\end{figure}

\begin{example}
  \refFig{dep-ext-examples} shows the results of dependency provenance
  extraction for \refExs{eg1}{eg2}.  The dependency-provenance is
  similar to the where-provenance for several fields such as $l_{11}$.
  The rows $l_1',l_2'$ have no (immediate) dependences.  The top-level
  labels $l,l'$ depend on many parts of the input --- essentially on
  all parts at which changes could lead to global changes to the
  output table.
\end{example}
\begin{figure}
  
  \[\small
  \begin{array}{c}
   \infer{\dext{h}{l \gets t}{h[l:= \dep(t,h)]}}{}
\quad
    \infer{\dext{h}{T_1;T_2}{h''}}{\dext{h}{T_1}{h'} & \dext{h'}{T_2}{h''}}
    \smallskip\\
    \infer{\dext{h}{l \gets \trproj{i}{l'}{l_i}}{h[l:=  h(l') \cup h(l_i)]}}{}
    \smallskip\\
%    \infer{\dext{h}{l \gets \trflat{l'}{L'}}{h[l:=  h(l') \cup \mflatten h[L']]}}{}
%    \smallskip\\
%    \infer{\dext{h}{l \gets \trsum{l'}{L'}}{h[l:= h(l') \cup \mflatten h[L']]}}{}
%    \smallskip\\
    \infer{\dext{h}{\trcond[l]{l'}{b}{T}}{h'[l' := h'(l') \cup h'(l)]}}{\dext{h}{T}{h'} }
    \smallskip\\
%    \infer{\dext{h}{l \gets \trmap{l'}{\Ts}}{h'}}{\dexts{h}{\Ts}{h'}}
%    \smallskip\\
\infer{\dext{h}{l \gets \trcomp{l}{\Ts}}{h'[l:=h'(l') \cup a]}}{\dexts{h}{\Ts}{\annot{h'}{a}}}
    \smallskip\\
\infer{\dext{h}{l \gets \trsum{l}{\Ts}}{h'[l:=h'(l') \cup a]}}{\dexts{h}{\Ts}{\annot{h'}{a}}}
\smallskip\\
    \infer{\dexts{h}{\emptyset}{\annot{h}{\emptyset}}}{}
    \quad
    \infer{\dexts{h}{\{[l]T\}}{\annot{h'}{h'(\result(T))}}}{\dext{h}{T}{h'}}
\smallskip\\
    \infer{\dexts{h}{\Ts_1\oplus \Ts_2}{\annot{(h_1 \uplus_h h_2)}{a_1 \cup a_2}}}{
      \dexts{h}{\Ts_1}{\annot{h_1}{a_1}}
      &
      \dexts{h}{\Ts_2}{\annot{h_2}{a_2}}
    }
  \end{array}
  \]
  \caption{Extracting dependency provenance}
\label{fig:dep-from-traces}
\end{figure}

\subsection{Semiring provenance}

\citet{DBLP:conf/pods/2007/GreenKT07} introduced the
\emph{semiring-annotated relational model}.  Recall that a
(commutative) semiring is an algebraic structure $(K,0_K,1_K,+_K,
\cdot_K)$ such that $(K,0,+)$ and $(K,1,\cdot)$ are commutative
monoids, $0$ is an annilhilator (that is, $0 \cdot x = 0 = x \cdot 0$)
and $\cdot$ distributes over $+$.  They considered $K$-relations to be
ordinary finite relations whose elements are annotated with elements
of $K$, and interpreted relational calculus queries over $K$-relations
such that many known variations of the relational model are a special
case.  For example, ordinary set-based semantics corresponds to the
semiring $(\Bool,\kfalse,\ktrue,\orr,\andd)$, whereas the multiset or
bag semantics corresponds to the semiring $(\Nat,0,1,+,\cdot)$.

The most general instance of the $K$-relational model is obtained by
taking $K$ to be the \emph{free semiring} $\Nat[X]$ of polynomials
with coefficients in $\Nat$ over indeterminates $X$, and
\citet{DBLP:conf/pods/2007/GreenKT07} considered this to yield a form
of provenance that they called \emph{how-provenance} because it
provides more information (than previous approaches such as
why-provenance or lineage) about how a tuple was derived from the
input.  Lineage and why-provenance can also be obtained as instances
of the semiring model (although the initial paper glossed over some
subtleties that were later clarified
by~\cite{DBLP:conf/pods/BunemanCTV08}).  Thus, if we can extract
semiring provenance from traces, we can also extract lineage and
why-provenance.

\citet{DBLP:conf/pods/FosterGT08} extended the semiring-valued model
to the NRC, and we will work in terms of this version.  Formally, given
semiring $K$, \citet{DBLP:conf/pods/FosterGT08} interpret types as
follows:
\begin{eqnarray*}
  \KSB{\intTy} &=& \Int \qquad  \KSB{\boolTy} = \Bool\\
  \KSB{\tau_1 \times \tau_2} &=& \KSB{\tau_1} \times \KSB{\tau_2}\\
  \KSB{\{\tau\}} &=& \{f : \KSB{\tau} \to K \mid \supp(f) \text{ finite}\}
\end{eqnarray*}
where $\supp(f) = \{x \in X \mid f(x) \neq 0_K\}$ provided $f : X \to
K$.  In other words, integer, boolean and pair types are interpreted
normally, and collections of type $\tau$ are interpreted as
\emph{finitely-supported} functions from $\KSB{\tau}$ to $K$.  For
example, finitely-supported functions $X \to \Bool$ correspond to
finite relations over $X$, whereas finitely-supported functions $X \to
\Nat$ correspond to finite multisets.  We overload the multiset
notation $\{v_1:k_1,\ldots\}$ for $K$-collections over $K$-values $v$
to indicate that the annotation of $v_i$ is $k_i$.  We write
$\KVal{K}$ for the set of all $K$-values of any type.

We write $\KK(X)$ for $\{f:X \to K \mid \supp(f) \text{ finite}\}$.
This forms an additive monad with zero.  To simplify notation, we
define its ``return'' ($\eta_\KK$), ``bind'' ($\bullet_\KK$), zero
($0_\KK$), and addition $(+_\KK)$ operators as follows:
\begin{eqnarray*}
  \eta_\KK(x) &=&\lambda y. \ifthenelse{x=y}{1_K}{0_K}\\
  f \bullet_\KK g &=& \lambda y.  \summ_{x \in \supp(f)} f(x) \cdot_K g(x)(y)\\
  0_\KK &=& \lambda x. 0_K\\
  f+_\KK g &=& \lambda x. f(x)+_K g(x)
\end{eqnarray*}
Moreover, if $f : X \to K$ and $k \in K$ then we write $k \cdot_\KK f$
for the ``scalar multiplication'' of $v$ by $k$, that is, $k \cdot f =
\lambda x. k \cdot_K f(x)$.

\citet{DBLP:conf/pods/FosterGT08} defined the semantics of NRC over
$K$-values denotationally. \refFig{semi-denot} presents a simplified
version of this semantics in terms of the $\KK$ monad operations; we
interpret an expression $e$ as a function from environments $\gamma :
\Var \to \KVal{K}$ to results in $\KVal{K}$.  Note that
\citet{DBLP:conf/pods/FosterGT08}'s version of NRC excludes emptiness
tests, integers, booleans and primitive operations other than
equality, but also includes some features we do not consider such as a
tree type used to model unordered XML.  Most of the rules are similar
to the ordinary denotational semantics of NRC; only the rules
involving collection types are different.  A suitable type soundness
theorem can be shown easily for this interpretation.

  \begin{figure}
    
    \begin{eqnarray*}
      \KSB{x}\gamma &=& \gamma(x)\\
      \KSB{\letin{x=e_1}{e_2}}\gamma&=& \KSB{e_2}\gamma[x \mapsto \KSB{e_1}\gamma]\\
      \KSB{b}\gamma &=& b\\
      \KSB{\nott e}\gamma &=& \nott\KSB{e}\gamma\\
      \KSB{e_1 \andd e_2}\gamma &=& \KSB{e_1}\gamma \andd \KSB{e_2}\gamma\\
      \KSB{(e_1,e_2)}\gamma &=& (\KSB{e_1}\gamma ,\KSB{e_2}\gamma )\\
      \KSB{\pi_i(e)}\gamma &=& \pi_i(\KSB{e}\gamma)\\
      \KSB{\emptyset}\gamma &=& 0_\KK\\
      \KSB{\setof{e}}\gamma &=& \eta_\KK(\KSB{e}\gamma)\\
      % \lambda x. \ifthenelse{x=\KSB{e}\gamma\\
%        &&\quad\quad}{1_K\\
%        &&\quad\quad}{0_K}\\
      \KSB{e_1 \cup e_2}\gamma &=& \KSB{e_1}\gamma +_\KK \KSB{e_2}\gamma \\
      \KSB{\flatten\{e \mid x \in e_0\}}\gamma &=&
      \KSB{e_0}\gamma \bullet_\KK (\lambda v. \KSB{e}\gamma[x\mapsto v])\\
% \lambda v'. \summ_K\{\KSB{e_0}\gamma v \cdot_K \KSB{e}\gamma[x\mapsto v]v' \\
%&&
%\quad\quad\mid v \in \supp(\KSB{e_0}\gamma)\}\\
      \KSB{\ifthenelse{e_0}{e_1}{e_2}}\gamma &=& \left\{ 
        \begin{array}{ll}
          \KSB{e_1}\gamma & \text{if $\KSB{e_0}{\gamma} = \ktrue$}\\
          \KSB{e_2}\gamma & \text{if $\KSB{e_0}{\gamma} = \kfalse$}
        \end{array}\right.\\
      \KSB{e_1 \eq e_2}\gamma &=& \left\{ 
        \begin{array}{ll}
          \ktrue & \text{if $\KSB{e_1}{\gamma} = \KSB{e_2}{\gamma}$}\\
          \kfalse& \text{if $\KSB{e_1}{\gamma} \neq \KSB{e_2}{\gamma}$}
        \end{array}\right.
    \end{eqnarray*}
      \caption{Semiring provenance, denotationally}
    \label{fig:semi-denot}
  \end{figure}

  Semiring-valued relations place annotations only on the elements of
  collections.  To model these annotations correctly using stores, we
  annotate labels of collections with $K$-collections of labels
  $\KK(\Lab)$.  As a simple example, consider store $[l_1:=1, l_2:=2,
  l_3:=1, l:=\{l_1:2,l_2:3,l_3\}]$ and annotation function $h(l) =
  [l_1:= k_1, l_2 := k_2,l_3:=k_3]$.  Then $l$ can be interpreted as
  the $K$-value $\{1 : 2k_1+k_3, 2 : 3k_2\}$.  The reason for
  annotating collections with $\KK(\Lab)$ instead of annotating
  collection element labels directly is that due to sharing, a label
  may be an element of more than one collection in a store (with
  different $K$-annotations). For example, consider $[l_1:=1, l_2:=2,l
  :=\{l_1:2,l_2\},l' := \{l_1:42\}]$.  If we annotate $l$ with $[l_1
  \mapsto k_1, l_2 \mapsto k_2]$ and $l'$ with $[l_1 := k_3]$ then we
  can interpret $l$ as $\{1:2k_1,2:k_2\}$ and $l'$ as $\{1:42k_3\}$
  respectively.  If the annotations were placed directly on $l_1,l_2$
  then this would not be possible.

  We will consider annotation functions $h : \Lab \to \KK(\Lab)_\bot$
  such that if $l$ is the label of a collection, then $h(l)$ maps the
  elements of $l$ to their $K$-values.  Labels of pair, integer, or
  boolean constructors are mapped to $\bot$.  In what follows, we
  will use an auxiliary function $\semiring(l,h)$ to deal with the
  basic operations:
\[\begin{array}{rcl}
  \semiring(l,h) &=& h(l)\\
  \semiring(\emptyset,h) &=& 0_\KK\\
  \semiring(\{l\},h) &=& \eta_\KK(l)\\
  \semiring(l_1\cup l_2,h) &=& h(l_1) +_\KK h(l_2)\\
  \semiring(t,h) &=& \bot \qquad (\text{otherwise})
\end{array}\]

As before, we consider an operational version of the denotational
semantics of NRC over $K$-values.  This is shown in \refFig{semi-op}.
As usual, there are two judgments, one for expression evaluation and
one for iterating over a set.  Many of the rules not involving
collections are standard.  The $\semiring$ function handles the cases
for $\emptyset$, $\cup$, and $\{e\}$.
\begin{sub}
  In the companion technical report~\cite{tr}, we define a translation from annotated stores to
  $K$-values and show that the operational version respects the
  denotational semantics.
\end{sub}

\begin{tr}
  There is a mismatch between the denotational semantics on $K$-values
  and the operational semantics.  The latter produces annotated
  stores, and we need to translate these to $K$-values in order to be
  able to relate the denotational and operational semantics.  The
  desired translation is different from the ones we have needed so
  far.  We define
  \begin{eqnarray*}
    \annot{\sigma}{h}\Uparrow_{\intTy}^K l &=& \sigma(l)\\
    \annot{\sigma}{h}\Uparrow_{\boolTy}^K l &=& \sigma(l)\\
    \annot{\sigma}{h}\Uparrow_{\tau_1 \times \tau_2}^K l &=& (\annot{\sigma}{h} \Uparrow_{\tau_1}^K l_1 ,\annot{\sigma}{h} \Uparrow_{\tau_1}^K l_1) \quad (\sigma(l) = (l_1,l_2))\\
    \annot{\sigma}{h}\Uparrow_{\{\tau\}}^K l&=& \lambda x. \summ\{h(l)(l') \\
&&\qquad \mid l' \in \dom(\sigma(l)), \annot{\sigma}{h}\Uparrow_{\{\tau\}}^K l' = x \}
  \end{eqnarray*}
  The translation steps for the basic types and pairing are
  straightforward.  For collection types, we need to construct a
  $K$-collection corresponding to $l$; to do so, given an input $x$ we
  sum together the values $h(l)(l')$ for each label $l'$ in
  $\dom(\sigma(l))$ such that the $K$-value of $l'$ in
  $\annot{\sigma}{h}$ is $x$.  In particular, note that we
  \emph{ignore} the multiplicity of $l'$ in $\sigma(l)$ here.

  We can now show the equivalence of the operational and denotational
  presentations of the semiring semantics:
  \begin{theorem}
      ~
      \begin{enumerate}
      \item Suppose $\wf{\Gamma}{e}{\tau}$ and
        $\wfstorectx{\Psi}{\sigma}{\gamma}{\Gamma}$.  Then
        $\heval[K]{\sigma}{h}{l}{e}{\sigma'}{h'}$ if and only if
        $\KSB{e} (\annot{\sigma}{h} \Uparrow^K_\Gamma \gamma)
        = \annot{\sigma'}{h'} \Uparrow^K_\tau l$.
        \item Suppose $\wf{\Gamma,x:\tau}{e}{\{\tau'\}}$ and
        $\wfstorectx{\Psi}{\sigma}{\gamma}{\Gamma}$.  Then 
        $\hevals[K]{\sigma}{h}{x}{L}{e}{\sigma'}{h'}{L'}$ if and only if
        $ \{\KSB{e}\gamma[x := v] \mid v \in
        \annot{\sigma}{h}\Uparrow^K_{\{\tau\}} L\} =
        \annot{\sigma'}{h'}\Uparrow^K_{\{\tau'\}} L' $.
        \end{enumerate}
  \end{theorem}
\end{tr}

  \begin{figure}
    \small
  \[
  \begin{array}{c}
    \infer{\heval[K]{\sigma}{h}{l}{t}{\sigma[l:=t]}{h[l:=\semiring(t,h)]}}{}
    \smallskip\\
    \infer{\heval[K]{\sigma}{h}{l}{\letin{x=e_1}{e_2}}{\sigma''}{h''}}{
      \heval[K]{\sigma}{h}{l'}{e_1}{\sigma'}{h'} &
      \heval[K]{\sigma'}{h'}{l}{e_2[l'/x]}{\sigma''}{h''} & 
      l' \fresh
    }
    \smallskip\\
    \infer{\heval[K]{\sigma}{h}{l}{\pi_i(l')}{\sigma[l:=\sigma(l_i)]}{h[l:=h(l_i)}}{\sigma(l') = (l_1,l_2)}
    \smallskip\\
    \infer{\heval[K]{\sigma}{h}{l}{\ifthenelse{l'}{e_{\ktrue}}{e_{\kfalse}}}{\sigma'}{h'}}{
      \sigma(l') = b & \heval[K]{\sigma}{h}{l}{e_b}{\sigma'}{h'}}
    \smallskip\\
    \infer{\heval[K]{\sigma}{h}{l}{\flatten \{e \mid x \in l'\}}{\sigma'[l:=\mflatten \sigma'[L']]}{h'[l:=k' \bullet_\KK h']}}{\hevals[K]{\sigma}{h}{x}{\annot{\sigma(l')}{h(l')}}{e}{\sigma'}{h'}{\annot{L'}{k'}}}
    \smallskip\\ 
    \infer{\hevals[K]{\sigma}{h}{x}{\annot{\emptyset}{k}}{e}{\sigma}{h}{\annot{\emptyset}{0_\KK}}}{}
    \smallskip\\
    \infer{\hevals[K]{\sigma}{h}{x}{\annot{(L_1 \oplus L_2)}{k}}{e}{(\sigma_1 \uplus_\sigma \sigma_2)}{h_1 \uplus _h h_2}{\annot{(L_1' \oplus L_2')}{k_1+_\KK k_2}}}{\hevals[K]{\sigma}{h}{x}{\annot{L_1}{k}}{e}{\sigma_1}{h_1}{\annot{L_1'}{k_1}} & \hevals[K]{\sigma}{h}{x}{\annot{L_2}{k}}{e}{\sigma_2}{h_2}{\annot{L_2'}{k_2}}}
    \smallskip\\
    \infer{\hevals[K]{\sigma}{h}{x}{\annot{\{l:m\}}{k}}{e}{\sigma'}{h'}{\annot{\{l':m\}}{k(l) \cdot \eta_\KK(l')}}}{\heval[K]{\sigma}{h}{l'}{e[l/x]}{\sigma'}{h'} & l' \fresh}
  \end{array}
  \]
    \caption{Semiring provenance, operationally}
    \label{fig:semi-op}
  \end{figure}

  Our main result is that extraction semantics is correct with respect
  to the operational semantics:
  \begin{sub}
    \begin{theorem}
      If $\treval{\sigma}{l}{e}{\sigma'}{T}$ then
      $\heval[K]{\sigma}{h}{l}{e}{\sigma'}{h'}$ holds if and only if
      $\kext{h}{T}{h'}$.
    \end{theorem}
  \end{sub}
  \begin{tr}
    \begin{theorem}
      \begin{enumerate}
      \item If $\treval{\sigma}{l}{e}{\sigma'}{T}$ then
        $\heval[K]{\sigma}{h}{l}{e}{\sigma'}{h'}$ holds if and only if
        $\kext{h}{T}{h'}$.
      \item If $\trevals{\sigma}{x}{L}{e}{\sigma'}{L'}{\Ts}$ then
        $\hevals[K]{\sigma}{h}{x}{\annot{L}{k}}{e}{\sigma'}{h'}{\annot{L'}{k'}}$
        if and only if $\kext{h,k}{\Ts}{h',k'}$.
      \end{enumerate}

    \end{theorem}
  \end{tr}

\begin{figure}
\begin{center}
\begin{tikzpicture}
  \node at (0,0) {A};
  \node at (1,0) {B};
  \node at (2,0) {D};
  
  \node (atom11) at (0,-0.5) [inner sep = 0pt] {$1$};
  \node (atom12) at (1,-0.5) [inner sep = 0pt] {$2$};
  \node (atom13) at (2,-0.5) [inner sep = 0pt] {$7$};
  \node (atom21) at (0,-1) [inner sep = 0pt] {$1$};
  \node (atom22) at (1,-1) [inner sep = 0pt] {$3$};
  \node (atom23) at (2,-1) [inner sep = 0pt] {$7$};

  \node (table) at (-0.3,0.3) [anchor=north west,rectangle,draw,minimum width = 3cm, minimum height
  = 1.6cm] {};
  
  \node (caption) at ([yshift=-0.4cm]table.south) {Output table $Q_1(A,B,D)$};

  \node (row1) at (-0.2,-0.3) [help lines,anchor=north west,rectangle,draw,minimum width = 2.8cm, minimum height
  = 0.4cm] {};
  \node (row2) at (-0.2,-0.8) [help lines,anchor=north west,rectangle,draw,minimum width = 2.8cm, minimum height
  = 0.4cm] {};

  \node (crow1) at ([xshift=6mm]row1.east) {\small $R_1 S_3$};
  \node (crow2) at ([xshift=6mm]row2.east) {\small $R_2S_3$};

  \draw (row1.east) -- (crow1);
  \draw (row2.east) -- (crow2);

  \node at (4,0) {A};
  \node at (5,0) {D};
  
  \node (atom1) at (4,-0.5) [inner sep = 0pt] {$1$};
  \node (atom2) at (5,-0.5) [inner sep = 0pt] {$7$};

  \node (table) at (3.7,0.3) [anchor=north west,rectangle,draw,minimum width = 2cm, minimum height
  = 1.1cm] {};
  
  \node (caption) at ([yshift=-0.9cm]table.south) {Output table $Q_3(A,D)$};

  \node (row1) at (3.8,-0.3) [help lines,anchor=north west,rectangle,draw,minimum width = 1.8cm, minimum height
  = 0.4cm] {};

  \node (crow1) at ([xshift=12mm]row1.east) {\small $R_1S_3+R_2S_3$};

  \draw (row1.east) -- (crow1);
\end{tikzpicture}%
\end{center}
\caption{Semiring provenance extraction examples}
\labelFig{semi-ext-examples}
\end{figure}
\begin{example}
  \refFig{semi-ext-examples} shows the result of semiring-provenance
  extraction on $Q_1$.  Here, we write $R_1,S_1$, etc. for the
  annotations of $r_1$ in $r$, $s_1$ in $s$, etc. respectively.  The
  second query $Q_2$ involves $\summ$ expressions, which are not
  handled by the semiring model.  Instead, the second part of
  \refFig{semi-ext-examples} shows the result of semiring provenance
  extraction on $Q_3 = \{(A:x.A,D:x.D) \mid x \in Q_1\}$, where we
  have merged the two copies of the record $(A:1,D:7)$ together and
  added their $K$-values.
\end{example}

\begin{figure}
  \[\small
  \begin{array}{c}
    \infer{\kext{h}{l\gets t}{h[l:=\semiring(t,h)]}}{}
    \quad
    \infer{\kext{h}{T_1;T_2}{h''}}{\kext{h}{T_1}{h'} & \kext{h'}{T_2}{h''}}
    \smallskip\\%\quad
    \infer{\kext{h}{l \gets \trproj{i}{l'}{l_i}}{h[l := h(l_i)]}}{}
    \quad
    \infer{\kext{h}{\trcond[l]{l'}{b}{T}}{h'}}{\kext{h}{T}{h'}}
    \smallskip\\%\quad
%    \infer{\kext{h}{l \gets \trflat{l'}{L}}{h[l := \lambda l''. \summ_{l_0 \in L}h(l')(l_0) \cdot h(l_0)(l'')]}}{}
%    \smallskip\\
%    \infer{\kext{h}{l\gets \trmap{l'}{\Ts}}{h'[l:=k]}}{\kexts{h,h(l')}{\Ts}{h',k}}
%    \quad
    \infer{\kext{h}{l\gets \trcomp{l'}{\Ts}}{h'[l:=k' \bullet_\KK h']}}{\kexts{h,h(l')}{\Ts}{h',k'}}
    \smallskip\\
    \infer{\kext{h,k}{\emptyset}{h,0_\KK}}{}
    \quad
    \infer{\kext{h,k}{\Ts_1\oplus \Ts_2}{h_1\uplus_h h_2, k_1 +_\KK k_2}}{\kext{h,k}{\Ts_1}{h_1,k_1} & \kext{h,k}{\Ts_2}{h_2,k_2}}
    \smallskip\\%\quad
    \infer{\kexts{h,k}{\{[l]T:m\}}{h',k(l) \cdot_\KK \eta_\KK(\result(T))}}{\kext{h}{T}{h'}}
  \end{array}
\]
  \caption{Extracting semiring provenance}
  \label{fig:semiring-from-traces}
\end{figure}

\section{Adaptation}\labelSec{adaptation}

\subsection{Adaptive semantics}

\begin{figure}
  
  \[\small
\begin{array}{c}
  \infer{\cp{\sigma}{l \gets t}{\sigma[l:=\op(t,\sigma)]}{l\gets t}}{}
\smallskip\\
\infer{\cp{\sigma}{\trlet{T_1}{T_2}}{\sigma''}{T_1';T_2'}}{
\cp{\sigma}{T_1}{\sigma'}{T_1'} & \cp{\sigma'}{T_2}{\sigma''}{T_2'}
}
\smallskip\\
\infer{\cp{\sigma}{l\gets \trproj{i}{l'}{l_i}}{\sigma[l:=l_i]}{l\gets \trproj{i}{l'}{l_i'}}}{\sigma(l') = (l_1',l_2')}
\smallskip\\
\infer{\cp{\sigma}{\trconde[l]{l'}{b}{T}{e_1}{e_2}}{\sigma'}{
\trconde[l]{l'}{b'}{T'}{e_1}{e_2}}}{b' = \sigma(l')\neq b & 
\treval{\sigma}{l}{e_{b'}}{\sigma'}{T'}}
\smallskip\\
\infer{\cp{\sigma}{\trconde[l]{l'}{b}{T}{e_1}{e_2}}{\sigma'}{
\trconde[l]{l'}{b}{T'}{e_1}{e_2}}}{\sigma(l')=b & 
\cp{\sigma}{T}{\sigma'}{T'}  & l = \result(T')}
\smallskip\\
%\infer{\cp{\sigma}{l\gets \trmape{l'}{\Ts}{x.e}}{\sigma'[l:=L']}{l\gets \trmape{l'}{\Ts'}{x.e}}}{
%\cps{\sigma}{x}{\sigma(l')}{\Ts}{e}{\sigma'}{L'}{\Ts'}
%}
%\smallskip\\
\infer{\cp{\sigma}{l\gets \trcompe{l'}{\Ts}{x.e}}{\sigma'[l:=\mflatten \sigma'[L']]}{l\gets \trcompe{l'}{\Ts'}{x.e}}}{
\cps{\sigma}{x}{\sigma(l')}{e}{\Ts}{\sigma'}{L'}{\Ts'}}
\smallskip\\
\infer{\cp{\sigma}{l\gets \trsume{l'}{\Ts}{x.e}}{\sigma'[l:=\summ \sigma'[L']]}{l\gets \trsume{l'}{\Ts'}{x.e}}}{
\cps{\sigma}{x}{\sigma(l')}{e}{\Ts}{\sigma'}{L'}{\Ts'}}
%\smallskip\\
%    \infer{\cp{\sigma}{l\gets \trflat{l'}{L}}{\sigma[l:=\flatten \sigma[L']]}{l\gets \trflat{l'}{L'}}}{
% \sigma(l') = L'
% }
% \smallskip\\
%   \infer{\cp{\sigma}{l\gets \trsum{l_0}{L}}{\sigma[l:=\summ \sigma[L']]}{l\gets \trsum{l_0}{L'}}}{\sigma(l_0) = L'}
\smallskip\\
\infer{\cps{\sigma}{x}{\emptyset}{e}{\Ts}{\sigma}{\emptyset}{\emptyset}}{}
\smallskip\\
\infer{\cps{\sigma}{x}{\{l:m\}}{e}{\Ts}{\sigma'}{\{\result(T'):m\}}{\{[l]T':m\}}}
{[l]T \in \Ts & \cp{\sigma}{T}{\sigma'}{T'}}
\smallskip\\
\infer{\cps{\sigma}{x}{\{l:m\}}{e}{\Ts}{\sigma'}{\{l':m\}}{\{[l]T':m\}}}
{l \notin \inputs(\Ts) & l' \fresh & \treval{\sigma}{l'}{e[l/x]}{\sigma'}{T'}}
\smallskip\\
\infer{\cps{\sigma}{x}{L_1 \oplus L_2}{e}{\Ts}{\sigma_1\smerge_\sigma \sigma_2}{L_1'\oplus L_2'}{\Ts_1\oplus \Ts_2}}
{\cps{\sigma}{x}{L_1}{e}{\Ts}{\sigma_1}{L_1'}{\Ts_1}
  &\cps{\sigma}{x}{L_2}{e}{\Ts}{\sigma_2}{L_2'}{\Ts_2}
}
\end{array}
  \]
  \caption{Trace adaptation semantics}
  \labelFig{cp}
\end{figure}

We also introduce an \emph{adaptive semantics} that adapts traces to
changes in the input.  Similarly to change-propagation in
AFL~\cite{acar06toplas}, we can use the adaptive semantics to
``recompute'' an expression when the input is changed, and to adapt
the trace to be consistent with the new input and output.  However,
unlike in AFL, our goal here is not to efficiently recompute results,
but rather to characterize how traces ``represent'' or ``explain''
computations.  We believe efficient techniques for recomputing
database queries could also be developed using similar ideas, but view
this as beyond the scope of this paper.

We define the adaptive semantics rules in \refFig{cp}.  Following the
familiar pattern established by the operational semantics, we use two
judgments: $\cp{\sigma}{T}{\sigma'}{T'}$, or ``Recomputing $T$ on
$\sigma$ yields result $\sigma'$ and new trace $T'$'', and
$\cps{\sigma}{x}{L}{e}{\Ts}{\sigma'}{L'}{\Ts'}$, or ``Reiterating $e$
on $\sigma$ for each $x \in L$ with cached traces $\Ts$ yields result
$\sigma'$, result labels $L'$, and new trace $\Ts'$''.

Many of the basic trace steps have straightforward adaptation rules.
For example, the rule for traces $l \gets t$ simply recomputes the
result using the values of the input labels in the current store.  For
projection, we recompute the operation and discard the cached labels.
Adaptation for sequential composition is also straightforward.  For
conditional traces, there are two rules.  If the boolean value of the
label is the same as that recorded in the trace, then we proceed by
re-using the subtrace.  Otherwise, we need to fall back on the trace
semantics to compute the other branch.

The rules for comprehension and summation traces make use of the
iteration adaptation judgment.  In each case, we traverse the current
store value of $l_0$.  For each label $l$ in this set, we re-compute
the body of the comprehension, re-using a trace $[l]T$ if present in
$\Ts$, otherwise evaluating $e[l/x]$ in the traced semantics.  The
iterative judgments return a new labeled trace set $\Ts$ and its
return labels $L'$.  Note that trace adaptation ignores the
multiplicity of cached traces.  When we re-use a cached trace $[l]T$ on
a label $l$ with multiplicity $m$, we simply rerun the trace and use
$m$ as the multiplicity of the result label and new trace.

\begin{example}
  TODO
\end{example}

% \begin{itemize}
% \item Type soundness (including traces?)
% \item Consistent with denotational semantics
% \item Trace Consistency
% \item Trace Fidelity
% \end{itemize}

\subsection{Metatheory of adaptation}\labelSec{metatheory}

  We now investigate the metatheoretic properties of the traced
  evaluation and trace adaptation semantics.  
% \begin{sub}
%  First, it is
%    straightforward to show that the trace semantics satisfies
%    appropriate type-soundness properties (see~\cite{tr}).
% \end{sub}
% \begin{tr}
%   We now investigate the metatheoretic properties of the traced
%   evaluation and trace adaptation semantics.  First, it is
%   straightforward to show that the trace semantics satisfies
%   appropriate type-soundness properties:
%   \begin{theorem}[Type soundness]
%     Suppose $\wf{\Gamma}{e}{\tau}$ and
%     $\wfstorectx{\Psi}{\sigma}{\gamma}{\Gamma}$.  Then there exist
%     $\Psi',\sigma',T$ such that
%     $\treval{\sigma}{l}{\gamma(e)}{\sigma'}{T}$ where $\Psi'(l) =
%     \tau$ and $\wfstore{\Psi'}{\sigma'}$.
%   \end{theorem}
% \end{tr}

  We first show that the traced semantics correctly implements the
  operational semantics of NRC expressions, if we ignore traces.  This
  is a straightforward induction in both directions.
\begin{theorem}
  For any $\sigma,l,e,\sigma'$, we have $\eval{\sigma}{l}{e}{\sigma'}$
  if and only if $\treval{\sigma}{l}{e}{\sigma'}{T}$ for some $T$.
\end{theorem}

We now turn to the correctness of the trace semantics.  We can view
the trace semantics as both evaluating $e$ in a store $\sigma$ yielding $\sigma'$ and
translating $e$ to a trace $T$ which ``explains'' the execution of
$e$.  What properties should a trace have in order to be a valid
explanation?  We identify two such properties which help to formalize
this intuition.  They are called \emph{consistency} and
\emph{fidelity}.

\paragraph{Consistency}

\begin{figure}
  
  \[
  \begin{array}{c}
    \infer{\trsat{\sigma}{l \gets t}}{\sigma(l) = \op(t,\sigma)}
\quad
\infer{\trsat{\sigma}{l\gets \trproj{i}{l'}{l_i}}}{\sigma(l') = (l_1,l_2) & \sigma(l) = \sigma(l_i)} 
\smallskip\\
\infer{\trsat{\sigma}{T_1;T_2}}{\trsat{\sigma}{T_1} & \trsat{\sigma}{T_2}}
\quad
\infer{\trsat{\sigma}{\trconde[l]{l'}{b}{T}{e_1}{e_2}}}
{\sigma(l') = b & \trsat{\sigma}{T} & \result(T) = l}
\smallskip\\
%\infer{\trsat{\sigma}{l \gets \trmape{l'}{\Ts}{x.e}}}
%{\sigma(l') = \inputs(\Ts) & \trsats{\sigma}{\Ts} & \sigma(l) = \outputs(\Ts)}
%\smallskip\\
\infer{\trsat{\sigma}{l \gets \trcompe{l'}{\Ts}{x.e}}}
{\sigma(l') = \inputs(\Ts) & \trsats{\sigma}{\Ts} & \sigma(l) = \mflatten \sigma[\outputs(\Ts)]}
\smallskip\\
\infer{\trsat{\sigma}{l \gets \trsume{l'}{\Ts}{x.e}}}
{\sigma(l') = \inputs(\Ts) & \trsats{\sigma}{\Ts} & \sigma(l) = \summ \sigma[\outputs(\Ts)]}
\smallskip\\
%\infer{\trsat{\sigma}{l \gets \trflat{l'}{L}}}
%{\sigma(l') = L & \sigma(l) = \flatten \sigma[L]}
%\quad
%\infer{\trsat{\sigma}{l \gets \trsum{l'}{L}}}
%{\sigma(l') = L & \sigma(l) = \summ \sigma[L]}
%\smallskip\\
\infer{\trsats{\sigma}{\emptyset}}{}
\quad
\infer{\trsats{\sigma}{\Ts_1 \oplus \Ts_2}}
{\trsats{\sigma}{\Ts_1} & \trsats{\sigma}{\Ts_2}}
\quad
\infer{\trsats{\sigma}{\{[l]T:m\}}}{\trsat{\sigma}{T}}
  \end{array}
\]
  \caption{Declarative semantics of traces}
  \label{fig:trace-declarative-sem}
\end{figure}

The trace is meant to be an explanation of what happened when $e$ was
evaluated on $\sigma$.  For example, if the trace says that $l \gets
l_1 + l_2$ but $\sigma'(l) \neq \sigma'(l_1) + \sigma'(l_2)$ then this
is inconsistent with the real execution.  Also, if the trace contains
$\trconde[l]{l'}{\kfalse}{T}{e_1}{e_2}$, but $l'$ actually evaluated
to $\ktrue$ in the evaluation of $e$, then the trace is inconsistent
with the actual execution.  As a third example, if the trace contains
$l' \gets \trcompe{l}{\{[l_1]T_1,[l_2]T_2\}}{x.e}$ whereas $\sigma(l)
= \{l_2,l_3\}$ then the trace is inconsistent because it does not
correctly show the behavior of the comprehension over $l$.

To formalize this notion of \emph{consistency}, observe that we can
view a trace declaratively as a collection of statements about the
values in the store.  We define a judgment $\trsat{\sigma}{T}$,
meaning ``$T$ is satisfied in store $\sigma$''.  We also employ an
auxiliary judgment $\trsats{\sigma}{\Ts}$, meaning ``Each trace in
$\Ts$ is satisfied in store $\sigma$''. The satisfiability relation is
defined in \refFig{trace-declarative-sem}.

\begin{theorem}[Consistency]
  If $\treval{\sigma}{l}{e}{\sigma'}{T}$ then $\trsat{\sigma'}{T}$.
\end{theorem}

\paragraph{Fidelity}

Consistency is a necessary, but not sufficient, requirement for traces
to be ``explanations''.  It tells us that the trace records valid
information about the results of an execution.  However, this is not
enough, in itself, to say that the trace really ``explains'' the
execution, because a consistent trace might not tell us what
might have happened in other possible executions.  To see why,
consider a simple expression $\ifthenelse{l_y}{l_x + l_z}{l_z}$ run
against input store $[l_x = 42,l_y=\ktrue,l_z=5\}$.  Consider the
traces, $T_1 = l \gets l_x+l_z$ and $T_2 = l \gets 47$.  Both of these
traces are consistent, but neither really ``explain'' what actually
happened.  Saying that $l \gets l_x + l_z$ or $l \gets 47$ is enough
to know what the result value was in the actual run, but not what the
result would have been under all conditions.  The dependence on $l_x$
is lost in $T_2$.  If we rerun $T_1$ with a different input store $l_x
= 37$, then $T_1$ will correctly return $42$ while $T_2$ will still
return $47$. Moreover, the dependences on $l_y$ are lost in both:
changing $l_y$ to $\kfalse$ invalidates both traces.  Instead, the
trace $T_3 = \trconde[l]{l_y}{\ktrue}{l \gets l_x+l_z}{l_x+l_z}{l_z}$
records enough information to recompute the result under \emph{any}
(reasonable) change to the input store.

We call traces \emph{faithful} to $e$ if they record enough
information to recompute $e$ when the input store changes.  We first
consider a property called \emph{partial fidelity}. Partial fidelity
tells us that the trace adaptation semantics is partially correct with
respect to the traced evaluation semantics.  That is, if $T$ was
obtained by running $e$ on $\sigma_1$ and we can successfully adapt
$T$ to a new input $\sigma_2$ to obtain $\sigma_2'$ and $T'$, then we
know that $\sigma_2'$ and $T'$ could also have been obtained by traced
evaluation from $\sigma_2$ ``from scratch''.  
\begin{tr}
 
We first need some lemmas:
\begin{lemma}\labelLem{fetch-trace}
  If $[l]T \in \Ts$ and $\trevals{\sigma}{x}{L}{e}{\sigma'}{L'}{\Ts}$
  then for some $\sigma''$ we have
  $\treval{\sigma}{\result(T)}{e[l/x]}{\sigma''}{T}$.
\end{lemma}
\begin{proof}
  Induction on the structure of
  $\trevals{\sigma}{x}{L}{e}{\sigma'}{L'}{\Ts}$.
  \begin{itemize}
  \item The case where $\Ts = \emptyset$ is vacuous since $[l]T \in
    \Ts$.

  \item Suppose the derivation is of the form
    \[
    \infer{\trevals{\sigma}{x}{L_1\cup L_2}{e}{\sigma_1 \uplus_\sigma
        \sigma_2}{L_1' \cup L_2'}{\Ts_1 \oplus \Ts_2}}
    {\trevals{\sigma}{x}{L_1}{e}{\sigma_1}{L_1' }{\Ts_1 } &
      \trevals{\sigma}{x}{L_2}{e}{\sigma_2}{L_2' }{\Ts_2 }}
    \]
    Then either $[l]T \in \Ts_1$ or $[l]T \in \Ts_2$; the cases are
    symmetric.  In either case, the induction hypothesis applies and
    we have $\treval{\sigma}{\result(T)}{e[l/x]}{\sigma_i}{T}$ as
    desired.
  \item Suppose the derivation is of the form
    \[
    \infer{\trevals{\sigma}{x}{\{l:m\}}{e}{\sigma'}{\{l':m\}}{\{[l]T:m\}}}{
      \treval{\sigma}{l'}{e[l/x]}{\sigma'}{T}}
    \]
    Then the subderivation $\treval{\sigma}{l'}{e[l/x]}{\sigma'}{T}$
    is the desired conclusion.
  \end{itemize}
\end{proof}
\begin{lemma}\labelLem{fetch-trace-wf}
  If $[l]T \in \Ts$ and $\wftrs{\Psi}{\tau}{\Ts}{\tau'}$ then we have
  $\wftr{\Psi,l{:}\tau}{T}{\result(T)}{\tau'}$.
\end{lemma}
\begin{proof}
  Straightforward induction similar to \refLem{fetch-trace}.
\end{proof}
\begin{lemma}\labelLem{cp-results-equal}
  If $\cp{\sigma}{T}{\sigma'}{T'}$ then $\result(T) = \result(T')$.
\end{lemma}
\begin{proof}
  Straightforward induction on derivations.
\end{proof}
\begin{theorem}[Partial fidelity]\labelThm{partial-fidelity}
  Let $\sigma_1,\sigma_1',\sigma_2,\sigma_2',T,T',\Theta,\Theta'$ be given.
  ~\begin{enumerate}
  \item If $\treval{\sigma_1}{l}{e}{\sigma_1'}{T}$ and
    $\cp{\sigma_2}{T}{\sigma_2'}{T'}$ then
    $\treval{\sigma_2}{l}{e}{\sigma_2'}{T'}$.
  \item If $\trevals{\sigma_1}{x}{L_1}{e}{\sigma_1'}{L_1'}{\Ts}$ and
    $\cps{\sigma_2}{x}{L_2}{e}{\Ts}{\sigma_2'}{L_2'}{\Ts'}$ then
    $\trevals{\sigma_2}{x}{L_2}{e}{\sigma_2'}{L_2'}{\Ts'}$
  \end{enumerate}
\end{theorem}
\begin{proof}
  Induction on the structure of the second derivation, with inversion
  on the first derivation.  \refLem{fetch-trace} is
  needed in part (2) to deal with the adaptation case where $[l]T \in
  \Ts$ holds.

  For part 1, the cases are as follows:
  \begin{itemize}
  \item If the second derivation is of the form
    \[\infer{\cp{\sigma_2}{l\gets t}{\sigma_2[l:=\op(t,\sigma_2)]}{l
        \gets t}}{}\]
    then the first must be of the form
    \[\infer{\treval{\sigma_1}{l}{t}{\sigma_1[l:=\op(t,\sigma_1)]}{l
        \gets t}}{}\]
    and so we can immediately conclude
    \[\infer{\treval{\sigma_2}{l}{t}{\sigma_2[l:=\op(t,\sigma_2)]}{l
        \gets t}}{}\]

  \item If the second derivation is  of the form
    \[\infer{\cp{\sigma_2}{l \gets
        \trproj{i}{l'}{l_i}}{\sigma_2[l:=\sigma_2(l_i')]}{l \gets
        \trproj{i}{l'}{l_i'}}}{\sigma_2(l') = (l_1',l_2')}\]
    then the first derivation is of the form
    \[\infer{\treval{\sigma_1}{l}{\pi_i(l')}{\sigma_1[l:=\sigma_1(l_i)]}{l
        \gets \trproj{i}{l'}{l_i}}}{\sigma_1(l') = (l_1,l_2)}\]
    and so we can immediately conclude
    \[\infer{\treval{\sigma_2}{l}{\pi_i(l')}{\sigma_2[l:=\sigma_2(l_i')]}{l
        \gets \trproj{i}{l'}{l_i'}}}{\sigma_2(l') = (l_1',l_2')}\]
  \item If the second derivation is of the form
    \[
    \infer{\cp{\sigma_2}{T_{11};T_{12}}{\sigma_2''}{T_{21};T_{22}}}
    {\cp{\sigma_2}{T_{11}}{\sigma_2'}{T_{21}} &
      \cp{\sigma_2'}{T_{12}}{\sigma_2''}{T_{22}}}\] 
    then the first derivation must be  of the form
    \[\infer{\treval{\sigma_1}{l}{\letin{x=e_1}{e_2}}{\sigma_1''}{T_{11};T_{12}}}
    {\treval{\sigma_1}{l'}{e_1}{\sigma_1'}{T_{11}} &
      \treval{\sigma_1'}{l}{e_2[l'/x]}{\sigma_1''}{T_{12}}}
    \]
Then by induction
    we have $\treval{\sigma_2}{l'}{e_1}{\sigma_2'}{T_{21}}$ and
    $\treval{\sigma_2'}{l}{e_2[l/x]}{\sigma_2''}{T_{22}}$, so can
    conclude
    \[\infer{\treval{\sigma_2}{l}{\letin{x=e_1}{e_2}}{\sigma_2''}{T_{21};T_{22}}}
    {\treval{\sigma_2}{l'}{e_1}{\sigma_2'}{T_{21}} &
      \treval{\sigma_2'}{l}{e_2[l'/x]}{\sigma_2''}{T_{22}} }\]

  \item If the second derivation is of the form
    \[
    \infer{\cp{\sigma_2}{\trconde[l]{l'}{b}{T_1}{e_\ktrue}{e_\kfalse}}{\sigma_2'}{\trconde[l]{l'}{b}{T_2}{e_\ktrue}{e_\kfalse}}}{
      \sigma_2(l) = b & \cp{\sigma_2}{T_1}{\sigma_2'}{T_2} }
    \]
then the first derivation must be of the form
    \[
    \infer{\treval{\sigma_1}{l}{\ifthenelse{l'}{e_\ktrue}{e_\kfalse}}{\sigma_1'}{\trconde[l]{l'}{b}{T_1}{e_\ktrue}{e_\kfalse}}}
    { \sigma_1(l') = b & \treval{\sigma_1}{l}{e_b}{\sigma_1'}{T_1} }
    \]
    We proceed by induction, obtaining
    $\treval{\sigma_2}{l}{e_b}{\sigma_2'}{T_2}$ and concluding
    \[
    \infer{\treval{\sigma_2}{l}{\ifthenelse{l'}{e_\ktrue}{e_\kfalse}}{\sigma_2'}{\trconde[l]{l'}{b}{T_2}{e_\ktrue}{e_\kfalse}}}{
      \sigma_2(l) = b & \treval{\sigma_2}{l}{e_b}{\sigma_2'}{T_2} }
    \]
\item 
    If the second derivation is of the form:
    \[
    \infer{\cp{\sigma_2}{\trconde[l]{l'}{b}{T_1}{e_\ktrue}{e_\kfalse}}{\sigma_2'}{\trconde[l]{l'}{b}{T_2}{e_\ktrue}{e_\kfalse}}}{
      b \neq \sigma_2(l) = b' &
      \treval{\sigma_2}{l}{e_{b'}}{\sigma_2'}{T_2} }
    \]
then again the first derivation must be of the form
    \[
    \infer{\treval{\sigma_1}{l}{\ifthenelse{l'}{e_\ktrue}{e_\kfalse}}{\sigma_1'}{\trconde[l]{l'}{b}{T_1}{e_\ktrue}{e_\kfalse}}}
    { \sigma_1(l') = b & \treval{\sigma_1}{l}{e_b}{\sigma_1'}{T_1} }
    \]
    and we may immediately conclude:
    \[
    \infer{\treval{\sigma_2}{l}{\ifthenelse{l'}{e_\ktrue}{e_\kfalse}}{\sigma_2'}{\trconde[l]{l'}{b'}{T_2}{e_\ktrue}{e_\kfalse}}}
    {\sigma_2(l) = b' & \treval{\sigma_2}{l}{e_{b'}}{\sigma_2'}{T_2}}
    \]
  \item If the second derivation is  of the form
    \[\small
    \begin{array}{c}
\infer{\cp{\sigma_2}{l \gets
        \trcompe{l'}{\Ts_1}{x.e}}{\sigma_2'[l:=\mflatten
        \sigma_2'[L_2]]}{l \gets \trcompe{l'}{\Ts_2}{x.e}}}
    {\cps{\sigma_2}{x}{\sigma_2(l')}{e}{\Ts_1}{\sigma_2'}{L_2}{\Ts_2}}
  \end{array}
    \]
    then the first derivation must be of the form
    \[
    \infer{\treval{\sigma_1}{l}{\flatten\{e \mid x \in
        l'\}}{\sigma_1'[l:=\mflatten \sigma'[L_1]]}{l \gets
        \trcompe{l'}{\Ts_1}{x.e}}}{
      \trevals{\sigma_1}{x}{\sigma_1(l')}{e}{\sigma_1'}{L_1}{\Ts_1} }
    \]
    By induction hypothesis (2), we have that
    $\trevals{\sigma_2}{x}{\sigma_2(l')}{e}{\sigma_2'}{L_2}{\Ts_2}$
    holds, so can conclude:
    \[
    \infer{\treval{\sigma_2}{l}{\flatten\{e \mid x \in
        l'\}}{\sigma_2'[l:=\mflatten \sigma'[L_2]]}{l \gets
        \trcompe{l'}{\Ts_2}{x.e}}}{\trevals{\sigma_2}{x}{\sigma_2(l')}{e}{\sigma_2'}{L_2}{\Ts_2}}
    \]
  \item If the second derivation is of the form
     \[\small
    \begin{array}{c}
\infer{\cp{\sigma_2}{l \gets
        \trsume{l'}{\Ts_1}{x.e}}{\sigma_2'[l:=\summ
        \sigma_2'[L_2]]}{l \gets \trsume{l'}{\Ts_2}{x.e}}}
    {\cps{\sigma_2}{x}{\sigma_2(l')}{e}{\Ts_1}{\sigma_2'}{L_2}{\Ts_2}}
  \end{array}
    \]
    the reasoning is similar to the previous case.
  \end{itemize}

  For part (2), the proof is by induction on the second derivation:
  \begin{itemize}
  \item If the derivation is of the form:
    \[\infer{\cps{\sigma_2}{x}{\emptyset}{e}{\Ts_1}{\sigma_2}{\emptyset}{\emptyset}}{}\]
    then we can immediately conclude
    \[\infer{\trevals{\sigma_2}{x}{\emptyset}{e}{\sigma_2}{\emptyset}{\emptyset}}{}\]
  \item If the derivation is of the form:
  \[\infer{\cps{\sigma_2}{x}{L_{21} \cup L_{22}}{e}{\Ts_1}{\sigma_{21} \uplus_{\sigma_2} \sigma_{22}}{L_{21}' \cup L_{22}'}{\Ts_{21} \cup \Ts_{22}}}
{
  \begin{array}{l}
    \cps{\sigma_2}{x}{L_{21}}{e}{\Ts_1}{\sigma_{21}}{L_{21}' }{\Ts_{21}}
    \\ \cps{\sigma_2}{x}{L_{22}}{e}{\Ts_1}{\sigma_{22}}{L_{22}' }{\Ts_{22}}
  \end{array}
}
\]
then we proceed by induction, concluding:
\[\infer{\trevals{\sigma_2}{x}{L_{21}\cup L_{22}}{e}{\sigma_{21} \uplus_{\sigma_2} \sigma_{22}}{L_{21}' \cup L_{22}'}{\Ts_{21} \cup \Ts_{22}}}{
  \begin{array}{l}
\trevals{\sigma_2}{x}{L_{21}}{e}{\sigma_{21}}{L_{21}'}{\Ts_{21}}  \\ 
\trevals{\sigma_2}{x}{L_{22}}{e}{\sigma_{22}}{L_{22}'}{\Ts_{22}}  
\end{array}
}
\]
\item If the derivation is of the form
  \[\infer{\cps{\sigma_2}{x}{\{l:m\}}{e}{\Ts_1}{\sigma_2'}{\{l':m\}}{\{[l]T_2:m\}}}
  { l \notin\inputs(\Ts_1) & l' \fresh &
    \treval{\sigma_2}{l'}{e[l/x]}{\sigma_2'}{T_2} }
  \]
  then we can immediately conclude:
  \[\infer{\trevals{\sigma_2}{x}{\{l:m\}}{e}{\sigma_2'}{\{l':m\}}{\{[l]T_2:m\}}}
  { \treval{\sigma_2}{l'}{e[l/x]}{\sigma_2'}{T_2} & l' \fresh }
  \]
\item If the derivation is of the form:
  \[\infer{\cps{\sigma_2}{x}{\{l:m\}}{e}{\Ts_1}{\sigma_2'}{\{\result(T_2):m\}}{\{[l]T_2:m\}}}
  { [l]T_1 \in \Ts_1 & \cp{\sigma_2}{T_1}{\sigma_2'}{T_2} }
  \]
  then observe that $\result(T_1) = \result(T_2)$ by
  \refLem{cp-results-equal}.  Moreover, by \refLem{fetch-trace}, we
  have $\treval{\sigma_1}{\result(T_1)}{e[l/x]}{\sigma_1''}{T_1}$, so
  by induction we have
  $\treval{\sigma_2}{\result(T_1)}{e[l/x]}{\sigma_2'}{T_2}$, and we
  can conclude
  \[
  \infer{\trevals{\sigma_2}{x}{\{l:m\}}{e}{\sigma_2'}{\{\result(T_2):m\}}{\{[l]T_2:m\}}}{
    \treval{\sigma_2}{\result(T_1)}{e[l/x]}{\sigma_2'}{T_2} }
  \]
\end{itemize}
\end{proof}
\end{tr}

\begin{sub}
  \begin{theorem}[Partial fidelity]\labelThm{partial-fidelity}
    If $\treval{\sigma_1}{l}{e}{\sigma_1'}{T}$ and
    $\cp{\sigma_2}{T}{\sigma_2'}{T'}$ then
    $\treval{\sigma_2}{l}{e}{\sigma_2'}{T'}$.
  \end{theorem}
\end{sub}

However, partial fidelity is rather weak since there is no guarantee
that $T$ can be adapted to a given $\sigma_2$.  To formalize and prove
total fidelity, we need to be careful about what changed inputs
$\sigma_2$ we consider.  Obviously, $\sigma_2$ must be type-compatible
with $T$ in some sense; for instance we cannot expect a trace such as
$l \gets l_1+l_2$ to adapt to an input in which $l_1 = \ktrue$.  Thus,
we need to set up a type system for stores and traces and prove
type-soundness for traced evaluation and adaptation.

More subtly, if we have a trace $l \gets t$ that writes to $l$ and we
try to evaluate it on a different store that \emph{already defines}
$l$, perhaps at a different type, then the adaptation step may
succeed, but the result store may be ill-formed, leading to problems
later on.  In general, we need to restrict attention to altered stores
$\sigma_2$ that \emph{preserve the types of labels read by $T$} and
\emph{avoid labels written by $T$}.  

We say that $\sigma$
\emph{matches $\Psi$ avoiding $S$} (written
$\matchavoids{\sigma}{\Psi}{S}$) if $\wfstore{\Psi'}{\sigma}$ for some
$\Psi' \supseteq \Psi$ with $\dom(\Psi') \cap S = \emptyset$.  That
is, $\sigma$ satisfies the type information in $\Psi$, and may have
other labels, but the other labels cannot overlap with $S$.
\begin{tr}
  Moreover, when $L$ is a collection of labels
  $\{l_1:m_1,\ldots,l_n:m_n\}$, we sometimes write $L{:}\tau$ as an
  abbreviation for $l_1:\tau,\ldots,l_n:\tau$; thus,
  $\matchavoids{\sigma}{\Psi,L{:}\tau}{S}$ stands for
  $\matchavoids{\sigma}{\Psi,l_1{:}\tau,\ldots,l_n{:}\tau}{S}$.
\end{tr}

We also need to be careful to avoid making the type system
too specific about the labels used internally by $T$, because these
may change when $T$ is adapted.  We therefore introduce a typing
judgment for traces $\wftr{\Psi}{T}{l}{\tau}$, meaning ``In a store
matching type $\Psi$, trace $T$ produces an output $l$ of type
$\tau$.''  Trace typing does not expose the types of labels created by
$T$ for internal use in the rules for let and comprehension.  The
rules are shown in \refFig{trace-wf}, along with the auxiliary
judgment $\wftrs{\Psi}{\tau}{\Ts}{\tau'}$, meaning ``In a store
matching $\Psi$, the labeled traces $\Ts$ operate on inputs of type
$\tau$ and produce outputs of type $\tau'$''.

\begin{figure}
  \[\small
  \begin{array}{c}
    \infer{\wftr{\Psi}{l \gets t}{l}{\tau}}{\wfterm{\Psi}{t}{\tau}}
    \quad%\smallskip\\
    \infer{\wftr{\Psi}{l \gets \trproj{i}{l'}{l_i}}{l}{\tau_i}}
    {\Psi(l') = \tau_1 \times \tau_2 %& \Psi(l_i) = \tau_i
    }
    \smallskip\\
    \infer{\wftr{\Psi}{T_1;T_2}{l}{\tau}}{\wftr{\Psi}{T_1}{l'}{\tau'} & \wftr{\Psi,l'{:}\tau'}{T_2}{l}{\tau}}
    \smallskip\\
    \infer{\wftr{\Psi}{\trconde[l]{l'}{b}{T}{e_\ktrue}{e_\kfalse}}{l}{\tau}}{
      \Psi(l') = \boolTy & \wftr{\Psi}{T}{l}{\tau}
      &
      \wf{\Psi}{e_\ktrue}{\tau}
      &
      \wf{\Psi}{e_\kfalse}{\tau}
    }
    \smallskip\\
    \infer{\wftr{\Psi}{l \gets \trcompe{l'}{\Ts}{x.e}}{l}{\{\tau\}}}{
      \Psi(l') = \{\tau'\} & 
      \wftrs{\Psi}{\tau'}{\Ts}{\{\tau\}} & 
      \wf{\Psi,x{:}\tau'}{e}{\{\tau\}}}
    \smallskip\\
    \infer{\wftr{\Psi}{l \gets \trsume{l'}{\Ts}{x.e}}{l}{\intTy}}{
      \Psi(l') = \{\tau'\} & 
      \wftrs{\Psi}{\tau'}{\Ts}{\intTy}& 
      \wf{\Psi,x{:}\tau'}{e}{\intTy}}
    \smallskip\\
    \infer{\wftrs{\Psi}{\tau}{\emptyset}{\tau'}}{}
    \quad%\smallskip\\
    \infer{\wftrs{\Psi}{\tau}{\{[l]T:m\}}{\tau'}}{ \wftr{\Psi,l{:}\tau}{T}{l'}{\tau'}}
    \smallskip\\
    \infer{\wftrs{\Psi}{\tau}{\Ts_1\oplus \Ts_2}{\tau'}}{\wftrs{\Psi}{\tau}{\Ts_1}{\tau'} & \wftrs{\Psi}{\tau}{\Ts_2}{\tau'}}
  \end{array}
  \]
  \caption{Trace well-formedness}\labelFig{trace-wf}
\end{figure}

We now show that for well-formed expressions and input stores, traced
evaluation can construct well-formed output stores and traces avoiding
any finite set of labels.  Here, we need label-avoidance constraints
to avoid label conflicts between $\sigma_1$ and $\sigma_2$ in the
$\Downarrow^\star$-rule for $\Ts_1 \oplus \Ts_2$.  We also need these
constraints later in proving \refThm{cp-exist}.  Next we show traced
evaluation is sound, that is, produces well-formed traces and states.

\begin{tr}
\begin{theorem}[Traceability]\labelThm{treval-exist}
  Let $S$ be a finite set of labels, and $\Psi,e,\tau,l,\sigma$ be
  arbitrary.
  \begin{enumerate}
  \item If $\wf{\Psi}{e}{\tau}$ and $\matchavoids{\sigma}{\Psi}{S \cup
      \{l\}}$ then there exists $\sigma',T$ such that
    $\treval{\sigma}{l}{e}{\sigma'}{T}$ and
    $\matchavoids{\sigma'}{\Psi,l{:}\tau}{S}$.
  \item  If $\wf{\Psi,x{:}\tau}{e}{\tau'}$ and
    $\matchavoids{\sigma}{\Psi,L{:}\tau}{S \cup L'}$ then there exists
    $\sigma',\Ts$ such that
    $\trevals{\sigma}{x}{L}{e}{\sigma'}{L'}{\Ts}$ and
    $\matchavoids{\sigma'}{\Psi,L'{:}\tau'}{S}$
  \end{enumerate}
\end{theorem}
\begin{proof}
  For part (1), proof is by induction on the structure of derivations
  of $\wf{\Psi}{e}{\tau}$.
  \begin{itemize}
  \item If the expression is a term $t$ then we have
    \[\infer{\wf{\Psi}{t}{\tau}}{\wfterm{\Psi}{t}{\tau}}\]
    Hence, $\wfcon{\Psi}{\op(\sigma,t)}{\tau}$ so
    \[\infer{\treval{\sigma}{l}{t}{\sigma[l:=\op(\sigma,t)]}{l \gets
        t}}{}\]
    where $\matchavoids{\sigma[l:=\op(\sigma,t)]}{\Psi,l{:}\tau}{S}$.
  \item If the derivation is of the form
    \[\infer{\wf{\Psi}{\pi_i(l')}{\tau_i}}{\wf{\Psi}{l'}{\tau_1\times\tau_2}}\]
    then we know $\wfcon{\Psi}{\sigma(l')}{\tau_1\times\tau_2}$ so we
    must have $\sigma(l') = (l_1,l_2)$.
    Hence, we can derive
    \[
    \infer{\treval{\sigma}{l}{\pi_i(l')}{\sigma[l:=\sigma(l_i)]}{l
        \gets \trproj{i}{l'}{l_i}}}{\sigma(l') = (l_1,l_2)}
    \]
    where $\matchavoids{\sigma[l:=\sigma(l_i)]}{\Psi,l{:}\tau_i}{S}$.
  \item If the derivation is of the form
\[\infer{\wf{\Psi}{\letin{x=e_1}{e_2}}{\tau}}{\wf{\Psi}{e_1}{\tau'} & \wf{\Psi,x{:}\tau'}{e_2}{\tau}}\]
then choose a fresh $l'\not\in \dom(\sigma) \cup S \cup \{l\}$. By
induction we have $\treval{\sigma}{l'}{e_1}{\sigma'}{T_1}$ where
$\matchavoids{\sigma'}{\Psi,l'{:}\tau'}{S \cup \{l\}}$.  Substituting
  $l'$ for $x$, we have $\wf{\Psi,l'{:}\tau'}{e_2[l/x]}{\tau}$ so by
  induction we also have
  $\treval{\sigma'}{l}{e_2[l'/x]}{\sigma''}{T_2}$ where $\matchavoids{\sigma''}{\Psi,l'{:}\tau',l{:}\tau}{S}$.  Finally we can derive
\[
\infer{\treval{\sigma}{l}{\letin{x=e_1}{e_2}}{\sigma''}{T_1;T_2}}{
l' \fresh & 
\treval{\sigma}{l'}{e_1}{\sigma'}{T_1} & 
\treval{\sigma'}{l}{e_2[l'/x]}{\sigma''}{T_2}}
\]
and $\matchavoids{\sigma}{\Psi,l{:}\tau}{S}$.
\item If the derivation is of the form
  \[\infer{\wf{\Psi}{\ifthenelse{l'}{e_\ktrue}{e_\kfalse}}{\tau}}{
    \Psi(l') = \boolTy & 
  \wf{\Psi}{e_\ktrue}{\tau} &
  \wf{\Psi}{e_\kfalse}{\tau}}
\]
then we must have $\sigma(l') = b \in \Bool$.  By induction, we obtain 
$\treval{\sigma}{l}{e_b}{\sigma'}{T}$ where $\matchavoids{\sigma'}{\Psi,l{:}\tau}{S}$.  Thus, we can conclude
\[
\infer{\treval{\sigma}{l}{\ifthenelse{l'}{e_\ktrue}{e_\kfalse}}{\sigma'}{\trconde[l]{l'}{b}{T}{e_\ktrue}{e_\kfalse}}}{
\sigma(l) = b & 
\treval{\sigma}{l}{e_b}{\sigma'}{T}
}
\]

\item If the derivation is of the form
\[
\infer{\wf{\Psi}{\flatten \{e \mid x\in l\}}{\{\tau\}}}
{\Psi(l) = \{\tau'\} & 
\wf{\Psi,x{:}\tau' }{e}{\{\tau\}}
}
\]
then we must have $\sigma(l) = L$ where $\wfcon{\Psi}{L'}{\{\tau'\}}$.
 Then there exist
$\sigma',L',\Ts$ such that
$\trevals{\sigma}{x}{\sigma(l)}{e}{\sigma'}{L'}{\Ts}$ and
$\matchavoids{\sigma }{\Psi,L'{:}\{\tau'\}}{\{l'\} \cup S}$.  Hence we can
conclude
\[
\infer{\treval{\sigma}{l'}{\flatten \{e \mid x\in l\}}{\sigma'[l':=\mflatten \sigma'[L']]}{l' \gets \trcompe{l}{\Ts}{x.e}} }{
\trevals{\sigma}{x}{\sigma(l)}{e}{\sigma'}{L'}{\Ts}
}
\]
and $\matchavoids{\sigma}{\Psi,l'{:}\{\tau'\}}{S}$.
\item The  case for $\summ\{e \mid x \in l\}$ is similar.
  \end{itemize}

For part (2), the proof is by induction on $L$:
\begin{itemize}
\item If $L = \emptyset$ then we can immediately conclude
\[\infer{\trevals{\sigma}{x}{\emptyset}{e}{\sigma}{\emptyset}{\emptyset}}{}\]
where $\matchavoids{\sigma}{\Psi}{S}$.
\item If $L = L_1 \oplus L_2$ then by induction we have
  $\trevals{\sigma}{x}{L_1}{e}{\sigma_1}{L_1'}{\Ts_1}$ where
  $\matchavoids{\sigma_1}{\Psi,L_1{:}\tau'}{S}$.  Moreover, we also
  have $\trevals{\sigma}{x}{L_2}{e}{\sigma_2}{L_2'}{\Ts_2}$ where
  $\matchavoids{\sigma_2}{\Psi,L_2{:}\tau'}{(\dom(\sigma_1)-\dom(\sigma))
    \cup S}$.  Thus, $\sigma_1 \uplus_\sigma \sigma_2$ exists and
  avoids $S$; hence,
\[
\infer{\trevals{\sigma}{x}{L_1 \oplus L_2}{e}{\sigma_1 \uplus_\sigma \sigma_2}{L_1' \oplus L_2'}{\Ts_1 \oplus \Ts_2}}{
\trevals{\sigma}{x}{L_1}{e}{\sigma_1}{L_1'}{\Ts_1}
&
\trevals{\sigma}{x}{L_2}{e}{\sigma_2}{L_2'}{\Ts_2}
}
\]
and $\matchavoids{\sigma_1 \uplus_\sigma \sigma_2}{\Psi,L_1 \cup
  L_2{:}\tau'}{S}$.
\item If $L = \{l:m\}$ then we can substitute to obtain
  $\wf{\Psi,l{:}\tau}{e[l/x]}{\tau'}$.  Choose $l'$ fresh for
  $\dom(\sigma) \cup S$ so that we have
  $\matchavoids{\sigma}{\Psi,l{:}\tau}{S \cup \{l'\}}$.  Then by
  induction we have $\treval{\sigma}{l'}{e[l/x]}{\sigma'}{T}$ where
  $\matchavoids{\sigma'}{\Psi,l{:}\tau,l'{:}\tau'}{S}$.  Then we can conclude
\[
\infer{\trevals{\sigma}{x}{\{l:m\}}{e}{\sigma'}{\{l':m\}}{\{[l]T:m\}}}
{l' \fresh & 
\treval{\sigma}{l'}{e[l/x]}{\sigma'}{T}
}
\]
since $\matchavoids{\sigma'}{\Psi,l'{:}\tau'}{S}$.
\end{itemize}
\end{proof}
\end{tr}

\begin{sub} % Summary of theorems in TR
\begin{theorem}[Traceability]\labelThm{treval-exist}
  Let $S$ be a finite set of labels, and $\Psi,e,\tau,l,\sigma$ be
  arbitrary.  If $\wf{\Psi}{e}{\tau}$ and
  $\matchavoids{\sigma}{\Psi}{S \cup \{l\}}$ then there exists
  $\sigma',T$ such that $\treval{\sigma}{l}{e}{\sigma'}{T}$ and
  $\matchavoids{\sigma'}{\Psi,l{:}\tau}{S}$.
\end{theorem}
\end{sub}

\begin{tr}
\begin{theorem}[Soundness of traced evaluation]\labelThm{treval-soundness}
  Let $\Psi,e,\tau,l,\sigma$ be arbitrary.
  \begin{enumerate}
  \item If $\wf{\Psi}{e}{\tau}$ and
    $\treval{\sigma}{l}{e}{\sigma'}{T}$ and $\sigma \matches \Psi$
    then $\wftr{\Psi}{T}{l}{\tau}$ and $\sigma' \matches
    \Psi,l{:}\tau$.
  \item If $\wf{\Psi,x{:}\tau}{e}{\tau'}$ and $\sigma \matches
    \Psi,L:\tau$ and $\trevals{\sigma}{x}{L}{e}{\sigma'}{L'}{\Ts}$
    then $\wftrs{\Psi}{\tau}{\Ts}{\tau'}$ and $\sigma' \matches
    \Psi,L':\tau'$.
  \end{enumerate}
\end{theorem}
\begin{proof}
  For part (1), proof is by induction on the second derivation.
  \begin{itemize}
  \item If the derivation is of the form
    \[
    \infer{\treval{\sigma}{l}{t}{\sigma[l:=\op(t,\sigma)]}{l\gets
        t}}{}
    \]
    then by inversion we have that $\wfterm{\Psi}{t}{\tau}$ and so we can
    derive
    \[\infer{\wftr{\Psi}{l\gets t}{l}{\tau}}{\wfterm{\Psi}{t}{\tau}}\]
  \item If the derivation is of the form
    \[
    \infer{\treval{\sigma}{l}{\pi_i l'}{\sigma[l:=\sigma(l_i)]}{l\gets
        \kproj_i(l',l_i)}}{\sigma(l') = (l_1,l_2)}
    \]
    then by inversion we have that $\Psi(l') = \tau_1 \times \tau_2$,
    so we may conclude:
    \[
    \infer{\wftr{\Psi}{l \gets
        \trproj{i}{l'}{l_i}}{l}{\tau_i}}{\Psi(l') = \tau_1 \times
      \tau_2}
    \]

  \item If the derivation is of the form
    \[
    \infer[l'
    \fresh]{\treval{\sigma}{l}{\letin{x=e_1}{e_2}}{\sigma_2}{\trlet{T_1}{T_2}}}
    {\treval{\sigma}{l'}{e_1}{\sigma_1}{T_1} &
      \treval{\sigma}{l}{e_2[l'/x]}{\sigma_2}{T_2} }
    \]
    then we must also have
    \[\infer{\wf{\Psi}{\letin{x=e_1}{e_2}}{\tau}}
    {\wf{\Psi}{e_1}{\tau'} & \wf{\Psi,x{:}\tau'}{e_2}{\tau}}
    \]
    and by induction and substituting $l'$ for $x$ we have
    $\wftr{\Psi}{T_1}{l'}{\tau'}$ and
    $\wftr{\Psi,l'{:}\tau'}{T_2}{l}{\tau}$.  So we may conclude
    \[\infer{\wftr{\Psi}{T_1;T_2}{l}{\tau}}{
      \wftr{\Psi}{T_1}{l'}{\tau'} &
      \wftr{\Psi,l'{:}\tau'}{T_2}{l}{\tau} }
    \]

  \item If the derivation is of the form:
    \[
    \infer{\treval{\sigma}{l}{\ifthenelse{l'}{e_\ktrue}{e_\kfalse}}{\sigma'}{\trconde[l]{l'}{b}{T}{e_\ktrue}{e_\kfalse}}}{\sigma(l')
      = b & \treval{\sigma}{l}{e_b}{\sigma'}{T} }
    \]
    then by inversion we must have
    \[
    \infer{\wf{\Psi}{\ifthenelse{l'}{e_\ktrue}{e_\kfalse}}{\tau}}{\Psi(l')
      = \boolTy & \wf{\Psi}{e_\ktrue}{\tau} &
      \wf{\Psi}{e_\kfalse}{\tau}}
    \]
    Hence whatever the value of $b$, by induction we can obtain
    $\wftr{\Psi}{T}{l}{\tau}$.  To conclude, we derive:
    \[\infer{\wftr{\Psi}{\trconde[l]{l'}{b}{T}{e_\ktrue}{e_\kfalse}}{l}{\tau}}
    {\Psi(l') = \boolTy & \wftr{\Psi}{T}{l}{\tau} &
      \wf{\Psi}{e_\ktrue}{\tau} & \wf{\Psi}{e_\kfalse}{\tau} }
    \]
  \item If the derivation is of the form
    \[ \infer{\treval{\sigma}{l}{\flatten\{e\mid x \in
        l'\}}{\sigma'[l:=\mflatten \sigma'[L']]}{l\gets
        \trcompe{l'}{\Ts}{x.e}}}
    {\trevals{\sigma}{x}{\sigma(l')}{e}{\sigma'}{L'}{\Ts}}\] then by
    inversion we have
    \[\infer{\wf{\Psi}{\flatten\{e\mid x \in l'\}}{\{\tau\}}}
    {\Psi(l') = \{\tau'\} & \wf{\Psi,x{:}\tau'}{e}{\{\tau\}}}
    \]
    Then by induction hypothesis (2) we have that
    $\wftrs{\Psi}{\tau'}{\Ts}{\{\tau\}}$, so we may conclude:
    \[\infer{\wftr{\Psi}{l \gets \trcompe{l'}{\Ts}{x.e}}{l}{\{\tau\}}}
    {\Psi(l') = \{\tau'\} & \wftrs{\Psi}{\tau'}{\Ts}{\{\tau\}}&
      \wf{\Psi,x{:}\tau'}{e}{\{\tau\}}}
    \]
  \item For the $\summ$ case,
    \[ \infer{\treval{\sigma}{l}{\summ\{e\mid x \in
        l'\}}{\sigma'[l:=\summ \sigma'[L']]}{l\gets
        \trsume{l'}{\Ts}{x.e}}}
    {\trevals{\sigma}{x}{\sigma(l')}{e}{\sigma'}{L'}{\Ts}}\] the
    reasoning is similar to the previous case.
  \end{itemize}

  For part (2), proof is by induction on the structure of the
  third derivation.
  \begin{itemize}
  \item If the derivation is of the form:
    \[
    \infer{\trevals{\sigma}{x}{\emptyset}{e}{\sigma}{\emptyset}{\emptyset}}{}\]
    then we can immediately derive
    \[\infer{\wftrs{\Psi}{\tau}{\emptyset}{\tau'}}{}\]
  \item If the derivation is of the form:
    \[\infer{\trevals{\sigma}{x}{\{l:m\}}{e}{\sigma'}{\{l':m\}}{\{[l]T:m\}}}{
      \treval{\sigma}{l'}{e[l/x]}{\sigma'}{T} }
    \]
    then we may substitute $l$ for $x$ to obtain
    $\wf{\Psi,l{:}\tau}{e[l/x]}{\tau'}$ and so by induction hypothesis
    (1) we have $\wftr{\Psi,l{:}\tau}{T}{l'}{\tau'}$.  We may conclude
    by deriving:
    \[
    \infer{\wftrs{\Psi}{\tau}{\{[l]T:m\}}{\tau'}}{\wftr{\Psi,l{:}\tau}{T}{l'}{\tau'}}
    \]

  \item If the derivation is of the form:
    \[ \infer{\trevals{\sigma}{x}{L_1\oplus L_2}{e}{\sigma_1
        \smerge_{\sigma} \sigma_2}{L_1' \oplus L_2'}{\Ts_1 \oplus
        \Ts_2}} {\trevals{\sigma}{x}{ L_1}{e}{\sigma_1}{L_1'}{\Ts_1} &
      \trevals{\sigma}{x}{ L_2}{e}{\sigma_2}{L_2'}{\Ts_2}} \] then by
    induction we obtain $\wftrs{\Psi}{\tau}{\Ts_1}{\tau'}$ and
    $\wftrs{\Psi}{\tau}{\Ts_2}{\tau'}$ so conclude
    \[
    \infer{\wftrs{\Psi}{\tau}{\Ts_1 \oplus \Ts_2}{\tau'} }{
      \wftrs{\Psi}{\tau}{\Ts_1}{\tau'} &
      \wftrs{\Psi}{\tau}{\Ts_2}{\tau'}}
    \]
  \end{itemize}
\end{proof}
\end{tr}

\begin{sub}
  \begin{theorem}[Soundness of traced evaluation]\labelThm{treval-soundness}
    Let $\Psi,e,\tau,l,\sigma$ be arbitrary.  If $\wf{\Psi}{e}{\tau}$ and
    $\treval{\sigma}{l}{e}{\sigma'}{T}$ and $\sigma \matches \Psi$
    then $\wftr{\Psi}{T}{l}{\tau}$ and $\sigma' \matches
    \Psi,l{:}\tau$.
  \end{theorem}
\end{sub}

We define the set of labels \emph{written} by $T$, or $\Wr(T)$, as follows:
\begin{eqnarray*}
  \Wr(l \gets t) &=& \{l\}\\
  \Wr(l \gets \trproj{i}{l'}{l_i}) &=& \{l\}\\
  \Wr(\trconde[l]{l'}{b}{T}{e_1}{e_2}) &=& \{l \} \cup \Wr(T)\\
  \Wr(T_1;T_2) &=& \Wr(T_1) \cup \Wr(T_2)\\
  \Wr(l \gets \trcompe{l'}{\Ts}{x.e}) &=& \{l\} \cup \Wr(\Ts)\\
  \Wr(l \gets \trsume{l'}{\Ts}{x.e}) &=& \{l\} \cup \Wr(\Ts)\\
  \Wr(\Ts) &=& \flatten \{\Wr(T) \mid [l]T :m \in \Ts\}
\end{eqnarray*}

Finally, we show that the adaptive semantics always succeeds for
well-formed traces $T$ and well-formed stores that avoid the labels
written by $T$.  
\begin{tr}
\begin{theorem}[Adaptability]\labelThm{cp-exist}
  Let $S$ be a finite set of labels, and $\Psi,T,\tau,l,\sigma$ be
  arbitrary.
  \begin{enumerate}
  \item If $\wftr{\Psi}{T}{l}{\tau}$ and $\matchavoids{\sigma}{\Psi}{S
      \cup \Wr(T)}$ then there exists $\sigma',T'$ such that
    $\cp{\sigma}{T}{\sigma'}{T'}$ and $
    \matchavoids{\sigma'}{\Psi,l{:}\tau}{S}$.
  \item If $\wftrs{\Psi}{\tau}{\Ts}{\tau'}$ and
    $\wf{\Psi,x{:}\tau}{e}{\tau'}$ and
    $\matchavoids{\sigma}{\Psi,L:\tau}{\Wr(\Ts) \cup S}$ then there
    exist $\sigma',L',\Ts'$ such that
    $\cps{\sigma}{x}{L}{e}{\Ts}{\sigma'}{L'}{\Ts'}$ and
    $\matchavoids{\sigma'}{\Psi,L'{:}\tau'}{S}$.
  \end{enumerate}
\end{theorem}
\begin{proof}
  For the first part, proof is by induction on the structure of the
  first derivation.
  \begin{itemize}
  \item If the derivation is of the form
    \[\infer{\wftr{\Psi}{l\gets t}{l}{\tau}}{\wfterm{\Psi}{t}{\tau}}
    \]
    then we can conclude
    \[\infer{\cp{\sigma}{l \gets t}{\sigma[l:=\op(t,\sigma)]}{l\gets
        t}}{}\]
    since $\sigma$ avoids $\Wr(l \gets t) = \{l\}$.  Moreover,
    $\matchavoids{\sigma}{\Psi,l{:}\tau}{S}$.
  \item If the derivation is of the form
    \[\infer{\wftr{\Psi}{l\gets
        \trproj{i}{l'}{l_i}}{l}{\tau_i}}{\Psi(l') = \tau_1 \times
      \tau_2}\]
    then $\sigma(l')$ must be a pair $(l_1',l_2')$, and we can
    conclude
    \[\infer{\cp{\sigma}{l \gets\trproj{i}{l'}{l_i}
      }{\sigma[l:=\sigma(l_i')]}{l\gets
        \trproj{i}{l'}{l_i'}}}{\sigma(l') = (l_1',l_2')}\]
    since $\sigma$ avoids $\Wr(l \gets \trproj{i}{l'}{l_i}) = \{l\}$.
    Note that we do not re-use $l_i$ so the typing judgment does not
    need to check that it is of the right type. In fact, $l_i$ need
    not be in $\Psi$ at all.  Finally,
    $\matchavoids{\sigma'}{\Psi,l{:}\tau_i}{S}$.

  \item If the derivation is of the form
    \[\infer{\wftr{\Psi}{T_1;T_2}{l}{\tau}}{
      \wftr{\Psi}{T_1}{l'}{\tau'} &
      \wftr{\Psi,l'{:}\tau'}{T_2}{l}{\tau} }
    \]
    then since $l' \in \Wr(T_1)$ and
    $\matchavoids{\sigma}{\Psi}{\Wr(T_1) \cup (\Wr(T_2) \cup S)}$, by
    induction we have that $\cp{\sigma}{T_1}{\sigma'}{T_1'}$ and
    $\matchavoids{\sigma'}{\Psi,l'{:}\tau'}{\Wr(T_2) \cup S}$.
    Moreover, since $\matchavoids{\sigma'}{\Psi,l'{:}\tau'}{\Wr(T_2)
      \cup S}$ by induction we have
    $\cp{\sigma'}{T_2}{\sigma''}{T_2'}$ and
    $\matchavoids{\sigma''}{\Psi,l'{:}\tau',l{:}\tau}{S}$.  Hence we
    may derive
    \[
    \infer{\cp{\sigma}{T_1;T_2}{\sigma''}{T_1';T_2'}}
    {\cp{\sigma}{T_1}{\sigma'}{T_1'} &
      \cp{\sigma'}{T_2}{\sigma''}{T_2'}}
    \]
    and also we have $\matchavoids{\sigma''}{\Psi,l{:}\tau}{S}$ as
    desired.

  \item If the derivation is of the form
    \[
    \infer{\wftr{\Psi}{\trconde[l]{l'}{b}{T}{e_\ktrue}{e_\kfalse}}{l}{\tau}}{
      \Psi(l') = \boolTy & \wftr{\Psi}{T}{l}{\tau} &
      \wf{\Psi}{e_\ktrue}{\tau} & \wf{\Psi}{e_\kfalse}{\tau} }
    \]
    then we must have $\sigma(l') \in \Bool$.  There are two cases.
    Suppose $\sigma(l) = b$.  Then by induction we have that
    $\cp{\sigma}{T}{\sigma'}{T'}$ and
    $\matchavoids{\sigma'}{\Psi,l{:}\tau}{S}$.  We can conclude
    \[
    \infer{\cp{\sigma}{\trconde[l]{l'}{b}{T}{e_\ktrue}{e_\kfalse}}{\sigma'}{\trconde[l]{l'}{b}{T'}{e_\ktrue}{e_\kfalse}}}
    {\sigma(l') = b & \cp{\sigma}{T}{\sigma'}{T'} }
    \]
    Otherwise, $\sigma(l') = b' \neq b$.  So using
    \refThm{treval-exist}, we have $\sigma',T'$ such that
    $\treval{\sigma}{l}{e_{b'}}{\sigma'}{T'}$ and
    $\matchavoids{\sigma'}{\Psi,l{:}\tau}{S}$, so we may conclude
    \[\infer{\cp{\sigma}{\trconde[l]{l'}{b}{T}{e_\ktrue}{e_\kfalse}}{\sigma'}{\trconde[l]{l'}{b}{T'}{e_\ktrue}{e_\kfalse}}}
    {\sigma(l') = b' \neq b & \treval{\sigma}{l}{e_{b'}}{\sigma'}{T'}}
    \]
  \item If the derivation is of the form
    \[
    \infer{\wftr{\Psi}{l \gets \trcompe{l'}{\Ts}{x.e}}{l}{\{\tau\}}}{
      \Psi(l') = \{\tau'\} & \wftrs{\Psi}{\tau'}{\Ts}{\{\tau\}} &
      \wf{\Psi,x{:}\tau'}{e}{\{\tau\}}}
    \]
    then for $L = \sigma(l')$, since $\wfcon{\Psi}{\sigma(l') }{
      \{\tau'\}}$ we have
    $\matchavoids{\sigma}{\Psi,L:\tau'}{\Wr(\Ts)\cup S}$.  Hence by
    induction we have $\sigma',L',\Ts'$ such that
    $\cps{\sigma}{x}{\sigma(l')}{e}{\Ts}{\sigma'}{L'}{\Ts'}$ and
    $\matchavoids{\sigma'}{\Psi,L':\{\tau\}}{S}$.  Therefore,
    $\mflatten\sigma'[L'] $ is well-defined so we can conclude
    \[\infer{\cp{\sigma}{l \gets
        \trcompe{l'}{\Ts}{x.e}}{\sigma'[l:=\mflatten \sigma'[L']]}{l
        \gets \trcompe{l'}{\Ts'}{x.e}}}
    {\cps{\sigma}{x}{\sigma(l')}{e}{\Ts}{\sigma'}{L'}{\Ts'}}
    \]
  \item If the derivation is of the form
    \[
    \infer{\wftr{\Psi}{l \gets \trsume{l'}{\Ts}{x.e}}{l}{\intTy}}{
      \Psi(l') = \{\tau'\} & \wftrs{\Psi}{\tau'}{\Ts}{\intTy}&
      \wf{\Psi,x{:}\tau'}{e}{\intTy}}
    \]
    then the reasoning is similar to the previous case.
  \end{itemize}

  For part (2), the proof is by induction on the structure of $L$.
  \begin{itemize}
  \item If $L = \emptyset$, then then we can simply conclude
    \[\infer{\cps{\sigma}{x}{\emptyset}{e}{\Ts}{\emptyset}{\emptyset}}{}\]

  \item If $L = \{l:m\}$ then there are two cases.  If $[l]T \in \Ts$
    for some $T$, then we proceed as follows. Let $l' = \result(T)$.
    By \refLem{fetch-trace-wf}, we have that
    $\wftr{\Psi,l{:}\tau}{e[l/x]}{l'}{\tau'}$.  So, by induction
    hypothesis (1), we have $\cp{\sigma}{T}{\sigma'}{T'}$ where
    $\matchavoids{\sigma'}{\Psi,l'{:}\tau'}{S}$.  To conclude, we
    derive:
    \[\infer{\cps{\sigma}{x}{\{l:m\}}{e}{\Ts}{\sigma'}{\{l':m\}}{\{[l]T':m\}}}
    {[l]T \in \Ts & \cp{\sigma}{T}{\sigma'}{T'}}
    \]

    Otherwise, $l \notin \inputs(\Ts)$, so we fall back on traced
    evaluation.  Choose $l'$ fresh for $l$, $\sigma$ and $S$. Since
    $\matchavoids{\sigma}{\Psi,l{:}\tau}{ S}$, by
    \refThm{treval-exist} we can obtain
    $\treval{\sigma}{l'}{e}{\sigma'}{T'}$ where
    $\matchavoids{\sigma}{\Psi,l'{:}\tau'}{S}$.  To conclude we derive
    \[
    \infer{\cps{\sigma}{x}{\{l:m\}}{e}{\Ts}{\sigma'}{\{l':m\}}{\{[l]T':m\}}}
    {l \not\in \inputs(\Ts) & l' \fresh &
      \treval{\sigma}{l'}{e[l/x]}{\sigma'}{T'}}
    \]

  \item If $L = L_1 \oplus L_2$, then clearly,
    $\matchavoids{\sigma}{\Psi,L_1{:}\tau}{\Wr(Ts)\cup S}$ so by
    induction we have
    $\cps{\sigma}{x}{L_1}{e}{\Ts}{\sigma_1}{L_1'}{\Ts_1}$ where
    $\matchavoids{\sigma_1}{\Psi,L_1'{:}\tau'}{S}$.  Similarly, we
    have $\cps{\sigma}{x}{L_2}{e}{\Ts}{\sigma_2}{L_2'}{\Ts_2}$ where
    $\matchavoids{\sigma_2}{\Psi,L_2'{:}\tau'}{(\dom(\sigma_1) -
      \dom(\sigma)) \cup S}$.  Hence, $\sigma_1$ and $\sigma_2$ are
    orthogonal extensions of $\sigma$, so $\sigma_1 \uplus_\sigma
    \sigma_2$ exists and $\matchavoids{\sigma_1 \uplus_\sigma
      \sigma_2}{\Psi,L_1' \cup L_2' {:}\tau'}{S}$.  We conclude by
    deriving:
    \[
    \infer{ \cps{\sigma}{x}{L_1\oplus L_2}{e}{\Ts}{\sigma_1
        \uplus_\sigma \sigma_2}{L_1' \oplus L_2'}{\Ts_1 \oplus \Ts_2}
    }{\cps{\sigma}{x}{L_1}{e}{\Ts}{\sigma_1}{L_1'}{\Ts_1} &
      \cps{\sigma}{x}{L_2}{e}{\Ts}{\sigma_2}{L_2'}{\Ts_2} }
    \]

  \end{itemize}
\end{proof}
\end{tr}

\begin{sub}
  \begin{theorem}[Adaptability]\labelThm{cp-exist}
    If $\wftr{\Psi}{T}{l}{\tau}$ and $\matchavoids{\sigma}{\Psi}{\Wr(T)}$ then there exists $\sigma',T'$ such that
    $\cp{\sigma}{T}{\sigma'}{T'}$ and $
    {\sigma'} \matches {\Psi,l{:}\tau}$.
  \end{theorem}
\end{sub}
By combining the above partial fidelity and soundness theorems, we can
finally obtain our main result:
\begin{corollary}[Total Fidelity]
  Suppose $\treval{\sigma_1}{l}{e}{\sigma_1'}{T_1}$ where $\sigma_1 :
  \Psi$ and $\wf{\Psi}{e}{\tau}$ and suppose
  $\matchavoids{\sigma_2}{\Psi}{\Wr(T)}$.  Then there exists
  $\sigma_2',T_2$ such that $\cp{\sigma_2}{T_1}{\sigma_2'}{T_2}$ and
  $\treval{\sigma_2}{l}{e}{\sigma_2'}{T_2}$.
\end{corollary}
\begin{proof}
  By \refThm{treval-soundness} we have that $\wftr{\Psi}{T_1}{l}{\tau}$.
  Thus, by \refThm{cp-exist} there must exist $T_2,\sigma_2'$ such that
  $\cp{\sigma_2}{T_1}{\sigma_2'}{T_2}$.  By \refThm{partial-fidelity},
  it follows that $\treval{\sigma_2}{l}{e}{\sigma_2'}{T_2}$.
\end{proof}

\section{Trace slicing}\labelSec{slicing}

As noted above, traces are often large.  Traces are also difficult to
interpret because they reduce computations to very basic steps, like
machine code.  In this section, we consider \emph{slicing} and other
simplifications for making trace information more useful and readable.
However, formalizing these techniques appears nontrivial, and is
beyond the scope of this paper.  Here we only consider examples of
trace slicing and simplification techniques that discard some of the
details of the trace information to make it more readable.

\begin{example}
  Recall query $Q_1$.  If we are only interested in how row $l_1$ in
  the output was computed, then the following \emph{backwards trace
    slice} answers this question.
\begin{verbatim}
l <- comp(r,{
  [r1] x11 <- proj_C(r1,r13); x1 <- comp(s,{
    [s3] x131 <- proj_C(s3,s31); x132 <- x11 = x131; 
         cond(x132,t,l11 <- proj_A(r1,r11);
                     l12 <- proj_B(r1,r12);
                     l13 <- proj_D(s3,s32);
                     l1 <- (A:l11,B:l12,D:l13);
                     x136 <- {l1})})})
\end{verbatim}
  Note that the slice refers only to the rows $r_1$ and $s_3$ that
  contribute to the semiring-provenance of $l_1$.  Moreover, the
  where-provenance and dependency-provenance of $l_1,l_{11},l_{12},$
  and $l_{13}$ can be extracted from this slice.

  To make the slice more readable, we can discard information about
  projection and assignment steps and substitute expressions for
  labels:
\begin{verbatim}
l <- comp(r,{
  [r1] x1 <- comp(s,{
    [s3] cond(r13 = s31,t,l1 <- (A:r11,B:r12,D:s32);
                          x136 <- {l1})})})
\end{verbatim}
  We can further simplify this to an expression
  $\{(A:r_{11},B:r_{12},D:s_{32})\}$ that shows how to calculate $l_1$
  from the original input, but this is not guaranteed to be valid if
  the input is changed.
\end{example}

\begin{example}
  In query $Q_2$, if we are only interested in the value $7$ labeled
  by $l_{12}'$, its (simplified) backwards trace slice is:
\begin{verbatim}
l12' <- sum(s,{[s1] cond(s11 = 2, t, x13 <- s12),
               [s2] cond(s12 = 2, t, x23 <- s22),
               [s3] cond(s13 = 2, f, x33 <- 0)});
\end{verbatim}
and from this we can extract an expression such as $s_{12}+s_{22}$ that
describes how the result was computed.
\end{example}

\section{Related and future work}\labelSec{related}

Provenance has been studied for database queries under various names,
including ``source tagging'' and ``lineage''.  We have already
discussed where-provenance, dependency provenance and the semiring
model.  \citet{DBLP:conf/vldb/WangM90} described an early provenance
semantics meant to capture the original and intermediate sources of
data in the result of a query.  Cui, Widom and Wiener defined
\emph{lineage}, which aims to identify source data relevant to part of
the output.  \citet{buneman01icdt} also introduced
\emph{why--provenance}, which attempts to highlight parts of the input
that explain why a part of the output is the way it is.  As discussed
earlier, lineage and why-provenance are instances of the semiring
model. Recently, \citet{benjelloun06vldb} have studied a new form of
lineage in the Trio system. According to Green (personal
communication), Trio's lineage model is also an instance of the
semiring model, so can also be extracted from traces.

\citet{DBLP:conf/sigmod/BunemanCC06} and \citet{buneman07icdt}
investigated provenance for database updates, an important
scenario because many scientific databases are \emph{curated}, or
maintained via frequent manual updates. Provenance is essential for
evaluating the scientific value of curated
databases~\cite{DBLP:conf/pods/BunemanCTV08}.  We have not considered
traces for update languages in this paper.  This is an important
direction for future work.

Provenance has also been studied in the context of \emph{(scientific)
  workflows}, that is, high-level visual programming languages and
systems developed recently as interfaces to complex distributed Grid
computation.  Techniques for workflow provenance are surveyed by
\citet{bose05cs} and \citet{DBLP:journals/sigmod/SimmhanPG05}.  Most
such systems essentially record call graphs including the names and
parameters of macroscopic computation steps, input and output
filenames, and other system metadata such as architecture, operating
system and library versions.  Similarly, provenance-aware storage
systems~\cite{muniswamy-reddy06usenix} record high-level trace
information about files and processes, such as the files read and
written by a process.

To our knowledge formal semantics have not been developed for most
workflow systems that provide provenance tracking.  Many of them
involve concurrency so defining their semantics may be nontrivial.
One well-specified approach is the NRC-based ``dataflow'' model
of~\cite{DBLP:conf/dils/HiddersKSTB07}, who define an instrumented
semantics that records ``runs'' and consider extracting provenance
from runs.  However, their formalization is incomplete and does not
examine semantic correctness properties comparable to consistency and
fidelity; moreover, they have not established the exact relationship
between their runs and existing forms of provenance.

As discussed in the introduction, provenance traces are related to the
traces used in the adaptive functional programming language
AFL~\cite{acar06toplas}.  The main difference is that AFL traces are
meant to model efficient self-adjusting computation implementations,
whereas provenance traces are intended as a model of execution history
that can be used to answer high-level queries comparable to other
provenance models.  Nevertheless, efficiency is obviously an important
issue for provenance-tracking techniques.  The problem of efficiently
recomputing query results after the input changes, also called
\emph{view maintenance}, has been studied extensively for
\emph{materialized views} (cached query results) in relational
databases~\cite{gupta95maintenance}.  View maintenance does not appear
to have been studied in general for NRC, but provenance traces may
provide a starting point for doing so.  View maintenance in the
presence of provenance seems to be an open problem.

Provenance traces may also be useful in studying the \emph{view
  update} problem for NRC queries, that is, the problem of updating
the input of a query to accommodate a desired change to the output.
This is closely related to bidirectional computation techniques that
have been developed for XML trees~\cite{foster07toplas}, flat relational
queries~\cite{bohannon06pods}, simple
functional programs~\cite{matsuda07icfp},  and text
processing~\cite{DBLP:conf/popl/BohannonFPPS08}.  Provenance-like
metadata has already been found useful in some of this work.  Thus, we
believe that it will be worthwhile to further study the relationship
between provenance traces and bidirectional computation.

There is a large body of related work on dynamic analysis techniques,
including slicing, debugging, justification, information flow,
dependence tracking, and profiling techniques, in which execution
traces play an essential role.  We cannot give a comprehensive
overview of this work here, but refer
to~\cite{venkatesh91pldi,arora93dood,abadi96icfp,field98ist,abadi99popl,ochoa04pepm}
as sources we found useful for inspiration.  However, to our
knowledge, none of these techniques have been studied in the context
of database query languages, and our work reported previously in
\cite{DBLP:conf/dbpl/CheneyAA07} and in this paper is the first to
connect any of these topics to provenance.

Trace semantics is also employed in static analysis; in particular,
see~\cite{rival07toplas}.  \citet{DBLP:conf/dbpl/CheneyAA07} defined a
type-and-effect-style static analysis for dependency provenance; to
our knowledge, there is no other prior work on using static analysis
to approximate provenance or optimize dynamic provenance tracking.

\section{Conclusions}\labelSec{concl}

Provenance is an important topic in a variety of settings,
particularly where computer systems such as databases are being used
in new ways for scientific research.  The semantic foundations of
provenance, however, are not well understood.  This makes it difficult
to judge the correctness and effectiveness of existing proposals and to
study their strengths and weaknesses.

This paper develops a foundational approach based on \emph{provenance
  traces}, which can be viewed as explanations of the operational
behavior of a query not on just the current input but also on other
possible (well-defined) inputs.  We define and give traced operational
semantics and adaptation semantics for traces and prove
\emph{consistency} and \emph{fidelity} properties that characterize
precisely how traces produced by our approach record the run-time
behavior of queries.  The proof of fidelity, in particular, involves
subtleties not evident in other trace semantics systems such as
AFL~\cite{acar06toplas} due to the presence of collection types and
comprehensions, which are characteristic of database query languages.

Provenance traces are very general, as illustrated by the fact that
other forms of provenance information may be extracted from them.  For
instance, we show how to extract where-provenance, dependency
provenance, and semiring provenance from traces.  Depending on the
needs of the application, these specialized forms of provenance may be
preferable to provenance traces due to efficiency concerns.  As a
further application, we informally discuss how we may slice or
simplify traces to extract smaller traces that are more relevant to
part of the input or output.

To our knowledge, our work is the first to formally investigate trace
semantics for collection types or database query languages and the
first to relate traces to other models of provenance in databases.
There are a number of compelling directions for future work, including
formalizing interesting definitions of trace slices, developing
efficient techniques for generating and querying provenance traces,
and relating provenance traces to the view-maintenance and view-update
problems.

\paragraph{Acknowledgments}
We gratefully acknowledge travel support from the UK e-Science
Institute Theme Program on Principles of Provenance for visits by Acar
to the University of Edinburgh and Cheney to Toyota Technological
Institute, Chicago.

\small
\bibliography{paper}
\bibliographystyle{plainnat}
%\bibliographystyle{plain}

% LocalWords:  Umut Acar Amal hoc MapReduce PigLatin semiring subtrace SQL AFL
% LocalWords:  nrc metatheory metatheoretic concl monad tuples multisets eg wf
% LocalWords:  multiset boolean formedness typechecked typechecking opsem unary
% LocalWords:  denotationally nontermination typings typings pointwise sep th
% LocalWords:  unioning subtraces yshift ctab xshift catom cp booleans
% LocalWords:  declaratively satisfiability sem eval coinductive subderivation
% LocalWords:  fixpoint fixpoints recomputation alg subexpression analyses dep
% LocalWords:  Karvounarakis Tannen monoids annilhilator indeterminates rcl Cui
% LocalWords:  Madnick Widom buneman icdt dependences workflows silico Frew et
% LocalWords:  nondeterministic Simmhan al filenames workflow dataflow acar IH
% LocalWords:  recordkeeping popl denot versa toplevel rclcrcl tuple TODO vldb
% LocalWords:  benjelloun metadata treval postprocessing bslice toplas inlining
% LocalWords:  bose

%% file: abstract.tex
\begin{abstract}
  Provenance is information about the origin, derivation, ownership,
  or history of an object.  It has recently been studied extensively
  in scientific databases and other settings due to its importance in
  helping scientists judge data validity, quality and integrity.
  However, most models of provenance have been stated as ad hoc
  definitions motivated by informal concepts such as ``comes from'',
  ``influences'', ``produces'', or ``depends on''.  These models lack
  clear formalizations describing in what sense the definitions
  capture these intuitive concepts.  This makes it difficult to
  compare approaches, evaluate their effectiveness, or argue about
  their validity.

  We introduce \emph{provenance traces}, a general form of provenance
  for the \emph{nested relational calculus} (NRC), a core database
  query language.  Provenance traces can be thought of as concrete
  data structures representing the operational semantics derivation of
  a computation; they are related to the traces that have been used in
  self-adjusting computation, but differ in important respects. We
  define a tracing operational semantics for NRC queries that produces
  both an ordinary result and a trace of the execution. We show that
  three pre-existing forms of provenance for the NRC can be extracted
  from provenance traces.  Moreover, traces satisfy two semantic
  guarantees: \emph{consistency}, meaning that the traces describe
  what actually happened during execution, and \emph{fidelity},
  meaning that the traces ``explain'' how the expression would behave
  if the input were changed.  These guarantees are much stronger than
  those contemplated for previous approaches to provenance; thus,
  provenance traces provide a general semantic foundation for
  comparing and unifying models of provenance in databases.
\end{abstract}

%% file: paper-tr.bbl
\begin{thebibliography}{32}
\providecommand{\natexlab}[1]{#1}
\providecommand{\url}[1]{\texttt{#1}}
\expandafter\ifx\csname urlstyle\endcsname\relax
  \providecommand{\doi}[1]{doi: #1}\else
  \providecommand{\doi}{doi: \begingroup \urlstyle{rm}\Url}\fi

\bibitem[Abadi et~al.(1996)Abadi, Lampson, and L\'evy]{abadi96icfp}
Mart\'in Abadi, Butler Lampson, and Jean-Jacques L\'evy.
\newblock Analysis and caching of dependencies.
\newblock In \emph{ICFP}, pages 83--91. ACM Press, 1996.

\bibitem[Abadi et~al.(1999)Abadi, Banerjee, Heintze, and Riecke]{abadi99popl}
Mart\'in Abadi, Anindya Banerjee, Nevin Heintze, and Jon~G. Riecke.
\newblock A core calculus of dependency.
\newblock In \emph{POPL}, pages 147--160. ACM Press, 1999.

\bibitem[Acar et~al.(2006)Acar, Blelloch, and Harper]{acar06toplas}
Umut~A. Acar, Guy~E. Blelloch, and Robert Harper.
\newblock Adaptive functional programming.
\newblock \emph{ACM Trans. Program. Lang. Syst.}, 28\penalty0 (6):\penalty0
  990--1034, 2006.

\bibitem[Arora et~al.(1993)Arora, Ramakrishnan, Roth, Seshadri, and
  Srivastava]{arora93dood}
Tarun Arora, Raghu Ramakrishnan, William~G. Roth, Praveen Seshadri, and Divesh
  Srivastava.
\newblock Explaining program execution in deductive systems.
\newblock In \emph{Deductive and Object-Oriented Databases}, pages 101--119,
  1993.

\bibitem[Benjelloun et~al.(2006)Benjelloun, Sarma, Halevy, and
  Widom]{benjelloun06vldb}
Omar Benjelloun, Anish~Das Sarma, Alon~Y. Halevy, and Jennifer Widom.
\newblock {ULDBs}: Databases with uncertainty and lineage.
\newblock In \emph{VLDB}, pages 953--964, 2006.

\bibitem[Bohannon et~al.(2006)Bohannon, Pierce, and Vaughan]{bohannon06pods}
Aaron Bohannon, Benjamin~C. Pierce, and Jeffrey~A. Vaughan.
\newblock Relational lenses: a language for updatable views.
\newblock In \emph{PODS}, pages 338--347. ACM Press, 2006.

\bibitem[Bohannon et~al.(2008)Bohannon, Foster, Pierce, Pilkiewicz, and
  Schmitt]{DBLP:conf/popl/BohannonFPPS08}
Aaron Bohannon, J.~Nathan Foster, Benjamin~C. Pierce, Alexandre Pilkiewicz, and
  Alan Schmitt.
\newblock Boomerang: resourceful lenses for string data.
\newblock In \emph{POPL}, pages 407--419. ACM, 2008.

\bibitem[Bose and Frew(2005)]{bose05cs}
Rajendra Bose and James Frew.
\newblock Lineage retrieval for scientific data processing: a survey.
\newblock \emph{ACM Comput. Surv.}, 37\penalty0 (1):\penalty0 1--28, 2005.

\bibitem[Buneman et~al.(1994)Buneman, Libkin, Suciu, Tannen, and
  Wong]{DBLP:journals/sigmod/BunemanLSTW94}
Peter Buneman, Leonid Libkin, Dan Suciu, Val Tannen, and Limsoon Wong.
\newblock Comprehension syntax.
\newblock \emph{SIGMOD Record}, 23\penalty0 (1):\penalty0 87--96, 1994.

\bibitem[Buneman et~al.(1995)Buneman, Naqvi, Tannen, and Wong]{buneman95tcs}
Peter Buneman, Shamim~A. Naqvi, Val Tannen, and Limsoon Wong.
\newblock Principles of programming with complex objects and collection types.
\newblock \emph{Theor. Comp. Sci.}, 149\penalty0 (1):\penalty0 3--48, 1995.

\bibitem[Buneman et~al.(2001)Buneman, Khanna, and Tan]{buneman01icdt}
Peter Buneman, Sanjeev Khanna, and {Wang-Chiew} Tan.
\newblock Why and where: A characterization of data provenance.
\newblock In \emph{ICDT}, number 1973 in LNCS, pages 316--330. Springer, 2001.

\bibitem[Buneman et~al.(2006)Buneman, Chapman, and
  Cheney]{DBLP:conf/sigmod/BunemanCC06}
Peter Buneman, Adriane Chapman, and James Cheney.
\newblock Provenance management in curated databases.
\newblock In \emph{SIGMOD}, pages 539--550, 2006.

\bibitem[Buneman et~al.(2007)Buneman, Cheney, and Vansummeren]{buneman07icdt}
Peter Buneman, James Cheney, and Stijn Vansummeren.
\newblock On the expressiveness of implicit provenance in query and update
  languages.
\newblock In \emph{ICDT}, number 4353 in LNCS, pages 209--223. Springer, 2007.

\bibitem[Buneman et~al.(2008)Buneman, Cheney, Tan, and
  Vansummeren]{DBLP:conf/pods/BunemanCTV08}
Peter Buneman, James Cheney, Wang-Chiew Tan, and Stijn Vansummeren.
\newblock Curated databases.
\newblock In \emph{PODS}, pages 1--12, 2008.

\bibitem[Cheney et~al.(2007)Cheney, Ahmed, and Acar]{DBLP:conf/dbpl/CheneyAA07}
James Cheney, Amal Ahmed, and Umut~A. Acar.
\newblock Provenance as dependency analysis.
\newblock In \emph{DBPL}, volume 4797 of \emph{Lecture Notes in Computer
  Science}, pages 138--152. Springer, 2007.

\bibitem[Cui et~al.(2000)Cui, Widom, and Wiener]{DBLP:journals/tods/CuiWW00}
Yingwei Cui, Jennifer Widom, and Janet~L. Wiener.
\newblock Tracing the lineage of view data in a warehousing environment.
\newblock \emph{ACM Trans. Database Syst.}, 25\penalty0 (2):\penalty0 179--227,
  2000.

\bibitem[Dean and Ghemawat(2008)]{mapreduce}
Jeffrey Dean and Sanjay Ghemawat.
\newblock {MapReduce}: simplified data processing on large clusters.
\newblock \emph{Commun. ACM}, 51\penalty0 (1):\penalty0 107--113, 2008.

\bibitem[Field and Tip(1998)]{field98ist}
John Field and Frank Tip.
\newblock Dynamic dependence in term rewriting systems and its application to
  program slicing.
\newblock \emph{Information and Software Technology}, 40\penalty0
  (11--12):\penalty0 609--636, November/December 1998.

\bibitem[Foster et~al.(2007)Foster, Greenwald, Moore, Pierce, and
  Schmitt]{foster07toplas}
J.~Nathan Foster, Michael~B. Greenwald, Jonathan~T. Moore, Benjamin~C. Pierce,
  and Alan Schmitt.
\newblock Combinators for bidirectional tree transformations: A linguistic
  approach to the view-update problem.
\newblock \emph{ACM Trans. Program. Lang. Syst.}, 29\penalty0 (3):\penalty0 17,
  2007.

\bibitem[Foster et~al.(2008)Foster, Green, and
  Tannen]{DBLP:conf/pods/FosterGT08}
J.~Nathan Foster, Todd~J. Green, and Val Tannen.
\newblock Annotated {XML}: queries and provenance.
\newblock In \emph{PODS}, pages 271--280, 2008.

\bibitem[Green et~al.(2007)Green, Karvounarakis, and
  Tannen]{DBLP:conf/pods/2007/GreenKT07}
Todd~J. Green, Gregory Karvounarakis, and Val Tannen.
\newblock Provenance semirings.
\newblock In \emph{PODS}, pages 31--40. ACM, 2007.

\bibitem[Gupta and Mumick(1995)]{gupta95maintenance}
Ashish Gupta and Inderpal~Singh Mumick.
\newblock Maintenance of materialized views: Problems, techniques and
  applications.
\newblock \emph{IEEE Data Engineering Bulletin}, 18\penalty0 (2):\penalty0
  3--18, 1995.

\bibitem[Hidders et~al.(2007)Hidders, Kwasnikowska, Sroka, Tyszkiewicz, and den
  Bussche]{DBLP:conf/dils/HiddersKSTB07}
Jan Hidders, Natalia Kwasnikowska, Jacek Sroka, Jerzy Tyszkiewicz, and Jan~Van
  den Bussche.
\newblock A formal model of dataflow repositories.
\newblock In \emph{DILS}, volume 4544 of \emph{LNCS}, pages 105--121. Springer,
  2007.

\bibitem[Matsuda et~al.(2007)Matsuda, Hu, Nakano, Hamana, and
  Takeichi]{matsuda07icfp}
Kazutaka Matsuda, Zhenjiang Hu, Keisuke Nakano, Makoto Hamana, and Masato
  Takeichi.
\newblock Bidirectionalization transformation based on automatic derivation of
  view complement functions.
\newblock In \emph{ICFP '07: Proceedings of the 12th ACM SIGPLAN international
  conference on Functional programming}, pages 47--58, New York, NY, USA, 2007.
  ACM.
\newblock ISBN 978-1-59593-815-2.
\newblock \doi{http://doi.acm.org/10.1145/1291151.1291162}.

\bibitem[Muniswamy-Reddy et~al.(2006)Muniswamy-Reddy, Holland, Braun, and
  Seltzer]{muniswamy-reddy06usenix}
Kiran-Kumar Muniswamy-Reddy, David~A. Holland, Uri Braun, and Margo Seltzer.
\newblock Provenance-aware storage systems.
\newblock In \emph{USENIX Annual Technical Conference}, pages 43--56. USENIX,
  June 2006.

\bibitem[Ochoa et~al.(2004)Ochoa, Silva, and Vidal]{ochoa04pepm}
Claudio Ochoa, Josep Silva, and Germ\'an Vidal.
\newblock Dynamic slicing based on redex trails.
\newblock In \emph{PEPM}, pages 123--134. ACM Press, 2004.

\bibitem[Olston et~al.(2008)Olston, Reed, Srivastava, Kumar, and
  Tomkins]{piglatin}
Christopher Olston, Benjamin Reed, Utkarsh Srivastava, Ravi Kumar, and Andrew
  Tomkins.
\newblock Pig latin: a not-so-foreign language for data processing.
\newblock In \emph{SIGMOD}, pages 1099--1110, New York, NY, USA, 2008. ACM.

\bibitem[Rival and Mauborgne(2007)]{rival07toplas}
Xavier Rival and Laurent Mauborgne.
\newblock The trace partitioning abstract domain.
\newblock \emph{ACM Trans. Program. Lang. Syst.}, 29\penalty0 (5):\penalty0 26,
  2007.

\bibitem[Simmhan et~al.(2005)Simmhan, Plale, and
  Gannon]{DBLP:journals/sigmod/SimmhanPG05}
Yogesh Simmhan, Beth Plale, and Dennis Gannon.
\newblock A survey of data provenance in e-science.
\newblock \emph{SIGMOD Record}, 34\penalty0 (3):\penalty0 31--36, 2005.

\bibitem[Venkatesh(1991)]{venkatesh91pldi}
G.~A. Venkatesh.
\newblock The semantic approach to program slicing.
\newblock In \emph{PLDI}, pages 107--119. ACM Press, 1991.

\bibitem[Wadler(1992)]{wadler92mscs}
P.~Wadler.
\newblock Comprehending monads.
\newblock \emph{Mathematical Structures in Computer Science}, 2:\penalty0
  461--493, 1992.

\bibitem[Wang and Madnick(1990)]{DBLP:conf/vldb/WangM90}
Y.~Richard Wang and Stuart~E. Madnick.
\newblock A polygen model for heterogeneous database systems: The source
  tagging perspective.
\newblock In \emph{VLDB}, pages 519--538, 1990.

\end{thebibliography}
